%% file: main.tex
\documentclass[sigconf,edbt]{acmart-edbt2021}

\usepackage{graphicx}
\usepackage{textcomp}
\usepackage{xcolor}
\usepackage{amsmath,amssymb}
\usepackage{amsthm} 
\usepackage[noend]{algpseudocode}
\usepackage{algorithm}
\usepackage{subfig}
\usepackage{url}
\usepackage{xcolor}

\usepackage[english]{babel}

\newtheorem{problem}{Problem}

\newcommand{\RNum}[1]{\uppercase\expandafter{\romannumeral #1\relax}}

\newcommand{\stitle}[1]{\vspace*{0.5em}\noindent{\bf \em #1\/} $\\$ \indent}
\newcommand{\polyfit}[1]{\textsf{PolyFit}}

\usepackage{enumitem} 

\newcommand*{\affaddr}[1]{\nolinkurl{#1}}

\acmConference[EDBT2021]{24th International Conference on Extending Database Technology (EDBT)}{Nicosia}{Cyprus}
\acmYear{2021}

\begin{document}
\title{PolyFit: Polynomial-based Indexing Approach for Fast Approximate Range Aggregate Queries} \titlenote{This work was supported by grant GRF 152050/19E from the Hong Kong RGC.}

\author{Zhe Li$^{\text{1}}$, Tsz Nam Chan$^{\text{2}}$, Man Lung Yiu$^{\text{1}}$, Christian S. Jensen$^{\text{3}}$}
\affiliation{%
  \institution{Hong Kong Polytechnic University$^{\text{1}}$, Hong Kong Baptist University$^{\text{2}}$, Aalborg University$^{\text{3}}$}
  \nolinkurl{richie.li@connect.polyu.hk, edisonchan@comp.hkbu.edu.hk, csmlyiu@comp.polyu.edu.hk, csj@cs.aau.dk}
}

\begin{abstract}
Range aggregate queries find frequent application in data analytics. In many use cases, approximate results are preferred over accurate results if they can be computed rapidly and satisfy approximation guarantees. Inspired by a recent indexing approach, we provide means of representing a discrete point dataset by continuous functions that can then serve as compact index structures. More specifically, we develop a polynomial-based indexing approach, called \polyfit{}, for processing approximate range aggregate queries. \polyfit{} is capable of supporting multiple types of range aggregate queries, including COUNT, SUM, MIN and MAX aggregates, with guaranteed absolute and relative error bounds. Experimental results show that \polyfit{} is faster and more accurate and compact than existing learned index structures.
\end{abstract}

\maketitle

\input{intro}
\input{related_work}
\input{prelim}

\input{index}
\input{query}
\input{case_2d_v2}
\input{experiment}
\input{conclusion}


\bibliographystyle{abbrv}
\bibliography{ref}

\input{appendix}
\end{document}

%% file: intro.tex
\vspace{-3.5mm}
\section{Introduction} \label{sec:intro}
\vspace{-1.0mm}
A {\em range aggregate query}~\cite{ho1997range} retrieves records in a dataset that belong to a given key range and then applies an aggregate function (e.g., \verb"SUM", \verb"COUNT", \verb"MIN", \verb"MAX") to an attribute of those records.
Range aggregate queries are used in OLAP~\cite{vitter1999approximate,ho1997range}
and data analytics applications, e.g., for outlier detection \cite{XGGKS15,ES13}, data visualization \cite{AMSAKS15}, and tweet analysis \cite{AMODSR11}.
For example,  network intrusion detection systems \cite{XGGKS15} utilize range \verb"COUNT" queries to monitor a network for anomalous activities. Furthermore, applications with huge numbers of users are expected to receive queries frequently. For instance, Foursquare, with more than 50 million monthly active users~\cite{Foursquareusers}, helps users find the number of specific POIs (e.g., restaurants) within given regions \cite{FoursquareAPI}.
In many application scenarios, users accept approximate results provided that
(i) they can be computed quickly and
(ii) they are sufficiently accurate (e.g., within 5\% error).
We target such applications and focus on
error-bounded evaluation of range aggregate queries.



A recent indexing approach represents the values of attributes in a dataset by continuous functions, which then serve to enable compact index structures~\cite{kraska2018case,fiting2019}.
When compared to traditional index structures,
this approach is able to yield a smaller index size and faster response time.
The existing studies~\cite{kraska2018case,fiting2019} focus on computing exact results
for point and range queries on 1-dimensional data. 
In contrast, we conduct a {\bf comprehensive study of approximate range aggregate queries,
supporting many aggregate functions and multi-dimensional data}.

The idea that underlies our proposal for using functions to answer approximate range aggregate queries may be explained as follows.
Consider a stock market index (e.g., the Hong Kong Hang Seng Index) at different times as a dataset $\mathcal{D}$ consisting of records of the form (index value, timestamp), where the former is our measure and the latter is our key that is used for specifying query ranges---see Figure~\ref{fig:example1}(a). A user can find the average stock market index value in a specified time range $[l_q,u_q]$ by issuing a range \verb"SUM" query (and divide by $u_q - l_q + 1$). We propose to construct the cumulative function of $\mathcal{D}$ as shown in Figure~\ref{fig:example1}(b).
If we can approximate this function well
by a polynomial function $\mathbb{P}(x)$ then
the range \verb"SUM" query can be approximated as
$\mathbb{P}(u_q)-\mathbb{P}(l_q)$, which takes $O(1)$ time.
As another example, the user may wish to find the maximum stock market index in a specified time range.
The timestamped index values in $\mathcal{D}$ can be modeled by the continuous function shown in Figure~\ref{fig:example1}(c).
Again, if we can approximate this function well
using a polynomial function $\mathbb{P}(x)$ then the range \verb"MAX" query
can be answered quickly using mathematical tools, e.g., by applying differentiation to identify maxima in $\mathbb{P}(x)$.



Regarding the two-dimensional case, consider the dataset of tweets' locations
as shown in Figure~\ref{fig:example2D}(a) in Section~\ref{sec:extend_2d},
where each data point has a longitude (as key 1) and a latitude (as key 2).
Suppose that the user wishes to count the number of tweets in a geographical region.
Our idea is to derive the cumulative count function shown in Figure~\ref{fig:example2D}(b),
and then approximate this function with a polynomial function $\mathbb{P}(x_1,x_2)$ (of two variables).
This enables us to answer a two-dimensional range \verb"COUNT" query in $O(1)$ time.


\begin{figure*}[!hbt]
\center
\subfloat[timestamped index values]{ \includegraphics[width=0.65\columnwidth]{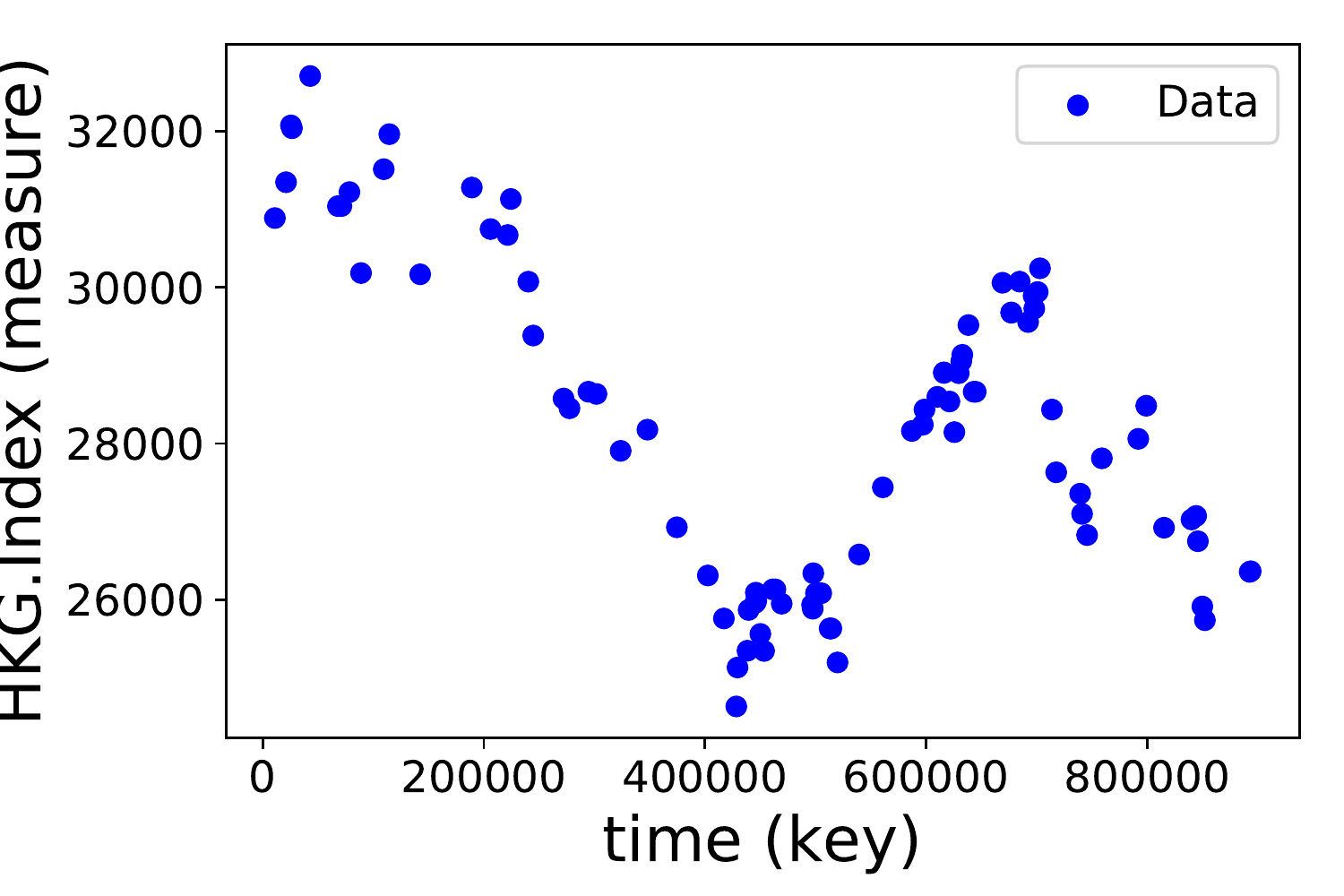}  }
\subfloat[function for range SUM queries]{ \includegraphics[width=0.7\columnwidth]{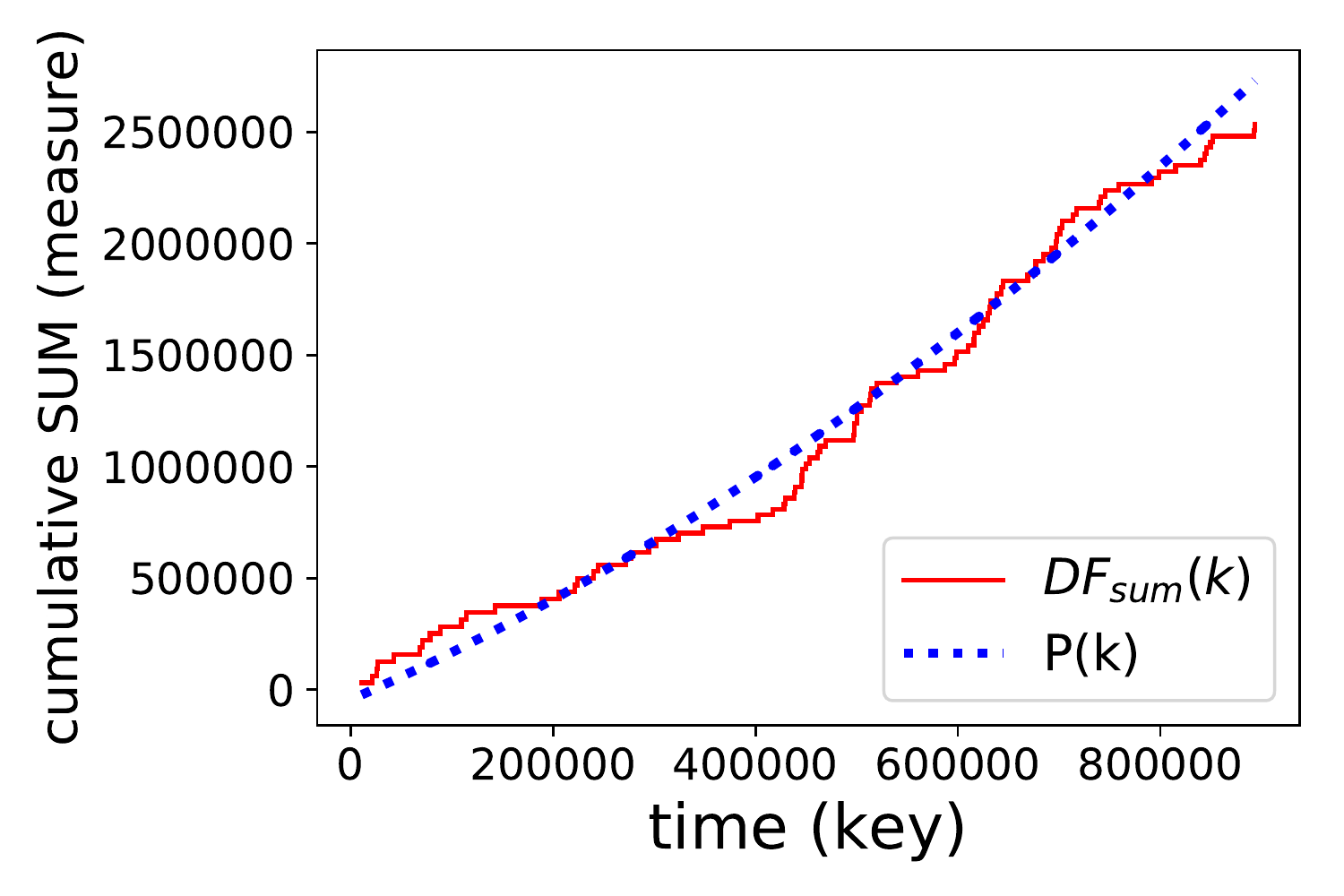} }
\subfloat[function for range MAX queries]{ \includegraphics[width=0.65\columnwidth]{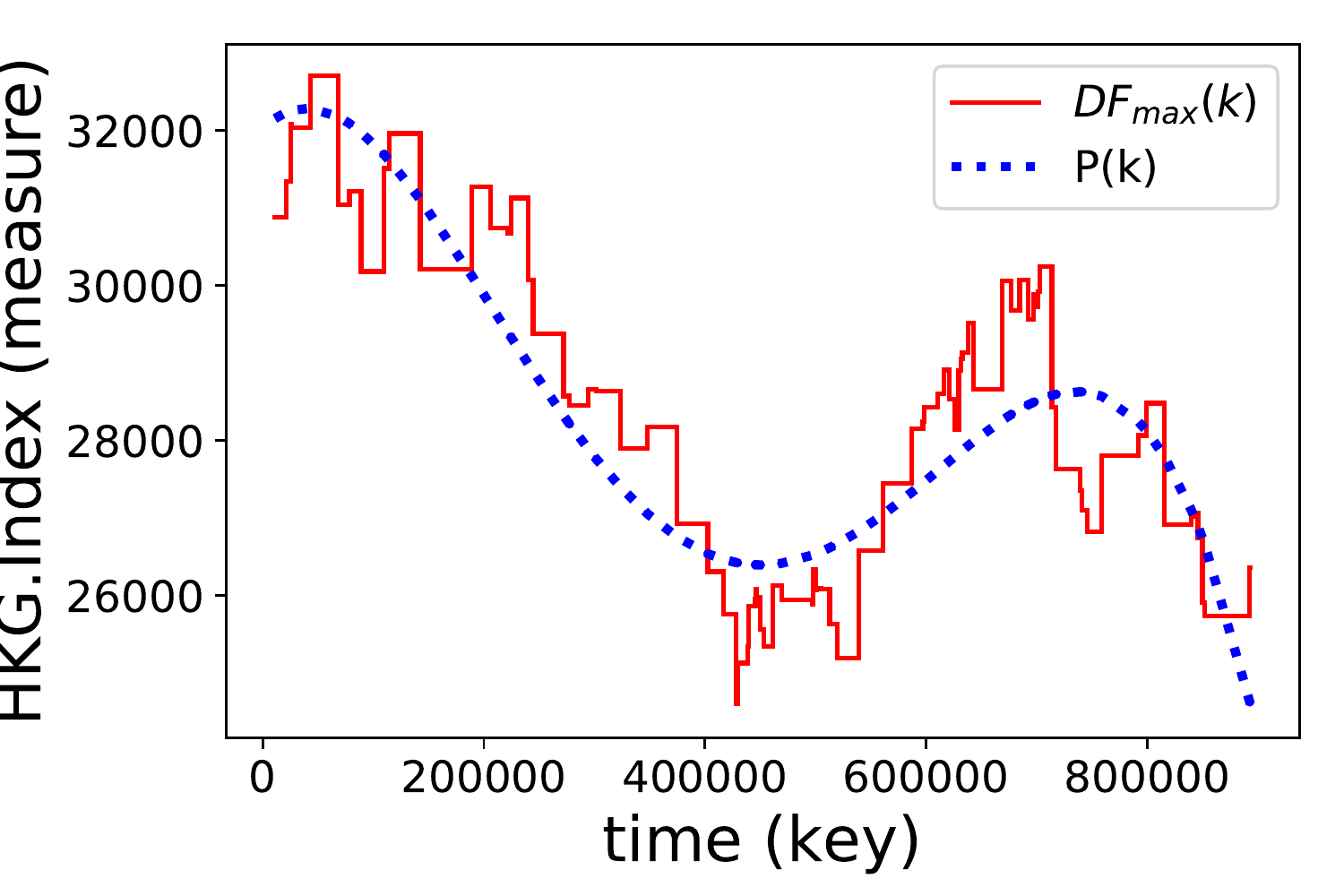} }
\vspace{-2mm}
\caption{Stock market index values, 1-dimensional keys: discrete data points vs. continuous function}
\label{fig:example1}
\vspace{-5mm}
\end{figure*}

Another difference between our work and existing studies \cite{kraska2018case,fiting2019}
is the types of functions used for approximation.
Our proposal uses piecewise polynomial functions,
rather than piecewise linear functions \cite{kraska2018case,fiting2019}.
As we will show in Section~\ref{sec:index_construction},
using polynomial functions yields lower fitting errors than using linear functions.
Thus, our proposal leads to smaller index sizes and faster queries.




The key technical challenges are as follows.
(1) How to find polynomial functions with low approximation error efficiently?
(2) How to answer range aggregate queries with error guarantees?
(3) How to support common aggregate functions (e.g., \verb"COUNT", \verb"SUM", \verb"MIN", \verb"MAX") and multi-dimensional data?


To tackle these challenges, we develop a polynomial-based indexing approach (\polyfit{}) for
processing approximate range aggregate queries.
Our contributions are summarized as follows.
\begin{itemize}
	\item To the best of our knowledge, this is the first study that utilizes polynomial functions to learn indexes that support approximate range aggregate queries.
	\item \polyfit{} supports multiple types of range aggregate queries, including \verb"COUNT", \verb"SUM", \verb"MIN" and \verb"MAX" with guaranteed deterministic absolute and relative error bounds.
	\item Experiment results show that \polyfit{} achieves significant speedups, compared with the closest related works~\cite{fiting2019,kraska2018case}, and traditional exact/approximate methods. For instance, for the OpenStreetMap dataset with 100M records, our index occupies only 4 MBytes and offers 5 $\mu s$ query response time (per 2-dimensional range \verb"COUNT" query).
\end{itemize}


The rest of the paper is organized as follows. We first review the related work in Section \ref{sec:related}. Next, we introduce preliminaries in Section~\ref{sec:prelim}. Then, we present our index construction techniques in Section~\ref{sec:index_construction} and
cover how to answer approximate range aggregate queries in~Section \ref{sec:query}. Next, we extend our proposal to datasets with two keys in Section~\ref{sec:extend_2d}. Lastly, we present experiments in Section~\ref{sec:exp} and conclude in Section~\ref{sec:conclusion}. 

%% file: related_work.tex
\vspace{-3mm}
\section{Related Work}
\label{sec:related}
\vspace{-1mm}
Range aggregate queries are used frequently in analytics applications and constitute important functionality in OLAP and data warehousing \cite{momjian2001postgresql,delaney2000inside,bartholomew2014mariadb,condie2010online,armbrust2015spark,jensen2018modelardb,vitter1999approximate,ho1997range}. Exact solutions are based on prefix-sum arrays \cite{ho1997range} or aggregate R-trees \cite{DPJY01}.
Due to the need for real-time performance in some applications (e.g., $\mu s$-level response time \cite{XGGKS15}), many proposals exist that aim to improve the efficiency of range aggregate queries. These proposals can be classified as being either data-driven or query-driven. In addition, we also review some other studies, including learned indexes, and time series databases, which are also related to this work.

\label{sec:data_aware_rel}
{\bf Data-driven proposals} build statistical models of a dataset for
estimating query selectivity or the results of range aggregate queries. These models employ
multi-dimensional histograms~\cite{lynch1988selectivity,muralikrishna1988equi,ilyas2004cords,to2013entropy},
data sampling~\cite{lipton1990practical,haas1994relative,riondato2011vc,BlinkDB13,VerdictDB18,han2019iterative}, or kernel density estimation~\cite{gunopulos2000approximating,gunopulos2005selectivity,heimel2015self}.
Although such proposals that compute approximate results are much faster than exact solutions, e.g., achieving ms ($10^{-3}$) level response time \cite{park2018quicksel},
they still do not offer real-time performance (e.g., $\mu s$ level \cite{XGGKS15}). Furthermore, these proposals do not offer theoretical approximation error guarantees.

\label{sec:query_aware_rel}
The {\bf query-driven approaches} utilize query workloads to build statistical models of datasets. Typical methods include error-feedback histograms \cite{aboulnaga1999self,lim2003sash,anagnostopoulos2015learning}, max-entropy histograms \cite{markl2007consistent,Re:2012:UCE:2109196.2109202}, and learning-based models \cite{QP19, PeterAggregate2019}. In addition, Park et al. \cite{park2018quicksel} explore the approach of using mixture probabilistic models. These methods assume that new queries follow historical query workload distributions. However, as one study \cite{BJCGF04} observes, this assumption may not always hold in practice. Further, even when this assumption is valid, the number of queries that are similar to those used for training may be much smaller if the queries follow a power law distribution \cite{yu2018sundial}, which can cause poor accuracy and may render it impossible to obtain useful approximation error guarantees for range aggregate queries.

\label{sec:learned_index_rel}
Recently, {\bf learning-based methods} have been used to construct more compact and effective index structure, that hold potential to accelerate database operations. Kraska et al. \cite{kraska2018case} propose the RMI index, which incorporates different machine learning models, e.g., linear regression and deep-learning, to improve the efficiency of range queries. Galakatos et al. \cite{fiting2019} develop the FITing-tree, which is a segment-tree-like structure \cite{SegmentTreeBook,AXNS17} that can significantly improve the efficiency of exact point queries. Ferragina et al. \cite{PG20} further support efficient update operations for range queries. Wang et al. \cite{WangFX019} extend this learning-based approach to the spatial domain with their learned Z-order model that aims to support fast spatial indexing. However, there are two main differences between these proposals and our proposal. First, they either support range queries \cite{kraska2018case,WangFX019,PG20} or point queries \cite{fiting2019}, but not range aggregate queries. Second, we are the first to exploit polynomial functions to build index structures for approximate range aggregate queries.
\label{sec:time_series_rel}
In the {\bf time series database} community, some research studies utilize mathematical models to approximate time series data. Representative approaches include piecewise linear approximation \cite{Plato20,Online2001,PLAKDD98,FastSimSearchKeogh97,HazemSwingSlide2009}, discrete wavelet transform \cite{WavletICDE02, WavletICDE99}, discrete Fourier transform \cite{DFTSIGMOD94,DFTICDE99}, and their combinations \cite{EamonnTKDE08, jensen2018modelardb}. However, these studies focus on either time series similarity search (e.g., range or nearest neighbor queries) or time series compression and they are not designed to answer the range aggregate queries we target. Some of these studies also utilize piecewise linear approximation \cite{Plato20,Online2001,HazemSwingSlide2009,FastSimSearchKeogh97} to approximate time-series, which we also do. In contrast, we achieve better performance by utilizing nonlinear (polynomial) functions to approximate curves, which can reduce the number of segments dramatically. Furthermore, we can also support the segmentation of surfaces (e.g., Figure \ref{fig:example2D}(b)), rather than only 1-D curves.



%% file: prelim.tex
\vspace{-3mm}
\section{Preliminaries}
\vspace{-1mm}
\label{sec:prelim}
First, we define range aggregate queries and their approximate versions in Section~\ref{sec:background}. Then, we discuss the baselines for answering exact range aggregate queries in Section~\ref{sec:exact}. Table~\ref{tab:symbol} summarizes frequently used symbols in this paper. 

\begin{table}[htb]
\vspace{-1mm}
\small \center
\caption{Symbols} \label{tab:conditions}
\vspace{-3mm}
\label{tab:symbol}
\begin{tabular}{| c | c |}
    \hline
    Symbol & Description\\ \hline \hline
    $\mathcal{D}$ & dataset \\ \hline
    $n$ & number of records in $\mathcal{D}$ \\ \hline
    $R_{count}$ & range \verb"COUNT" query \\ \hline
    $R_{sum}$ & range \verb"SUM" query \\ \hline
    $R_{min}$ & range \verb"MIN" query \\ \hline
    $R_{max}$ & range \verb"MAX" query \\ \hline
    $CF_{sum}$ & cumulative function for range \verb"SUM" query \\ \hline
    $DF_{max}$ & key-measure function for range \verb"MAX" query \\ \hline
    $\mathbb{P}(k)$ & polynomial function \\ \hline
    $I$ & interval \\ \hline
    $deg$ & degree of polynomial function \\ \hline
\end{tabular}
\vspace{-3mm}
\end{table}

\vspace{-3mm}
\subsection{Problem Definition}
\label{sec:background}
\vspace{-1mm}
%
We focus on the setting that
a range aggregate query specifies a $key$ attribute (for range selection)
and a $measure$ attribute for aggregation.
We shall consider the setting of two keys in Section~\ref{sec:extend_2d}.
As such, the dataset $\mathcal{D}$ is a set of $(key, measure)$ records,
i.e., $\mathcal{D}=\{(k_1, m_1), (k_2, m_2),..., (k_n, m_n)\}$.
For ease of discussion, we assume that key values are distinct and measure values are numerical. 
We leave the discussion of repeated keys and negative measure values in Appendices A.3 and A.4 \cite{ZTMC20_arxiv}. Then we define a range aggregate query as follows.
\begin{definition} \label{def:RAQ}
Let $\mathcal{G}$ be an aggregate function (e.g.,  \verb"COUNT", \verb"SUM", \verb"MIN", \verb"MAX")
on a measure attribute.
Given a dataset $\mathcal{D}$ and a key range $[l_q, u_q]$,
we define $V$ as the following multi-set
$$V= \{ m\mid (k,m) \in \mathcal{D} \wedge l_q \leq k \leq u_q \}$$
and then define the result of the range aggregate query as
\begin{equation}
\label{eq:query}
R_{\mathcal{G}}(\mathcal{D},[l_q, u_q]) = \mathcal{G}(V).
\end{equation}
\end{definition}
%

We aim to develop efficient methods for obtaining an approximate result of $R_{\mathcal{G}}(\mathcal{D},[l_q, u_q])$
with two types of error guarantees~\cite{MP02,MP04}, namely the absolute error guarantee (cf. Problem \ref{prob:abs_error}) and the relative error guarantee (cf. Problem \ref{prob:rel_error}).

\begin{problem}[$Q_{abs}$]
\label{prob:abs_error}
Given an absolute error $\varepsilon_{abs}$ and a range aggregate query, we ask for an approximate result $A_{abs}$ such that:
\begin{equation}
\label{eq:abs_error}
|A_{abs}-R_{\mathcal{G}}(\mathcal{D},[l_q, u_q])| \leq \varepsilon_{abs}
\end{equation}
\end{problem}

\begin{problem}[$Q_{rel}$]
\label{prob:rel_error}
Given a relative error $\varepsilon_{rel}$ and a range aggregate query, we ask for an approximate result $A_{rel}$ such that:
\begin{equation}
\label{eq:rel_error}
\left| \frac{A_{rel}-R_{\mathcal{G}}(\mathcal{D},[l_q, u_q])}{R_{\mathcal{G}}(\mathcal{D},[l_q, u_q])} \right| \leq \varepsilon_{rel}
\end{equation}
\end{problem}

\subsection{Baselines: Exact Methods}
\label{sec:exact}
\vspace{-0.5mm}
We proceed to discuss exact methods for answering
range \texttt{SUM} queries and range \texttt{MAX} queries.
These methods can be easily extended to support \texttt{COUNT} and \texttt{MIN}, respectively.

\subsubsection{{\bf Exact method for range \texttt{SUM} queries}}
\label{sec:exact_count}
First, we define the key cumulative function as $CF_{sum}(k)$:
\vspace{-1mm}
\begin{equation}
CF_{sum}(k)=R_{sum}(\mathcal{D},[-\infty,k]).
\end{equation}

The additive property of $CF_{sum}$ enables us
to compute the exact result of the range \texttt{SUM} query as:
\vspace{-1mm}
\begin{equation}
\label{eq:R_c_cum}
R_{sum}(\mathcal{D},[l_q, u_q])=CF_{sum}(u_q)- CF_{sum}(l_q).
\end{equation}

Then, we discuss how to obtain the terms
$CF_{sum}(l_q)$ and $CF_{sum}(u_q)$ efficiently.
Although $CF_{sum}$ is a continuous function, it can be expressed
by a discrete data structure in finite space.
Specifically, we presort dataset $\mathcal{D}$ in ascending key order and
then follow this order to construct a key-cumulative array
of entries $(k,CF_{sum}(k))$. 
At query time, the terms $CF_{sum}(l_q)$ and $CF_{sum}(u_q)$
are obtained by performing binary search on the above key-cumulative array. This step takes $O(\log n)$ time.


As a remark, this key-cumulative array is similar to the prefix-sum array~\cite{ho1997range}.
The difference is that our array allows floating-point search keys,
while the prefix-sum array does not.





\subsubsection{{\bf Exact method for range \texttt{MAX} queries}}
\label{sec:exact_max}
First, we define the key-measure function $DF_{max}(k)$ in Equation \ref{eq:DF} to capture the data distribution in the dataset $\mathcal{D}$.
In the definition, we assume that each pair $(k_i, m_i)$ in $\mathcal{D}$ is arranged in ascending order by the key.

\vspace{-3.0mm}
\begin{eqnarray}
\label{eq:DF}
   DF_{max}(k) =
   \begin{cases}
   	 m_1 & \mbox{ if } k_1 \leq k < k_{2} \\
   	 \vdots &\mbox{ } \vdots \\
   	 m_i &\mbox{ if } k_i \leq k < k_{i+1} \\
   	 \vdots &\mbox{ } \vdots \\
   	 m_n & \mbox{ if } k = k_n \\
   	 -\infty &\mbox{ otherwise} \\
   \end{cases}
\end{eqnarray}
Figure~\ref{fig:max_index}(a) exemplifies the function $DF_{max}(k)$.


An aggregate max-tree~\cite{DPJY01} (cf. Figure~\ref{fig:max_index}(b)) can be
built to answer range \texttt{MAX} queries.
In this tree, each internal node stores two entries, where each entry stores an interval and the maximum measure within that interval (e.g., $(I_1,m_6)$ and $(I_2,m_7)$ are two entries of the root node $N_{root}$). We then explain how to process the query $R_{max}(\mathcal{D},[l_q, u_q])$, whose query range is indicated by the red line in Figure~\ref{fig:max_index}(a). In Figure~\ref{fig:max_index}(b), we start from the root of the tree. If the interval of an entry intersects with the query range (e.g., $I_1$ and $I_2$ in Figure~\ref{fig:max_index}(a)), we visit its child nodes (e.g., $N_1$  and $N_2$). When the interval of an entry (e.g., $I_4$ and $I_5$ in Figure~\ref{fig:max_index}a) is covered by the query range, we directly use its stored aggregate value without visiting its child nodes (e.g., yellow nodes in Figure~\ref{fig:max_index}b). During the traversal, we keep track of the maximum measure seen so far. This procedure takes $O(\log n)$ time as we check at most two branches per level.

\begin{figure}[!hbt]
\vspace{-3mm}
\center
\includegraphics[width=1.05\columnwidth,page=2]{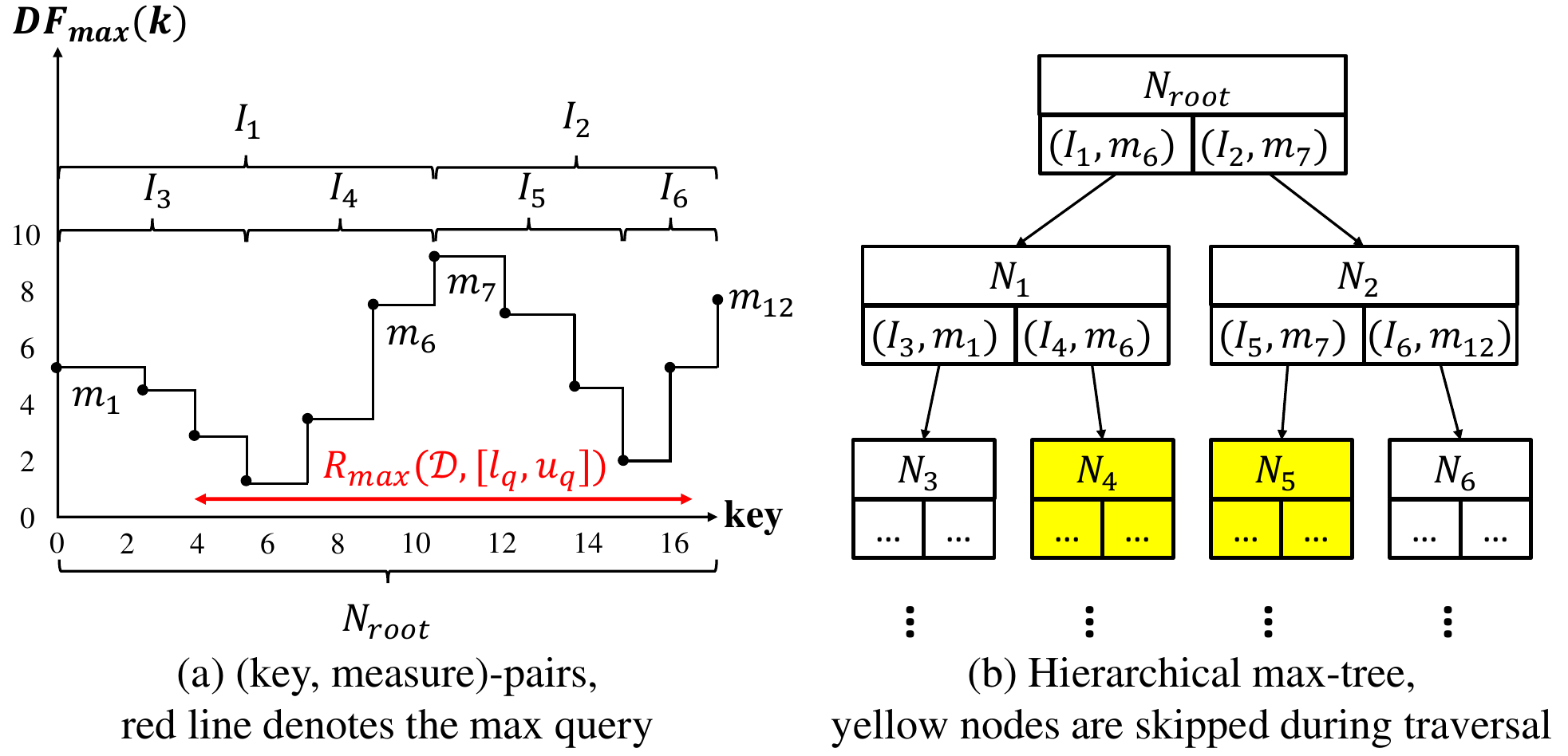}
\vspace{-4mm}
\caption{Aggregate \texttt{MAX} tree}
\label{fig:max_index}
\vspace{-3mm}
\end{figure}

%% file: index.tex
\section{Index Construction}
\label{sec:index_construction}
Traditional index structures (e.g., B-tree \cite{TCRC01}) need to store $n$ keys, where $n$ is the cardinality of the dataset $\mathcal{D}$.
Thus, the index size grows linearly with the data size.
To reduce the index size dramatically, we plan to index a limited number of functions (instead of $n$ keys).


As a case study, we compare existing fitting functions~\cite{kraska2018case,fiting2019}
with our fitting function (polynomial)
on a real dataset (the Hong Kong 40 Index in 2018 \cite{hk40index}) in Figure \ref{fig:idea}.
The exact key-measure function $DF_{max}(k)$ exhibits a complex shape.
Observe that linear functions, e.g., linear regression $LR(k)$ \cite{kraska2018case} and linear segment $FIT(k)$ \cite{fiting2019},
cannot accurately approximate the exact function.
In this paper, we adopt the polynomial function $\mathbb{P}(k)$, which captures the nonlinear property\footnote{As a remark, other types of nonlinear functions (e.g., logarithmic and trigonometric functions) require higher computation cost than polynomial functions. Thus,
we leave other types of nonlinear functions as future work.} and
achieves a better approximation of $DF_{\max}(k)$. In this example, $\mathbb{P}(k)$ is a degree-4 polynomial function (blue dotted line).


\begin{figure}[!hbt]
\vspace{-2mm}
\center
\includegraphics[width=0.8\columnwidth]{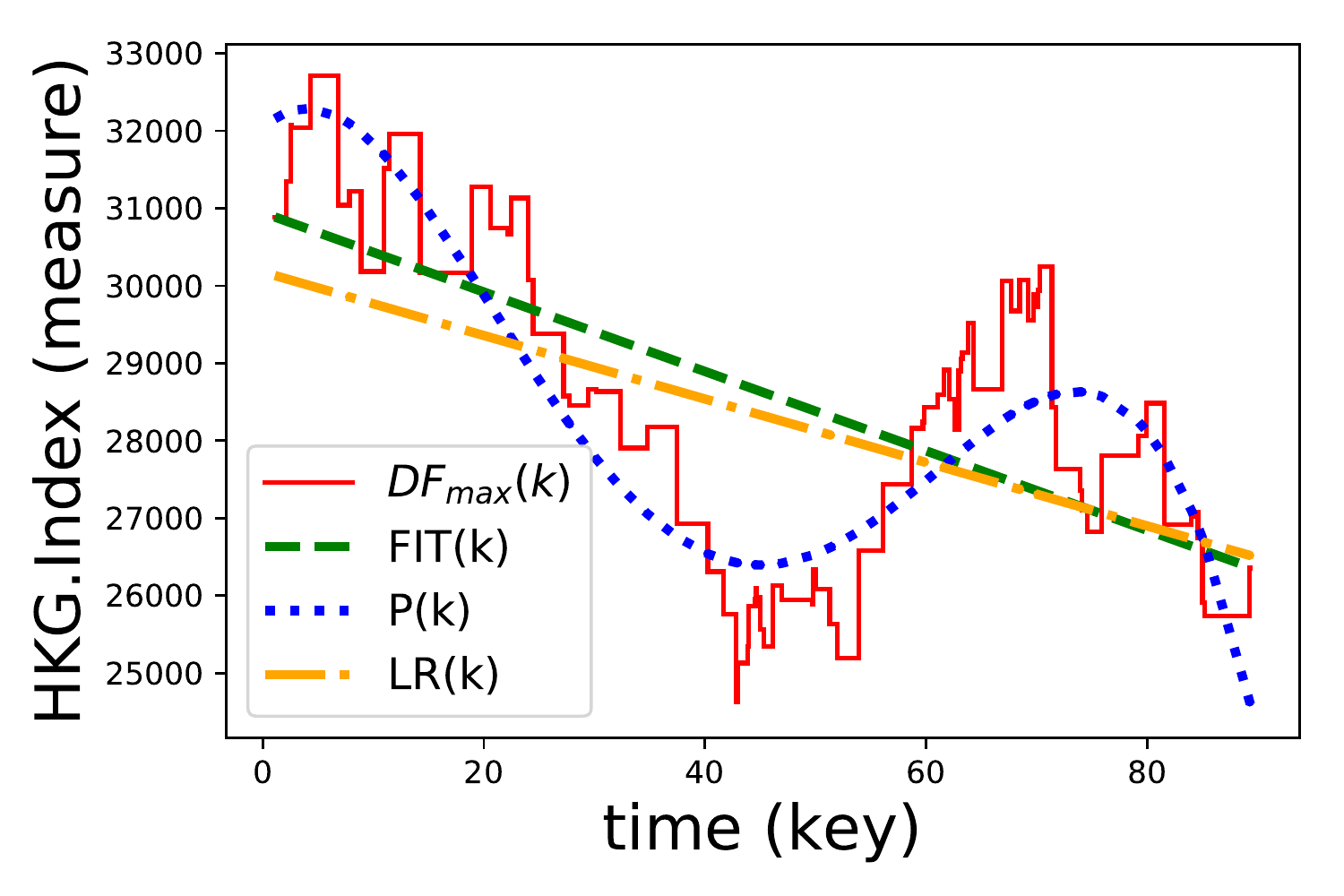}
\vspace{-2mm}
\caption{Curve fitting of the HKG 40 Index in 2018 \cite{hk40index}}
\label{fig:idea}
\vspace{-2mm}
\end{figure}

\begin{figure*}[hbt]
\vspace{-2mm}
\center
\includegraphics[width=2.0\columnwidth]{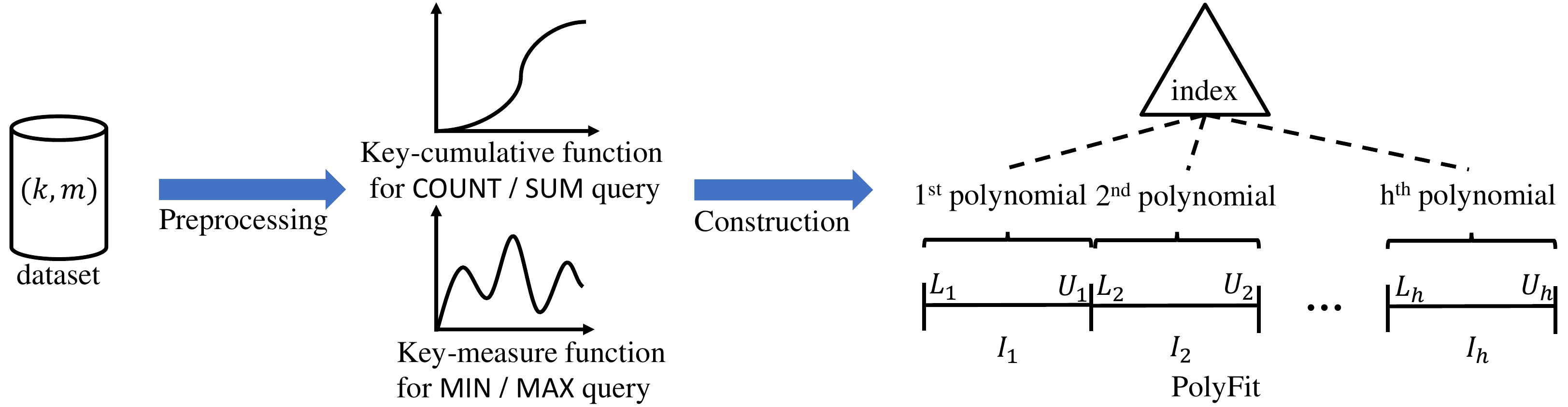}
\vspace{-3mm}
\caption{Indexing framework for \polyfit{}, each leaf entry stores a polynomial function}
\label{fig:indexframework}
\vspace{-4mm}
\end{figure*}

We introduce our indexing framework in Figure \ref{fig:indexframework}.
First, we convert the dataset into the following exact function $F(k)$ based on
the aggregate function $\mathcal{G}$ and the functions in Section~\ref{sec:exact}.
\begin{eqnarray}
\label{eq:Fk}
   F(k) =
   \begin{cases}
       CF_{sum}(k) & \text{if } \mathcal{G} = \texttt{SUM}\\
       DF_{max}(k) & \text{if } \mathcal{G} = \texttt{MAX}
   \end{cases}
\end{eqnarray}
We plan to compute an error-bounded approximation of $F(k)$
by using a sequence of polynomial functions.
In Section~\ref{sec:min_error}, we examine how to find the
best polynomial fitting of $F(k)$ in a given key interval $I$.
Then, in Section~\ref{sec:optimal}, we propose a segmentation method
for $F(k)$ in order to minimize the index size subject to a given deviation threshold.
Finally, in Section~\ref{sec:organizing}, we discuss how to
build an index for a sequence of polynomial functions.


\subsection{Polynomial Fitting in a Key Interval} \label{sec:min_error}
%
We discuss how to find the best fitting polynomial function of $F(k)$
in a given key interval $I$.
First, we express a polynomial function $\mathbb{P}(k)$ as follows:
\begin{equation}
\mathbb{P}(k)=\sum_{j=0}^{deg} a_j k^j,
\end{equation}
where $deg$ is the degree and each $a_j$ is a coefficient.
Note that the choice of $deg$ entails tradeoffs between the fitting error and the online query evaluation cost.
We discuss the choice of $deg$ in Section \ref{sec:choose_degree}.


We formulate the following optimization problem in order to minimize the fitting error
between $\mathbb{P}(k)$ and $F(k)$.
\begin{definition}
\label{def:min_error_problem}
Let $F(k)$ be the exact function and $I$ be a given key interval.
Let $k_1, k_2, \cdots , k_\ell$ be the keys of $\mathcal{D}$ in interval $I$.
We aim to find polynomial coefficients, $a_0, a_1, \cdots, a_{deg}$ that minimize the following error:
\begin{equation}
\label{eq:error}
E(I)=\min_{a_0,a_1,..., a_{deg} \in \mathbb{R}} \; \max_{1\leq i\leq \ell} |F(k_i)-\mathbb{P}(k_i)|
\end{equation}
\end{definition}
This is equivalent to the following linear programming problem, where the coefficients $a_0, a_1, \cdots, a_{deg}$ and $t$ are variables.
\begin{equation}
\label{eq:LP}
\begin{cases}
& \textsc{minimize} \quad t \\
& \textsc{subject \: to:} \\
& -t \leq F(k_1) - (a_{deg} k_1^{deg}+...+a_2 k_1^2+a_1 k_1+a_0) \leq t \\
& -t \leq F(k_2) - (a_{deg} k_2^{deg}+...+a_2 k_2^2+a_1 k_2+a_0) \leq t \\
&... \\
& -t \leq F(k_\ell) - (a_{deg} k_\ell^{deg}+...+a_2 k_\ell^2+a_1 k_\ell+a_0) \leq t \\
& \forall a_i \in \mathbb{R} \\
\end{cases}
\end{equation}

It takes $O(\ell^{2.5})$ time to solve the above linear programming problem (Equation~\ref{eq:LP})~\cite{lee2015efficient}. In our experimental study, we adopt the IBM CPLEX linear programming library as the LP Solver, which is believed to be the most reliable and efficient among other implementations~\cite{LPSolverComparison}. We discuss some subtle issues like precision limitations in Section~\ref{sec:choose_degree}.

\vspace{-2mm}
\subsection{Minimal Index Size with Bounded Error} \label{sec:optimal}
\vspace{-1mm}
To support approximate query evaluation (in Section~\ref{sec:query}),
we require that the fitting polynomial functions should satisfy a given error constraint.
However, a single polynomial function is unlikely to fit accurately for the entire key domain.
Thus, we propose to partition the key domain into intervals $I_1, I_2, \cdots, I_h$ so that
each interval $I_i$ satisfies the following requirement:
$$E(I)\leq \delta ,$$
where $\delta$ is a given deviation threshold.
For instance, in Figure~\ref{fig:bounded_error_fitting},
the key domain is partitioned into two intervals $I_1$ and $I_2$
so that the best fitting polynomial function in each interval satisfies the error requirement.


\begin{figure}[!hbt]
	\vspace{-1mm}
	\center
	\includegraphics[width=0.8\columnwidth]{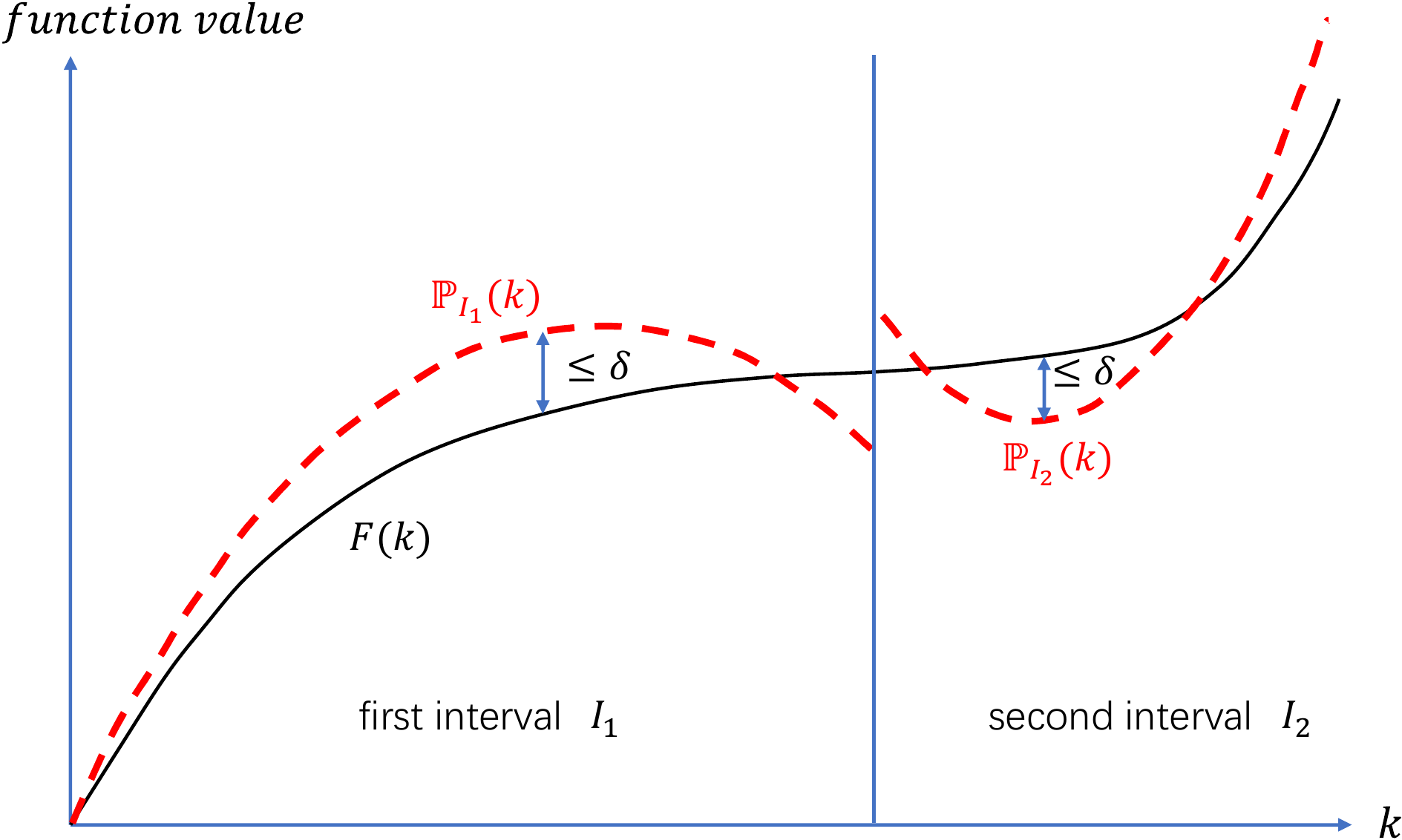}
	\vspace{-2mm}
	\caption{Fitting $F(k)$ with multiple polynomial functions, subject to the deviation threshold $\delta$}
	\label{fig:bounded_error_fitting}
	\vspace{-4mm}
\end{figure}

To achieve a small index size, we aim to minimize the number of intervals (i.e., $h$ in Figure \ref{fig:indexframework}).
An existing dynamic programming (DP) approach~\cite{leenaerts2013piecewise}, though designed for piecewise linear functions,
can be adapted to solve our partitioning problem of $F(k)$.
However, this method takes $O(n^2 \times \ell_{max}^{2.5})$ time\footnote{Recall that the state-of-the-art linear programming solver \cite{lee2015efficient} takes $O(\ell_{max}^{2.5})$ time for each curve-fitting problem (cf. Equation \ref{eq:LP}).}, where $\ell_{max}$ is the maximum number of keys covered by any interval. Obviously, this method does not scale well with the data size $n$.

In Section~\ref{sec:GS}, we present a more efficient method, called greedy segmentation (GS), to segment the exact function $F(k)$.
As we show later, the time complexity of GS is $O(n \times \ell_{max}^{2.5})$,
which scales well with the data size $n$.
Then, in Section~\ref{sec:OPT}, we show that GS is guaranteed to return the optimal solution.



\subsubsection{{\bf Greedy Segmentation (GS) Method}}
\label{sec:GS}
We present the pseudo-code of the Greedy Segmentation (GS) method in Algorithm~\ref{alg:GS}.
It examines the key domain from left to right (line 2).
In each iteration, it expands the interval $I$ by including the next key (line 3),
calls an LP solver on the interval $I$ to obtain a fitting function $\mathbb{P}_{now}$ (line 4),
and tests whether it fulfills the error requirement.
When this test fails (i.e., $E(I) > \delta$),
we conclude that the previous interval is a maximal interval
and thus insert its corresponding fitting function $\mathbb{P}_{prev}$ into the result.
The above procedure is repeated until all keys are covered.

\vspace{-2.5mm}
\begin{algorithm}[hbt]
	\small
	\caption{Greedy Segmentation (GS)}
	\begin{algorithmic}[1]
		\Statex {\bf Input:} function $F(k)$, degree $deg$, deviation threshold $\delta$
		\Statex {\bf Output:} sequence of polynomial functions $\text{Seq}_{\mathbb{P}}$
		\State $\text{Seq}_{\mathbb{P}} \leftarrow \emptyset$; \; $l \gets 1$; \; $\mathbb{P}_{prev} \leftarrow null$
		\For{$u \gets 2$ to $n$}
    		\State $I\gets [k_l,k_u]$ \Comment the interval for polynomial function $\mathbb{P}$
    		\State $\mathbb{P}_{now} \leftarrow$ call LP solver on $I$ \Comment Equation \ref{eq:LP}
    		\If{$E(I) > \delta$ or $u = n$} \Comment Equation \ref{eq:error}
                \State insert $\mathbb{P}_{prev}$ into $\text{Seq}_{\mathbb{P}}$
                \State $l\gets u$
    		\EndIf
    		\State $\mathbb{P}_{prev} \leftarrow \mathbb{P}_{now}$
		\EndFor
		\State \Return $\text{Seq}_{\mathbb{P}}$
	\end{algorithmic} \label{alg:GS}
\end{algorithm}
\vspace{-2.5mm}


The time complexity of GS is $O(n\ell_{max}^{2.5})$
because it invokes $O(n)$ calls to the LP solver, where each call takes $O(\ell_{max}^{2.5})$ time~\cite{lee2015efficient}.
We further accelerate GS by applying an existing exponential search technique~\cite{JA76},
which can reduce the number of LP calls per interval by $\frac{\ell}{\log \ell}$ times.
With this technique, GS takes only 70 seconds (cf. Section \ref{sec:exp_construction_time}) to complete for a real dataset with 1 million data points.
This is acceptable for many data analytics tasks (with static datasets) in OLAP. In our experiments, we find that $\ell_{max}$ usually ranges between hundreds and thousands, thus the term $O(\ell_{max}^{2.5})$ is acceptable in practice. In Appendix \cite{ZTMC20_arxiv}, we discuss how to utilize parallel computation to further improve the construction time.


\vspace{-2mm}
\subsubsection{{\bf GS is Optimal}} \label{sec:OPT}
%
We first prove the following property (Lemma \ref{lem:mono}) of our curve fitting problem (cf. Definition \ref{def:min_error_problem}).
\vspace{-1mm}
\begin{lemma} \label{lem:mono}
Let $I_l$ and $I_u$ be two intervals, which contain two sets of keys $S_{l}$ and $S_{u}$, respectively.
If $S_{l} \subseteq S_{u}$, then $E(I_l) \leq E(I_u)$.
\end{lemma}
\vspace{-2mm}
\begin{proof}
Recall that the value of $E(I)$ (cf. Equation \ref{eq:error}) is equal to the minimum value of the optimization problem (Equation~\ref{eq:LP}). Since $S_l$ is a subset of $S_u$, the set of constraints for solving $E(I_l)$ is also the subset of constraints for solving $E(I_u)$.
Thus, for the minimization problem in Equation~\ref{eq:LP},
the possible solution space for $S_l$ is a superset of the possible solution space for $S_u$.
Therefore, we conclude that $E(I_l) \leq E(I_u)$.
\end{proof}
\vspace{-2mm}


Based on Lemma \ref{lem:mono}, we then show that GS produces the fewest polynomial functions (cf. Theorem \ref{thm:GS_opt}), i.e., the optimal solution.
\vspace{-2mm}
\begin{theorem}
\label{thm:GS_opt}
GS always produces the optimal number of functions (with respect to the given parameters $deg$ and $\delta$).
\end{theorem}

\begin{proof}
We denote the minimum key and the maximum key of an interval $I$
by $I.\min$ and $I.\max$, respectively.

Let $\mathcal{I}_{\text{OPT}}^*=(I_{\text{OPT}}^{(1)},I_{\text{OPT}}^{(2)},\cdots)$ and $\mathcal{I}_{\text{GS}}^*=(I_{\text{GS}}^{(1)},I_{\text{GS}}^{(2)},\cdots)$ be two ascending sequences of intervals for the optimal solution and our GS method, respectively (i.e., $I^{(i)}.\max < I^{(i+1)}.\min$ for $i=1,2,...,n-1$).
Every interval $I$ in $\mathcal{I}_{\text{OPT}}^*$ and $\mathcal{I}_{\text{GS}}^*$ must satisfy $E(I) \le \delta$.
We now prove the theorem by mathematical induction.

In the base step, we consider the first interval in each sequence. Since both GS and OPT must cover the key domain, we have:
$$I_{\text{GS}}^{(1)}.\min = I_{\text{OPT}}^{(1)}.\min$$

According to GS, the first interval $I_{\text{GS}}^{(1)}$ is maximal,
because a longer interval would violate the deviation threshold $\delta$.
Thus, we have:
\begin{equation}
\label{eq:initial_end}
I_{\text{GS}}^{(1)}.\max \geq I_{\text{OPT}}^{(1)}.\max
\end{equation}

In the inductive step, assume that the first $\ell$ intervals of the two sequences satisfy the following property:
\begin{equation}
\label{eq:initial_end}
I_{\text{GS}}^{(\ell)}.\max \geq I_{\text{OPT}}^{(\ell)}.\max
\end{equation}

Since $\mathcal{I}_{\text{OPT}}^*$ and $\mathcal{I}_{\text{GS}}^*$ are ascending sequences of intervals,
Equation~\ref{eq:initial_end} implies the following:
\begin{equation}
\label{eq:second_intervals_initial}
I_{\text{GS}}^{(\ell+1)}.\min \geq I_{\text{OPT}}^{(\ell+1)}.\min
\end{equation}
Now, we consider two cases for comparing $I_{\text{GS}}^{(\ell+1)}$ and $I_{\text{OPT}}^{(\ell+1)}$.

{\bf Case 1:}
$$I_{\text{GS}}^{(\ell+1)}.\max \ge I_{\text{OPT}}^{(\ell+1)}.\max$$
In this case, the first $\ell+1$ intervals of GS cover all keys in the first $\ell+1$ intervals of OPT.

{\bf Case 2:}
\begin{equation}
\label{eq:second_intervals_end_assumption}
I_{\text{GS}}^{(\ell+1)}.\max < I_{\text{OPT}}^{(\ell+1)}.\max
\end{equation}

Consider the interval $I'=[I_{\text{GS}}^{(\ell+1)}.\min,I_{\text{OPT}}^{(\ell+1)}.\max]$.
By using Equations \ref{eq:second_intervals_initial} and \ref{eq:second_intervals_end_assumption},
we obtain: $I' \subset I_{\text{OPT}}^{(\ell+1)}$.
By Lemma~\ref{lem:mono}, we get: $E(I') \le E(I_{\text{OPT}}^{(\ell+1)})$.
Since $E(I_{\text{OPT}}^{(\ell+1)}) \le \delta$, we get: $E(I') \le \delta$.

Observe that $I'$ has the same minimum key as  $I_{\text{GS}}^{(\ell+1)}$
but a larger maximum key than $I_{\text{GS}}^{(\ell+1)}$.
Since $I'$ does not pass the error test in GS, we get $E(I') > \delta$.
This contradicts the statement $E(I') \le \delta$.

Therefore, only the first case is true, and we have:
$$I_{\text{GS}}^{(\ell+1)}.\max \ge I_{\text{OPT}}^{(\ell+1)}.\max$$
	
This means GS always covers no fewer keys than OPT with the same number of intervals. Thus, GS produces the optimal number of functions.
\end{proof}

\vspace{-5mm}
\subsection{Indexing of polynomial functions} \label{sec:organizing}
\vspace{-1mm}
In our experimental study, the number of intervals (for polynomial functions) ranges from 100 to 1000.
We adopt existing index structures on these intervals to support fast query evaluation.
Specifically, we employ an in-memory index called the STX B-tree~\cite{STXBtree} to index intervals.
In each internal node entry, we maintain an additional attribute to store the aggregate value of its subtree.
In each leaf node entry, we store an interval and its corresponding polynomial model (in the form of coefficients). 
In summary, this index is similar to the aggregate tree exemplified in Figure \ref{fig:max_index}(b), 
except that we store polynomial models in leaf nodes.


%% file: query.tex
\vspace{-2mm}
\section{Approximate Query Evaluation} \label{sec:query}
\vspace{-1mm}
We present our framework for answering approximate range aggregate queries in Figure~\ref{fig:queryframework}.
The first step is to compute an initial approximate result quickly by using our index (\polyfit{}).
Then, we check whether the error condition is satisfied and refine the approximate result if necessary.
We discuss how to answer the approximate range \verb"SUM" query
and the approximate range \verb"MAX" query in Sections~\ref{sec:count} and \ref{sec:max}, respectively.
Finally, in Section~\ref{sec:choose_degree}, we discuss how to tune our index parameters (e.g., $deg, \delta$)
in order to optimize the query response time.



\begin{figure*}[!hbt]
\center
\includegraphics[width=2.0\columnwidth]{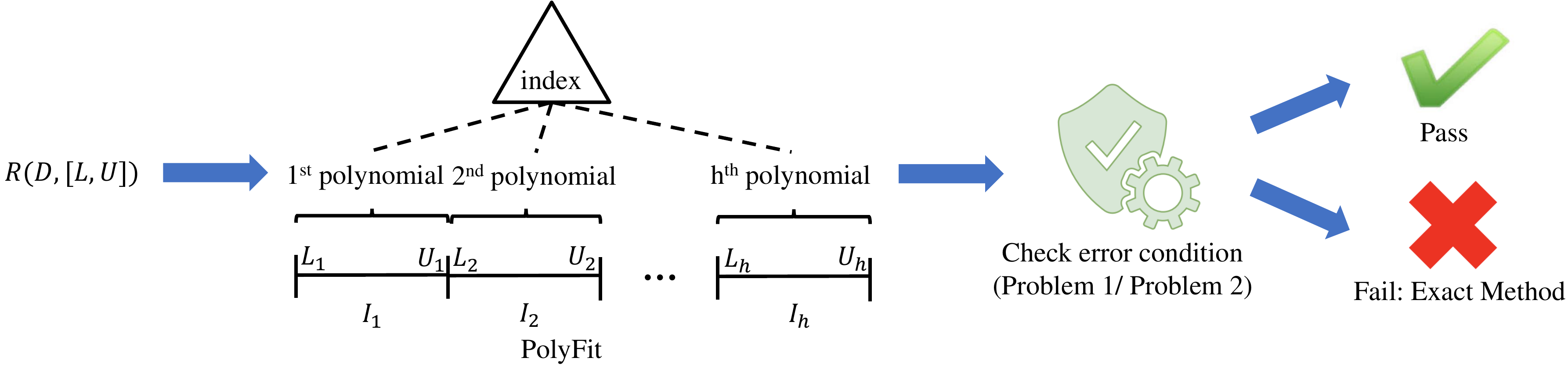}
\vspace{-2.5mm}
\caption{Querying framework for \polyfit{}}
\label{fig:queryframework}
\vspace{-4mm}
\end{figure*}

\vspace{-2mm}
\subsection{Approximate range \texttt{SUM} Query} \label{sec:count}
\vspace{-1mm}
Given the query range $[l_q,u_q]$, we propose to compute the approximate result as:
\begin{equation} \label{eq:approx-sum}
\tilde{A}_{sum}=\mathbb{P}_{I_u}(u_q)-\mathbb{P}_{I_l}(l_q),
\end{equation}

\noindent where $I_l$ and $I_u$ denote the intervals of $\mathbb{P}$ that contain the values $l_q$ and $u_q$, respectively.


Then, we show the error conditions for
$Q_{abs}$ (cf. Problem \ref{prob:abs_error}) and $Q_{rel}$ (cf. Problem \ref{prob:rel_error}).


\stitle{Error condition for $Q_{abs}$}
Given the absolute error $\varepsilon_{abs}$,
we recommend to use the deviation threshold $\delta=\frac{\varepsilon_{abs}}{2}$ in constructing \polyfit{}.
With this setting, the following lemma offers the absolute error guarantee
for the approximate result $\tilde{A}_{sum}$ (in Equation~\ref{eq:approx-sum}).


\begin{lemma}
\label{lem:count_abs_error}
If $\delta=\frac{\varepsilon_{abs}}{2}$,
then $\tilde{A}_{sum}$ (in Equation~\ref{eq:approx-sum}) satisfies the absolute error guarantee with respect to $\varepsilon_{abs}$.
\end{lemma}
\begin{proof}
Let $I_l$ and $I_u$ be two intervals (in \polyfit{}) which contain $l_q$ and $u_q$
(of the query range $[l_q,u_q]$), respectively.
Based on the deviation threshold guarantee in Section~\ref{sec:OPT}, we obtain:
$$|CF_{sum}(l_q)-\mathbb{P}_{I_l}(l_q)| \leq \delta$$
$$|CF_{sum}(u_q)-\mathbb{P}_{I_u}(u_q)| \leq \delta$$

By combining them, we have:
{\small $$CF_{sum}(u_q)-CF_{sum}(l_q)-2\delta \leq \tilde{A}_{sum} \leq CF_{sum}(u_q)-CF_{sum}(l_q)+2\delta$$}
By using Equation~\ref{eq:R_c_cum}, we have:
$$R_{sum}(\mathcal{D},[l_q,u_q])-2\delta \leq \tilde{A}_{sum} \leq R_{sum}(\mathcal{D},[l_q,u_q])+2\delta$$
Since $\delta=\frac{\varepsilon_{abs}}{2}$, $\tilde{A}_{sum}$ satisfies the absolute error guarantee $\varepsilon_{abs}$.
\end{proof}

\vspace{-3mm}
\stitle{Error condition for $Q_{rel}$}
In this scenario, there is no specific preference for
setting the deviation threshold $\delta$ when constructing \polyfit{}.
The following lemma suggests a condition to test whether $\tilde{A}_{sum}$ satisfies the relative error guarantee.
If this test fails, we resort to the exact method (cf. Section \ref{sec:exact_count}) to obtain the exact result.



\begin{lemma}
\label{lem:count_rel_error}
If $\tilde{A}_{sum} \geq 2\delta(1+\frac{1}{\varepsilon_{rel}})$,
then $\tilde{A}_{sum}$ (in Equation~\ref{eq:approx-sum}) satisfies the relative error guarantee with respect to $\varepsilon_{rel}$.
\end{lemma}
\begin{proof}
Like in the proof of Lemma~\ref{lem:count_abs_error}, we can derive
Equations~\ref{eq:numerator} and~\ref{eq:denominator}.
\begin{equation}
\label{eq:numerator}
|\tilde{A}_{sum}-R_{sum}(\mathcal{D},[l_q,u_q])| \leq 2\delta
\end{equation}
which also implies (by simple derivations):
\begin{equation}
\begin{split}
\label{eq:denominator}
R_{sum}(\mathcal{D},[l_q,u_q])\geq \tilde{A}_{sum} - 2\delta
\end{split}
\end{equation}

Since $\delta$ and $\varepsilon_{rel}$ must be positive,
the given condition $\tilde{A}_{sum} \geq 2\delta(1+\frac{1}{\varepsilon_{rel}})$
implies that $\tilde{A}_{sum} > 2\delta$ and $\frac{2\delta}{\tilde{A}_{sum}-2\delta}\leq \varepsilon_{rel}$.

Dividing Equation \ref{eq:numerator} by Equation \ref{eq:denominator}, we obtain the following inequality (under the condition $\tilde{A}_{sum} > 2\delta$).
$$\frac{|\tilde{A}_{sum}-R_{sum}(\mathcal{D},[l_q,u_q])|}{R_{sum}(\mathcal{D},[l_q,u_q])} \leq \frac{2\delta}{\tilde{A}_{sum}-2\delta}$$

This completes the proof because $\frac{2\delta}{\tilde{A}_{sum}-2\delta}\leq \varepsilon_{rel}$.
%
\end{proof}

\stitle{The overall query algorithm}
We summarize the query algorithm for both types of error guarantees in Algorithm~\ref{alg:query_processing_count}.
The processing for $Q_{abs}$ is composed of two parts: index search $T_1$ (i.e., Lines 1-2) and function evaluation $T_2$ (i.e., Line 3).
The processing for $Q_{rel}$ includes $T_1$, $T_2$, and possible refinement $T_3$ (i.e., Lines 4-6).
The time complexity of $T_1, T_2$, and $T_3$ are $O(\log(|\text{Seq}_{\mathbb{P}}|))$, $O(deg)$, and $O(\log |\mathcal{D}|)$ respectively.

\vspace{-3mm}
\begin{algorithm}[hbt]
	\small
	\caption{Query Processing for SUM (or COUNT)}
	\label{alg:query_processing_count}
	\begin{algorithmic}[1]
		\Statex {\bf Input:} $\text{Seq}_{\mathbb{P}}$ (output from Algorithm \ref{alg:GS}), $l_q$, $u_q$, $\mathcal{D}$, $\delta$, $Q_{type}$
		\Statex {\bf Output:} Approximate query result $A$
		\State $\mathbb{P}_{I_l} \leftarrow$ index search $\mathbb{P}$ from $\text{Seq}_{\mathbb{P}}$ that includes $l_q$
		\State $\mathbb{P}_{I_u} \leftarrow$ index search $\mathbb{P}$ from $\text{Seq}_{\mathbb{P}}$ that includes $u_q$
		\State $\tilde{A}_{sum} \leftarrow \mathbb{P}_{I_u}(u_q) - \mathbb{P}_{I_l}(l_q)$
		\If{$Q_{type}=Q_{rel}$}
    		\If{$\tilde{A}_{sum}$ fails the error condition of Lemma~\ref{lem:count_rel_error}}
                \State $\tilde{A}_{sum} \leftarrow$ perform refinement on $\mathcal{D}$ \Comment Section \ref{sec:exact_count}
            \EndIf
		\EndIf
		\State \Return $\tilde{A}_{sum}$
	\end{algorithmic}
\end{algorithm}
\vspace{-3mm}

\vspace{-3mm}
\subsection{Approximate range \texttt{MAX} Query} \label{sec:max}
The query method described in Section~\ref{sec:exact_max} can be applied here,
except that we employ the index described in Section~\ref{sec:organizing}.

Given the query range $[l_q,u_q]$, we propose to compute the approximate result as:
\begin{equation} \label{eq:approx-max}
\begin{split}
\tilde{A}_{max}= \max \{
\max_{k \in I_l, k \ge l_q} \mathbb{P}_{I_l}(k),
\max_{k \in I_u, k \le u_q} \mathbb{P}_{I_u}(k), \\
\max_{N_j.I \subseteq [l_q,u_q]}  N_j.max  \}
\end{split}
\end{equation}

\noindent where $N_j$ denotes an internal node of the index built on top of $\text{Seq}_{\mathbb{P}}$. $I_l$ and $I_u$ denote the intervals of $\mathbb{P}$ that contain the values $l_q$ and $u_q$, respectively.

The error conditions for $Q_{abs}$ and $Q_{rel}$ are presented in
Lemmas \ref{lem:max_abs_error} and \ref{lem:max_rel_error} respectively.
We omit their proofs; they are similar to the proofs of Lemmas \ref{lem:count_abs_error} and \ref{lem:count_rel_error}.

\begin{lemma} \label{lem:max_abs_error}
If $\delta=\varepsilon_{abs}$, then
$\tilde{A}_{max}$ (in Equation~\ref{eq:approx-max}) satisfies the absolute error guarantee $\varepsilon_{abs}$.
\end{lemma}

\begin{lemma} \label{lem:max_rel_error}
If $\tilde{A}_{max}\geq \delta(1+\frac{1}{\varepsilon_{rel}})$, then
$\tilde{A}_{max}$ (in Equation~\ref{eq:approx-max}) satisfies the relative error guarantee $\varepsilon_{rel}$.
\end{lemma}

We now discuss how to evaluate Equation~\ref{eq:approx-max} in greater detail.
The third term is contributed by the inner nodes of the aggregate R-tree whose intervals are covered by $[l_q,u_q]$.
Regarding the first two terms, it suffices to find the maximum values for $\mathbb{P}_{I_l}(k)$ and $\mathbb{P}_{I_u}(k)$ in regions $[l_q, U_{I_l}]$ and $[L_{I_u}, u_q]$, as shown in Figure~\ref{fig:max_query}, where $U_{I_l}$ ($L_{I_u}$) is the upper (lower) end of the leaf node interval that $l_q$ ($u_q$) overlaps.
These values (i.e., red dots) can be calculated by checking the border points and the zero derivative points.


\begin{figure}[!hbt]
\vspace{-2mm}
\center
\includegraphics[width=0.9\columnwidth]{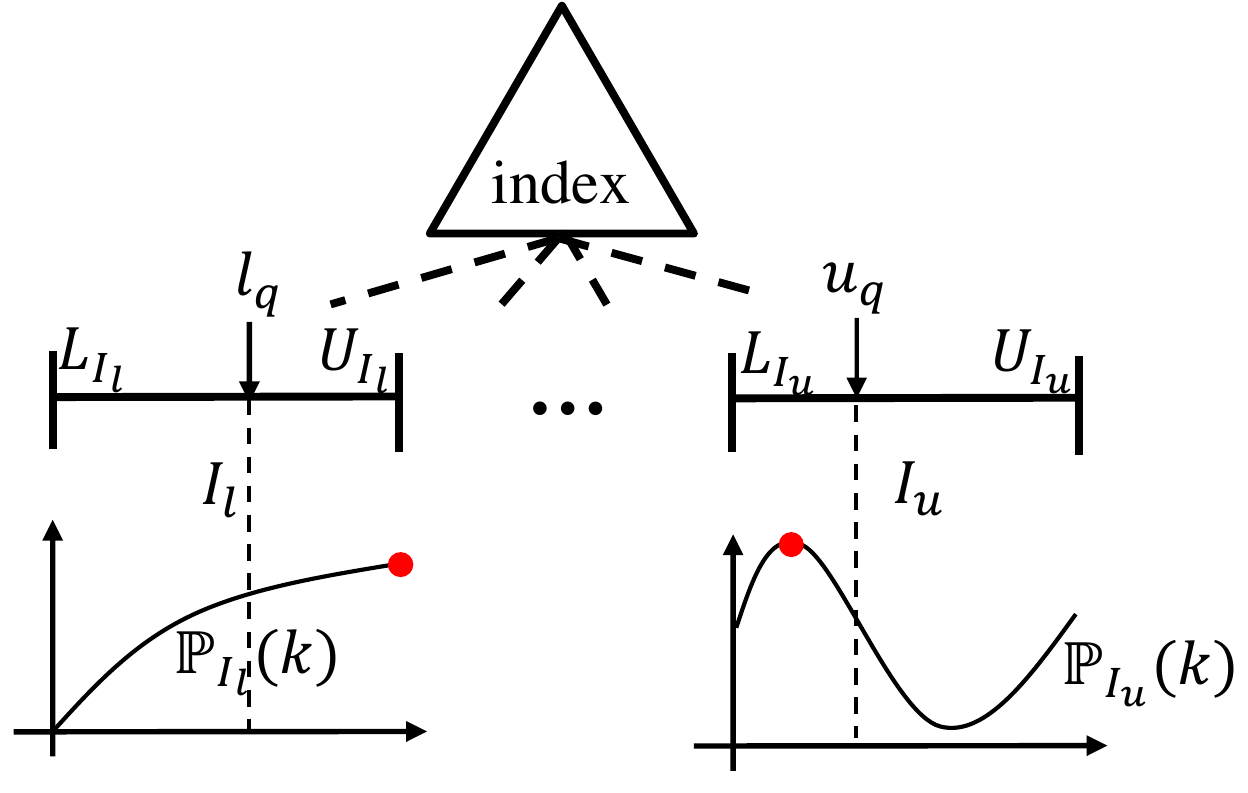}
\vspace{-3.5mm}
\caption{The maximum measure values (red dots) for two leaf nodes, which include $l_q$ and $u_q$}
\label{fig:max_query}
\vspace{-4.5mm}
\end{figure}



\stitle{The overall query algorithm}
We conclude the query algorithm for both types of error guarantees in Algorithm \ref{alg:query_processing_max}.
The processing for $Q_{abs}$ consists of two parts: index search $T_1$ (i.e., Line 3) and function evaluation $T_2$ (Lines 8-9).
The processing for $Q_{rel}$ includes $T_1, T_2$, and possible refinement $T_3$ (i.e., Lines 10-12).
The time complexities of $T_1$ and $T_3$ are still $O(\log(|\text{Seq}_{\mathbb{P}}|))$ and $O(\log |\mathcal{D}|)$.
However, for $T_2$, this includes calculating the zero derivative points within the intersection region.
If the degree is between 1 and 5, closed-form equations exist,
where the number of arithmetic operations in these cases are summarized in Table~\ref{tab:max_degree}.
Starting from degree 6, there is no closed-form equations, and thus require expensive
numerical evaluation methods like gradient descent~\cite{stewart2015galois}.
In practice, we recommend to use degrees up to 3 for the approximate range \verb"MAX" query.


\vspace{-3mm}
\begin{algorithm}[hbt]
	\small
	\caption{Query Processing for MAX (or MIN)}
	\label{alg:query_processing_max}
	\begin{algorithmic}[1]
		\Statex {\bf Input:} Aggregate R-tree $N$ on $\text{Seq}_{\mathbb{P}}$, $l_q$, $u_q$, $\mathcal{D}$, $\delta$, $Q_{type}$
		\Statex {\bf Output:} Approximate query result $A$
		\State $\tilde{A}_{max} \leftarrow -\infty$
		\If{$N$ is an internal node}
		\State update $\tilde{A}_{max}$ based on aggregate R-tree's mechanism
		\Else{}
		\For{leaf element $\mathbb{P}$ in $N$}
		\If{$\mathbb{P}.I \cap [l_q, u_q] \ne \varnothing$} \Comment the interval $\mathbb{P}$ covered
		\State $I^* \leftarrow \mathbb{P}.I \cap [l_q, u_q]$
		\State $\beta \leftarrow \{x \in I^* \; | \; \mathbb{P}'(x) = 0 \}$ \Comment zero derivative points
		\State $\tilde{A}_{max} \leftarrow \max( \tilde{A}_{max}, \max_{x \in \beta} \mathbb{P}(x), \mathbb{P}(I^*.l), \mathbb{P}(I^*.u)))$
		\EndIf
		\EndFor
		\EndIf
		\If{$N$ is root node and $Q_{type}=Q_{rel}$}
		\If{$\tilde{A}_{max}$ fails the error condition of Lemma~\ref{lem:max_rel_error}}
		\State $\tilde{A}_{max} \leftarrow$ perform refinement on $\mathcal{D}$ \Comment Section \ref{sec:exact_max}
		\EndIf
		\EndIf
		\State \Return $A$
	\end{algorithmic}
\end{algorithm}
\vspace{-3mm}

\vspace{-2mm}
\begin{table}[!htb]
	\centering
	\caption{Number of arithmetic operations for calculating zero derivative points}
	\label{tab:max_degree}
	\vspace{-3mm}
	\begin{tabular}{|c|c|c|c|c|c|}
		\hline
		degree & 1 & 2 & 3 & 4 & 5 \\ \hline
		operations & 0 & 2 & up to 18 & up to 261 & up to 1612  \\ \hline
	\end{tabular}
	\vspace{-2mm}
\end{table}


\vspace{-3mm}
\subsection{Tuning $deg$ and $\delta$}
\label{sec:choose_degree}
We discuss the effect of our index parameters (i.e., $deg, \delta$)
on the query response time and examine how to tune them.

\begin{figure}[!hbt]
	\centering
	\vspace{-0.5mm}
	\includegraphics[width=0.8\columnwidth]{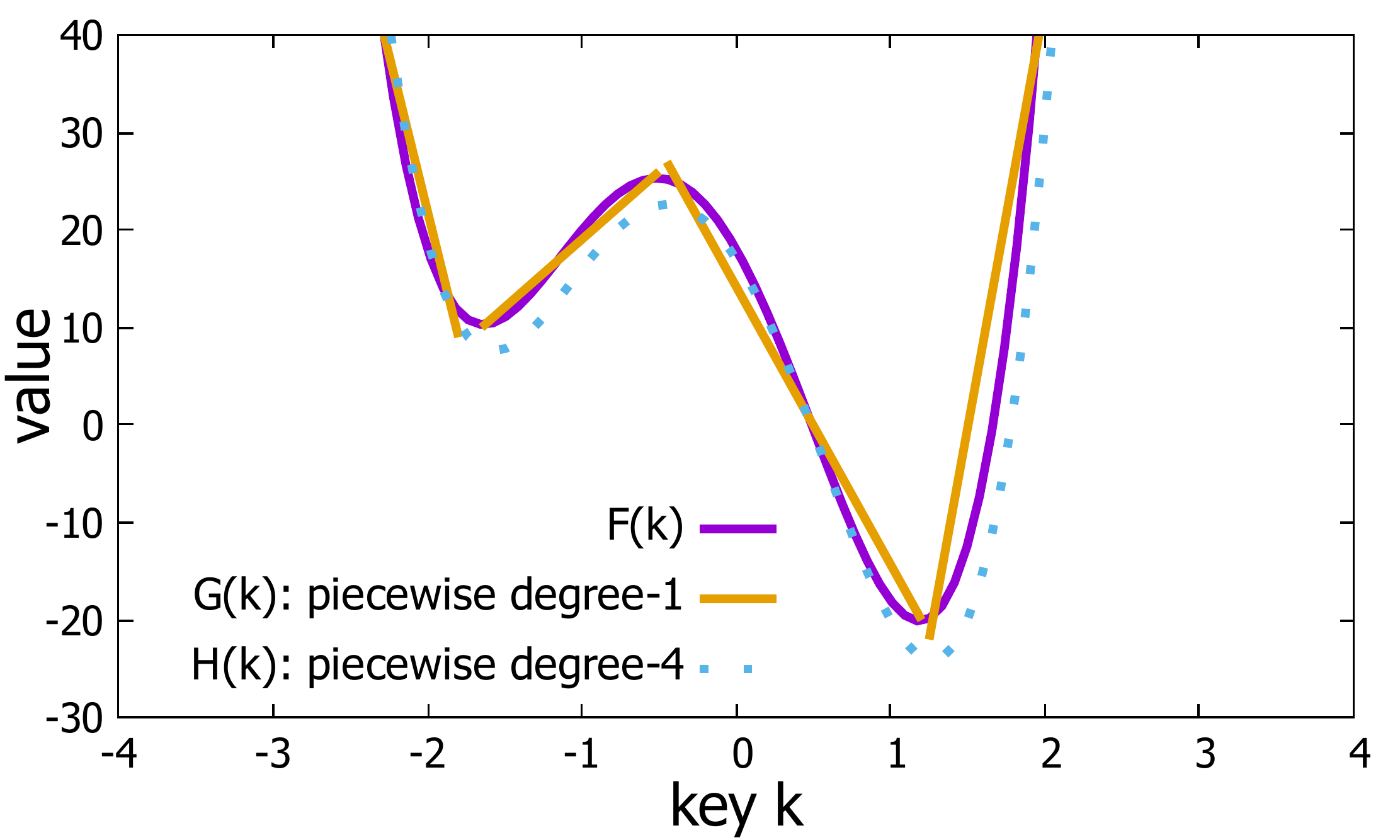}
	\vspace{-2.5mm}
	\caption{An example of degree selection}
	\vspace{-3.5mm}
	\label{fig:choosedegree}
\end{figure}

\stitle{How to tune the degree $deg$?}
The exact function $F(k)$ is approximated by
different polynomial functions with different degrees.
For instance, in Figure~\ref{fig:choosedegree},
the exact function $F(k)$ is approximated, among others, by the following functions (within the deviation threshold $\delta$):
(i) a piecewise function $G(k)$ with four pieces of degree-1 functions, or
(ii) a single-piece function $H(k)$ of degree-4.
Based on our experimental findings (cf. Section~\ref{sec:exp_polyfit_analysis}), we recommend to set the degree to 2 or 3.
In general, one could generate a random workload of queries to measure the performance of an index,
and then test the performance of index structures using different degrees (e.g., from 1 to 4).


	
As a remark, it is not practical to use large degree,
due to the limited precision of numeric data types in both the linear programming solver and the programming language~\cite{CPLEXPrecision1, CPLEXPrecision2}.
For example, IBM CPLEX uses $\kappa$ (kappa) as a statistical measurement of numerical difficulties.
In our experiments, the $\kappa$ value of a degree-4 polynomial (1E+10) is much higher than that of a degree-1 polynomial (1E+05).

\stitle{How to tune $\delta$?}
The tuning of $\delta$ depends on the most frequent query type used in the given application.
For $Q_{abs}$ (i.e., Problem \ref{prob:abs_error}),
if all users share the same absolute error threshold $\varepsilon_{abs}$, then it is used to
derive the value of $\delta$, according to Lemmas \ref{lem:count_abs_error} and \ref{lem:max_abs_error}.
Otherwise, we can select the value of $\delta$ such that it satisfies the error requirements
for the majority of users (e.g., 80\%).

For $Q_{rel}$ (i.e., Problem \ref{prob:rel_error}),
the processing includes three phases: index search, function evaluation, and refinement
(cf. Algorithms \ref{alg:query_processing_count} and \ref{alg:query_processing_max}).
A large $\delta$ leads to fast index search but high refinement probability.
In contrast, a small $\delta$ leads to slow index search but low refinement probability.
Observe that refinement is often more expensive than index search.
We recommend to pick a small $\delta$ such that
most users avoid the refinement phase.
In our experiments, we examine different values of $\delta$ (e.g., 25, 50, 100,
200, 500, and 1000) to identify the best setting in terms of the query response time.



%% file: case_2d_v2.tex
\vspace{-2mm}
\section{Extensions: Queries with Two Keys}
\vspace{-1mm}
\label{sec:extend_2d}
Previous sections consider range aggregate queries with a single key (cf. Definition \ref{def:RAQ}). 
We now discuss how to support range aggregate queries with two keys (cf. Definition \ref{def:2D_Approx_RAQ}). 
Due to the space limit, we only consider the \verb"COUNT" query. In Appendix A.5 \cite{ZTMC20_arxiv}, we discuss the case of more than two keys.

\begin{definition}
\label{def:2D_Approx_RAQ}
Let $\mathcal{D}$ be a set of records $(u, v, w)$, where $u$, $v$, and $w$ are the first key, the second key, and the measure, respectively.
Given the query ranges $[l_q^{(1)}, u_q^{(1)}]$ and $[l_q^{(2)}, u_q^{(2)}]$
for $u$ and $v$, respectively, we define the \verb"COUNT" query as: 
\vspace{-0.1cm}
{
\small
\begin{equation}
R_{count}(\mathcal{D},[l_q^{(1)}, u_q^{(1)}][l_q^{(2)}, u_q^{(2)}])=\verb"COUNT"(V)
\end{equation}
}
where $V$ is the multi-set of measure values defined below:
$$V= \{ m: (k^{(1)}, k^{(2)}, m) \in \mathcal{D}, l_q^{(1)} \leq k^{(1)} \leq u_q^{(1)}, l_q^{(2)} \leq k^{(2)} \leq u_q^{(2)} \}$$
\end{definition}

We build the following key-cumulative function to represent the surface (cf. Figure \ref{fig:example2D}), which is formulated in Definition \ref{def:accu2D}.
\begin{definition}
\label{def:accu2D}
The key-cumulative function with two keys for \verb"COUNT" query is defined as $CF_{count}(u,v)$, where:
\begin{eqnarray}
\label{eq:accu2d}
CF_{count}(u,v) = R_{count}(\mathcal{D}[-\infty, u][-\infty, v])
\end{eqnarray}
\end{definition}

\begin{figure}[!hbt]
\vspace{-5mm}
\center
\subfloat[tweet locations as data points]{ \includegraphics[width=0.5\columnwidth]{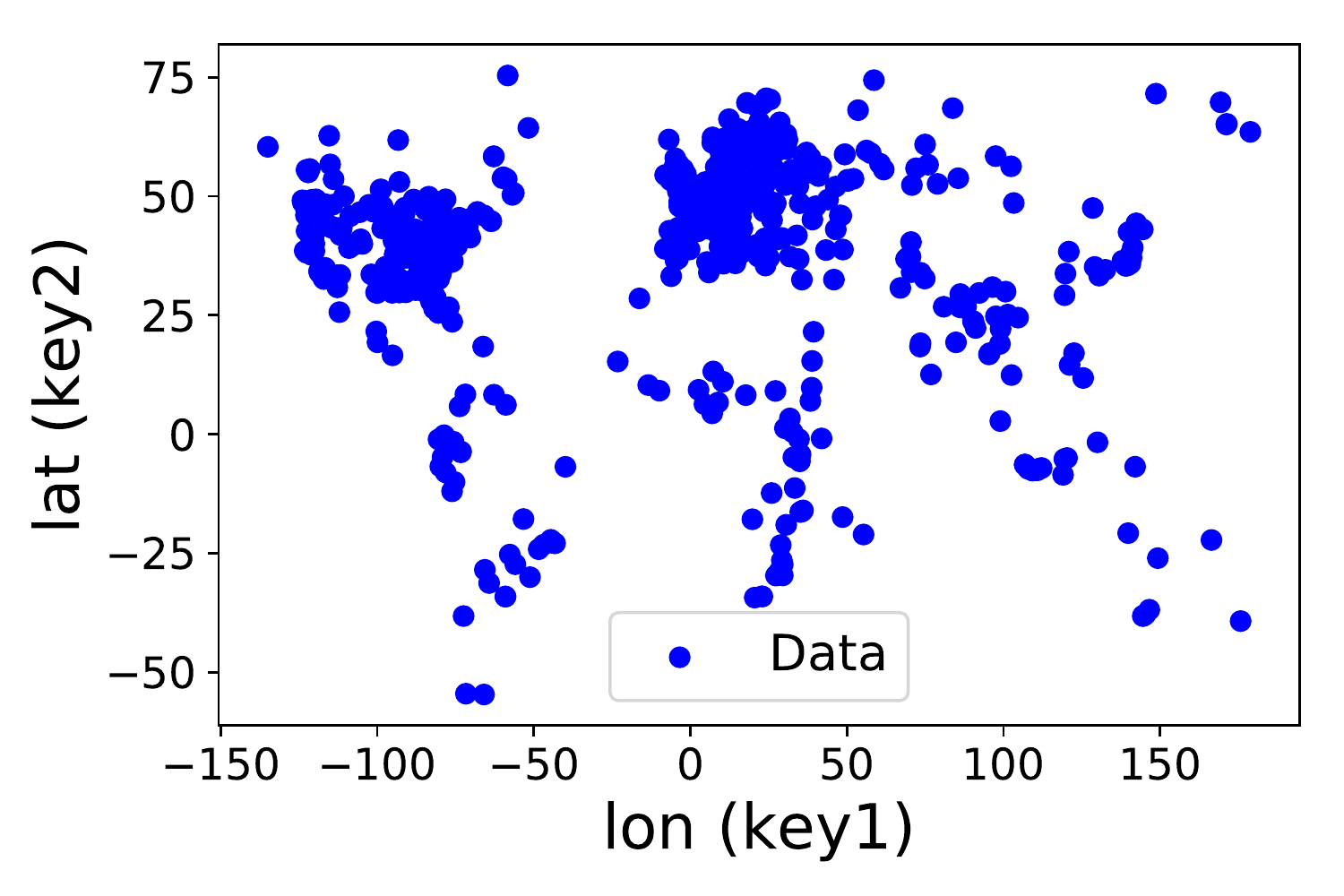} }
\subfloat[function for range COUNT queries]{ \includegraphics[width=0.5\columnwidth]{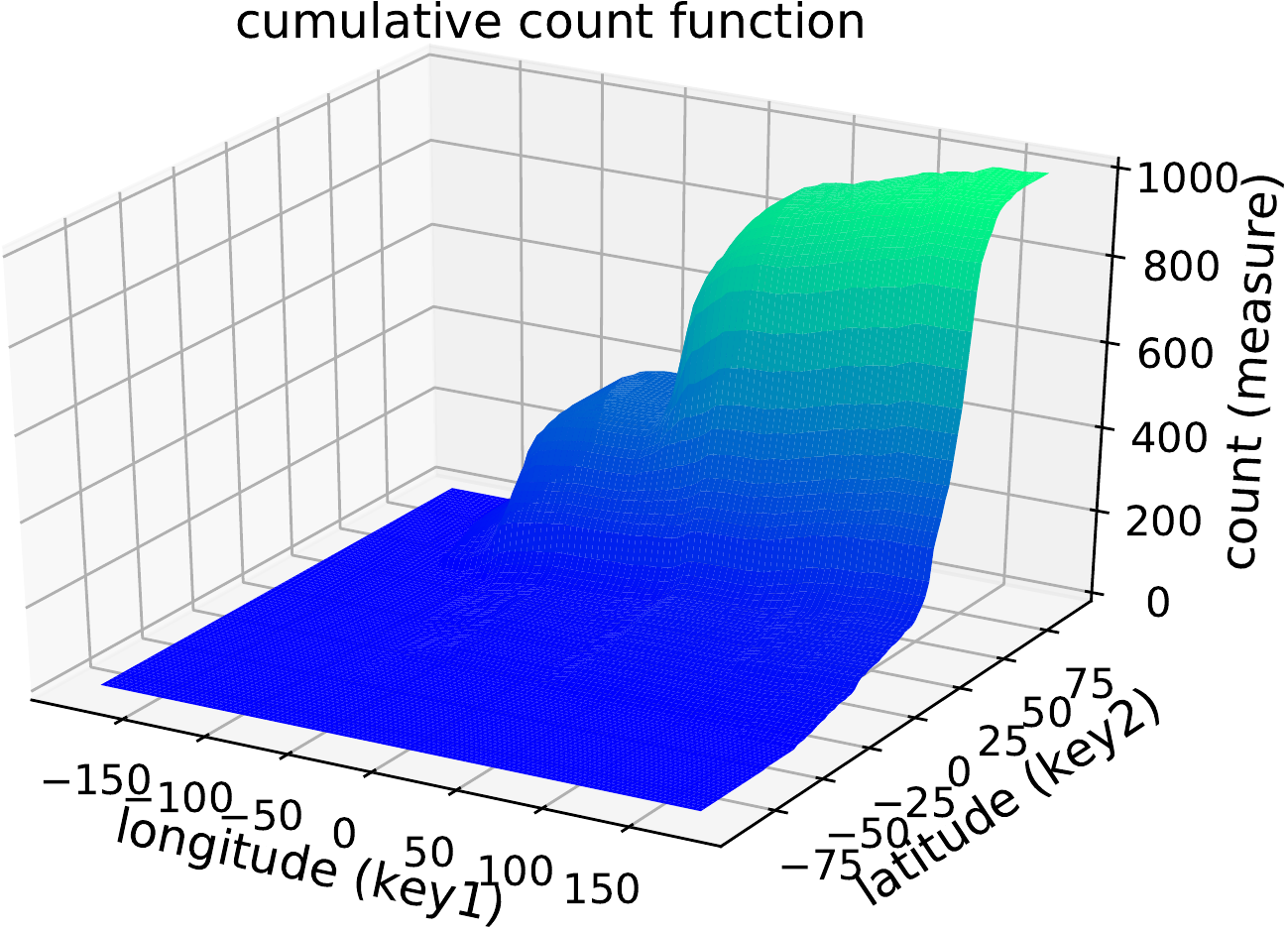} }
\vspace{-3mm}
\caption{Tweet locations, 2-dimensional keys: discrete data points vs. continuous function}
\label{fig:example2D}
\vspace{-3mm}
\end{figure}


The following equation enables us to answer the \verb"COUNT" query quickly.
{
\vspace{-1mm}
\small
\begin{equation*}
\begin{split}
R_{count}(\mathcal{D}[l_q^{(1)},u_q^{(1)}][l_q^{(2)},u_q^{(2)}])=CF_{count}(u_q^{(1)},u_q^{(2)})-CF_{count}(l_q^{(1)},u_q^{(2)})\\
-CF_{count}(u_q^{(1)},l_q^{(2)})+CF_{count}(l_q^{(1)},l_q^{(2)})
\end{split}
\end{equation*}
\vspace{-2mm}
}

Then, we follow an idea similar to that used in Section \ref{sec:min_error} and utilize the polynomial surface $\mathbb{P}(u,v)$ to approximate the key cumulative function $CF_{count}(u,v)$ with two keys, where:
\vspace{-2mm}
$$\mathbb{P}(u,v)=\sum_{i=0}^{deg} \sum_{j=0}^{deg}  a_{ij} u^i v^j$$
\vspace{-2mm}

By replacing $F(k_i)$ and $\mathbb{P}(k_i)$ in Equation \ref{eq:error} with $F(u_i,v_i)$ and $\mathbb{P}(u_i,v_i)$, respectively, we obtain a similar linear programming problem for obtaining the best parameters $a_{ij}$. However, unlike the one-dimensional case, it takes at least $O(n^2)$ to obtain the minimum number of segmentations when using the GS method (cf. Section \ref{sec:GS}), which is infeasible even for small-scale datasets (e.g., 10000 points). Instead, we propose a heuristics-based solution that performs quad-tree-like segmentations. As illustrated in Figure~\ref{fig:quad_tree}, when a region does not fulfill the error guarantee $\delta$ (e.g., white rectangles), it is  decomposed into four smaller regions. This procedure terminates when all regions satisfy the error guarantee $\delta$.

\begin{figure}[!hbt]
\vspace{-0.3cm}
\center
\includegraphics[width=1.0\columnwidth,page=1]{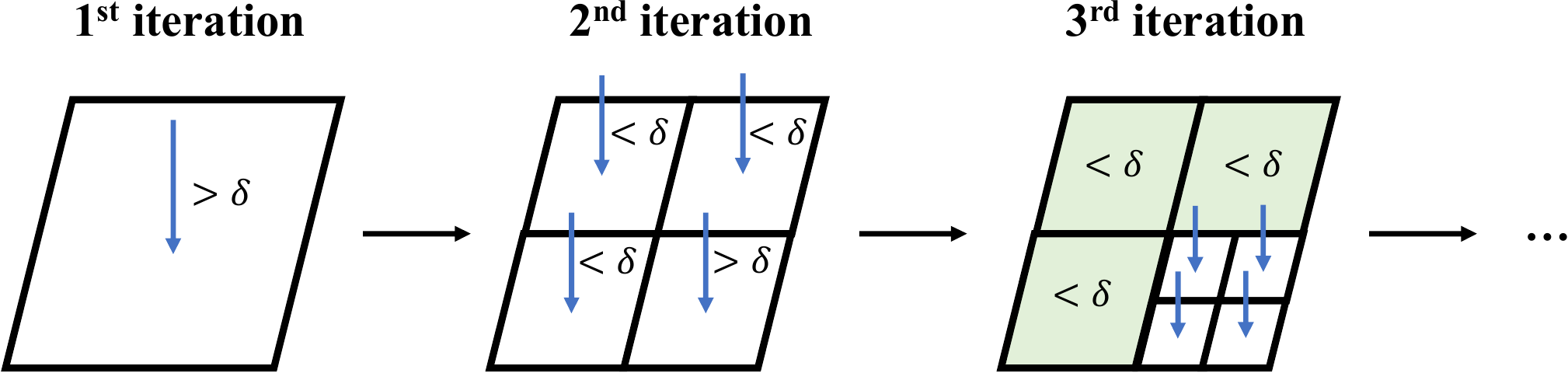}
\vspace{-0.4cm}
\caption{Quad-tree based approach for obtaining the segmentation}
\label{fig:quad_tree}
\vspace{-0.3cm}
\end{figure}

After building the \polyfit{} index structure, we utilize a similar approach in Section \ref{sec:query} to answer range aggregate queries with theoretical guarantees (cf. Lemmas \ref{lem:sum_abs_error_two_D} and \ref{lem:sum_rel_error_two_D}).


Given the query range $[l_q^{(1)}, u_q^{(1)}]$ for $u$ and $[l_q^{(2)}, u_q^{(2)}]$ for $v$, we propose to compute the approximate result as:
\begin{equation}\label{eq:2d_count}
\begin{split}
\tilde{A}_{count}=
\mathbb{P}_{I_{uu}}(u_q^{(1)},u_q^{(2)})-\mathbb{P}_{I_{lu}}(l_q^{(1)},u_q^{(2)})\\
-\mathbb{P}_{I_{ul}}(u_q^{(1)},l_q^{(2)})+\mathbb{P}_{I_{ll}}(l_q^{(1)},l_q^{(2)})
\end{split}
\end{equation}

where $I_{uu}$, $I_{lu}$, $I_{ul}$, and $I_{ll}$ denote the coverage regions of $\mathbb{P}$ that $(u_q^{(1)},u_q^{(2)})$, $(l_q^{(1)},u_q^{(2)})$, $(u_q^{(1)},l_q^{(2)})$, and $(l_q^{(1)},l_q^{(2)})$ overlap, respectively. These regions could be efficiently found with the same quad-tree index used in construction.

\begin{lemma}
\vspace{-1mm}
\label{lem:sum_abs_error_two_D} If we set $\delta=\frac{\varepsilon_{abs}}{4}$, then $\tilde{A}_{count}$ satisfies the absolute error guarantee $\varepsilon_{abs}$.
\vspace{-1mm}
\end{lemma}

\begin{lemma}
\vspace{-1mm}
\label{lem:sum_rel_error_two_D} If $\tilde{A}_{count}\geq 4\delta(1+\frac{1}{\varepsilon_{rel}})$, then $\tilde{A}_{count}$ satisfies the relative error guarantee $\varepsilon_{rel}$.
\vspace{-1mm}
\end{lemma}

The proofs of Lemma \ref{lem:sum_abs_error_two_D} and \ref{lem:sum_rel_error_two_D} are similar to those of Lemmas \ref{lem:count_abs_error} and \ref{lem:count_rel_error}, respectively. 

%% file: experiment.tex
\vspace{-2mm}
\section{Experimental Evaluation}
\label{sec:exp}
\vspace{-1mm}
We introduce the experimental setting in Section \ref{exp:setting}. Then, we investigate the performance of \polyfit{} in Section \ref{sec:investigation_polyfit}. Next, we compare PolyFit and error-bounded competitors on real datasets in Section \ref{exp:response_time_guarantee_error}. After that, we compare the response time of \polyfit{} with other heuristic methods in Section \ref{exp:heuristics}. Lastly, we compare the construction times of all methods in Section \ref{sec:construction_time_all_methods}. 

\vspace{-2mm}
\subsection{Experimental Setting}
\label{exp:setting}
\vspace{-1mm}
We use three real large-scale datasets (0.9M to 100M records) to evaluate the performance. They are summarized in Table \ref{tab:datasets}. For each dataset, we randomly generate 1000 queries. In the single-key case, we randomly choose two key values in the datasets as the start and end points of each query interval. In the two-key case, we randomly sample rectangles from the dataset as query regions. In our experiments, we focus on \verb"COUNT" and \verb"MAX" queries. Nevertheless, our methods are readily applicable to \verb"SUM" and \verb"MIN" queries.

\begin{table}[!htb]
\vspace{-3mm}
\small \center
\caption{Datasets} \label{tab:datasets}
\vspace{-3mm}
\begin{tabular}{|@{ }c@{ }|@{ }c@{ }|@{ }c@{ }|@{ }c@{ }|@{ }c@{ }|}
    \hline
    Name & Size & Key(s) & Measure & Aggregate function \\ \hline
    HKI \cite{hk40index} & 0.9M & timestamp & index value & \verb"MAX" \\ \hline
    TWEET \cite{chen2015temporal} & 1M & latitude & \# of tweets & \verb"COUNT" \\ \hline
    OSM \cite{osmplanet2019} & 100M & latitude, longitude & \# of records & \verb"COUNT" \\ \hline
\end{tabular}
\vspace{-3mm}
\end{table}

Table \ref{tab:methods} summarizes different methods for supporting range aggregate queries.
We classify these methods based on five features:
(i) whether it provides absolute error guarantees (cf. Problem \ref{prob:abs_error} ($Q_{abs}$)),
(ii) whether it provides relative error guarantees (cf. Problem \ref{prob:rel_error} ($Q_{rel}$)),
(iii) whether it supports queries with two keys (cf. Section \ref{sec:extend_2d}),
(iv) whether it supports the \verb"COUNT" query, and
(v) whether it supports the \verb"MAX" query.

\begin{table}[!htb]
\vspace{-3.5mm}
\caption{Methods for range aggregate queries} \label{tab:methods}
\vspace{-4mm}
\hspace*{-1.5mm}
\begin{tabular}{|c|c|c|c|c|c|}
\multicolumn{6}{c}{$\checkmark$ Directly support \; $\triangle$ Extend to support \; $\times$ Cannot support} \\ 
	\hline
	Method & $Q_{abs}$ & $Q_{rel}$ & 2 keys & \verb"COUNT" & \verb"MAX" \\ \hline \hline
	aR-tree \cite{DPJY01} & $\checkmark$ & $\checkmark$ & $\checkmark$ & $\checkmark$ & $\checkmark$ \\ \hline
	MRTree \cite{MRTree01} & $\checkmark$ & $\checkmark$ & $\checkmark$ & $\checkmark$ & $\checkmark$ \\ \hline
	RMI \cite{kraska2018case} & $\triangle$ & $\triangle$ & $\times$ & $\checkmark$ & $\times$ \\ \hline
	FITing-tree \cite{fiting2019}  & $\triangle$ & $\triangle$ & $\times$ & $\checkmark$ & $\times$ \\ \hline
	PGM \cite{PG20} & $\triangle$ & $\triangle$& $\times$ & $\checkmark$ & $\times$ \\ \hline
	\polyfit{} (ours) & $\checkmark$ & $\checkmark$ & $\checkmark$ & $\checkmark$ & $\checkmark$ \\ \hline \hline
	Hist \cite{to2013entropy} & $\times$ & $\times$ & $\times$ & $\checkmark$ & $\times$ \\ \hline
	S-tree \cite{STXBtree} & $\times$ & $\times$ & $\times$ & $\checkmark$ & $\times$ \\ \hline
	S2 \cite{haas1992sequential} & $\times$ & $\times$ & $\checkmark$ & $\checkmark$ & $\times$ \\ \hline
	VerdictDB \cite{VerdictDB18} & $\times$ & $\times$ & $\checkmark$ & $\checkmark$ & $\times$ \\ \hline
	DBest \cite{QP19} & $\times$ & $\times$ & $\checkmark$ & $\checkmark$ & $\times$ \\ \hline
	PLATO \cite{Plato20} & $\times$ & $\times$ & $\times$ & $\checkmark$ & $\times$ \\ \hline
\end{tabular}
\vspace{-3.5mm}
\end{table}

We first introduce the methods that can satisfy deterministic error guarantees (i.e., those with $\checkmark$ or $\triangle$ in the $Q_{abs}$ and $Q_{rel}$ columns in Table \ref{tab:methods}). The aR-tree \cite{DPJY01} is a traditional tree-based method for answering exact \verb"COUNT" and \verb"MAX" queries. The MRTree \cite{MRTree01} extends the aR-tree by utilizing progressive lower and upper bounds to answer approximate \verb"COUNT" and \verb"MAX" queries with error guarantees. In addition, both the aR-tree and the MRTree can support the range aggregate queries with two keys.
With simple modifications, the learned-index methods, including RMI \cite{kraska2018case}, FITing-tree \cite{fiting2019}, and PGM \cite{PG20}, can be extended to support range aggregate queries with both absolute and relative error guarantees. However, they are unable to support queries with two keys and the \verb"MAX" query. Due to the space limitation, we cover the modifications and parameter tuning in our technical report (cf. Appendix in \cite{ZTMC20_arxiv}). \polyfit{} supports all these five features. By default, we follow Lemmas \ref{lem:count_abs_error}, \ref{lem:max_abs_error}, and \ref{lem:sum_abs_error_two_D} to set the $\delta$ values in Problem \ref{prob:abs_error} ($Q_{abs}$), for different absolute error threshold $\varepsilon_{abs}$. In addition, we adopt $\delta=100$ in \polyfit{} for the experiments with two keys in Problem \ref{prob:rel_error} ($Q_{rel}$).

We then discuss the methods that are unable to fulifll the deterministic error guarantee (i.e., the methods with $\times$ in the $Q_{abs}$ and $Q_{rel}$ columns in Table \ref{tab:methods}). Hist \cite{to2013entropy} adopts the entropy-based histogram for answering the \verb"COUNT" query. The S-tree prebuilds the STX B-tree \cite{STXBtree} on top of a sampled subset of each dataset. S2 \cite{haas1992sequential} and VerdictDB \cite{VerdictDB18} are sampling-based approaches that can only provide probabilistic error guarantees. By default, we set the probability to 0.9 in our experiments.
Both DBest \cite{QP19} and PLATO \cite{Plato20} are the state-of-the-art methods in approximate query processing and time series databases, respectively, that can be also adapted to answer approximate range aggregate queries. Since these methods cannot provide deterministic error guarantees, we regard them as heuristic methods.


We implemented all methods in C++ and conducted experiments on an Intel Core i7-8700 3.2GHz PC using WSL (Windows 10 Subsystem for Linux).

\vspace{-2mm}
\subsection{\polyfit{} Tuning}
\label{sec:investigation_polyfit}
\vspace{-1mm}
In this section, we investigate two research questions for \polyfit{}, namely (1) how does the degree $deg$ affect the query response time of \polyfit{}? (2) how does the degree $deg$ affect the construction time of \polyfit{}?

\vspace{-1mm}
\subsubsection{{\bf Effect of $deg$ on the query response time}}
\label{sec:exp_polyfit_analysis}
\vspace{-1mm}
Recall that we need to select the degree $deg$ in order to build \polyfit{}. It is thus important to understand how this parameter affects the query response time. Here, we use the form \polyfit{}-$deg$ to represent the degree $deg$ of \polyfit{}. Figure \ref{fig:vary_deg} shows the trends for the query response time for both \verb"COUNT" (one key and two keys) and \verb"MAX" (one key) queries, using the absolute error threshold $\varepsilon_{abs}=100$. When we choose a larger degree $deg$, the polynomial function can provide better approximation for $F(k)$, and thus reduce the index size, which can reduce the response time for each query. However, the larger the degree $deg$, the larger the computation time for each node in \polyfit{}. Therefore, we can find that the response time increases (e.g., $deg=3$ and 4 in Figure \ref{fig:vary_deg}a), once we utilize a high degree $deg$. By default, in subsequent experiments, we choose deg = 2 for the COUNT query with a single key, and deg = 3 for the COUNT query with two keys and for the MAX query.


\begin{figure*}[!hbt]
\vspace{-3mm}
\centering
\begin{tabular}{c c c}
\hspace{-2mm}
\includegraphics[width=0.50\columnwidth]{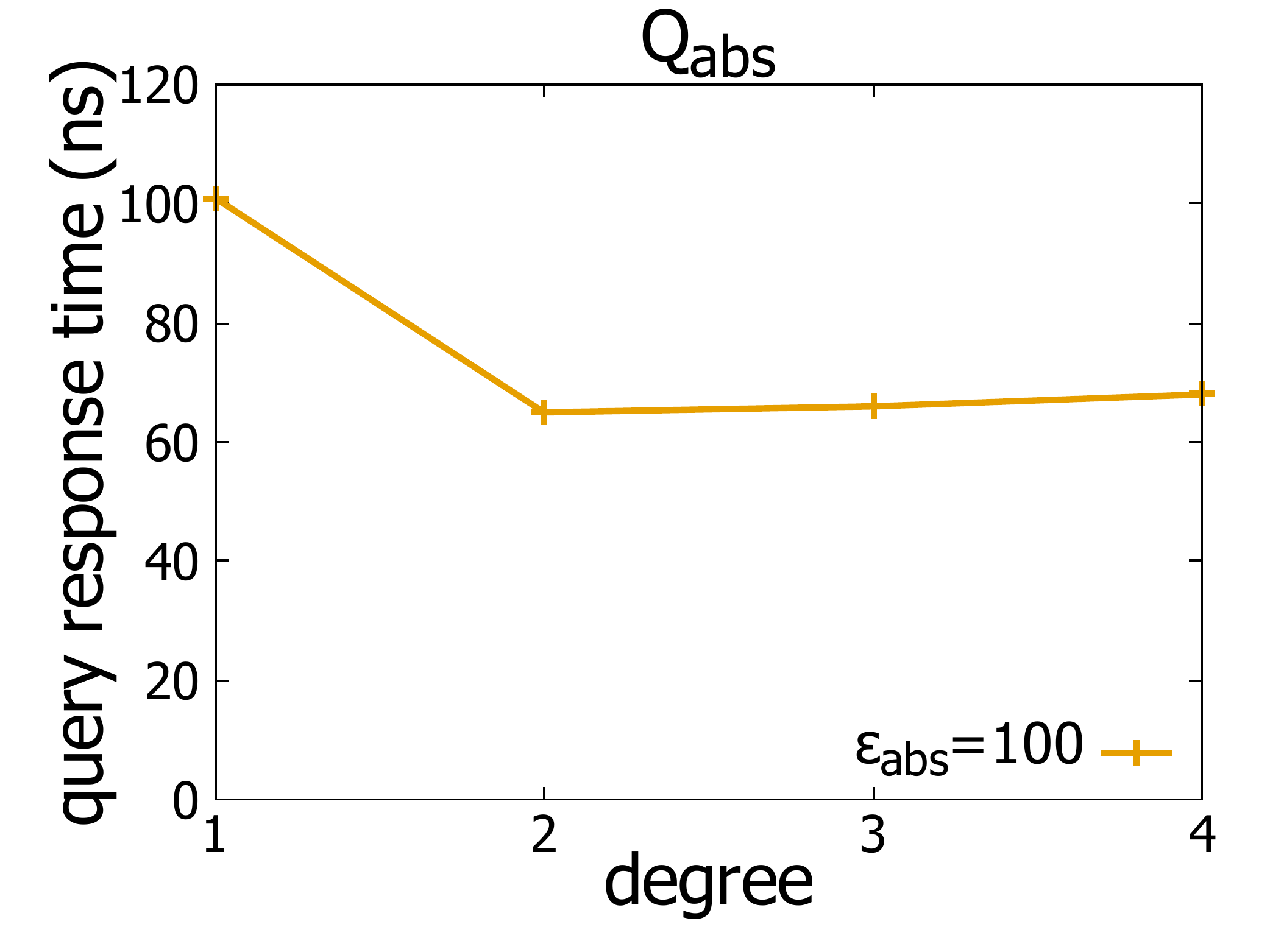} &
\hspace{-2mm}
\includegraphics[width=0.50\columnwidth]{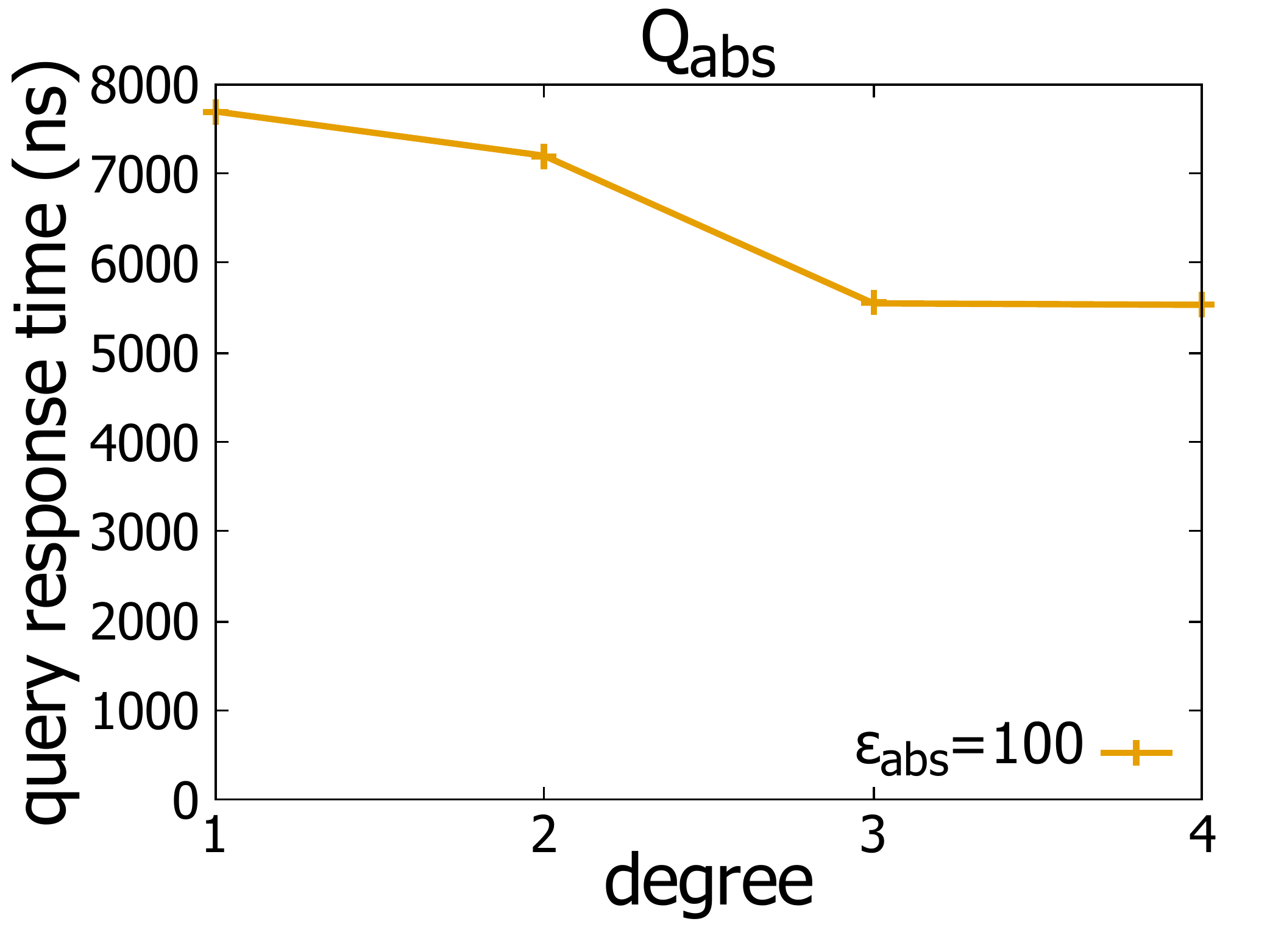} &
\hspace{-2mm}
\includegraphics[width=0.50\columnwidth]{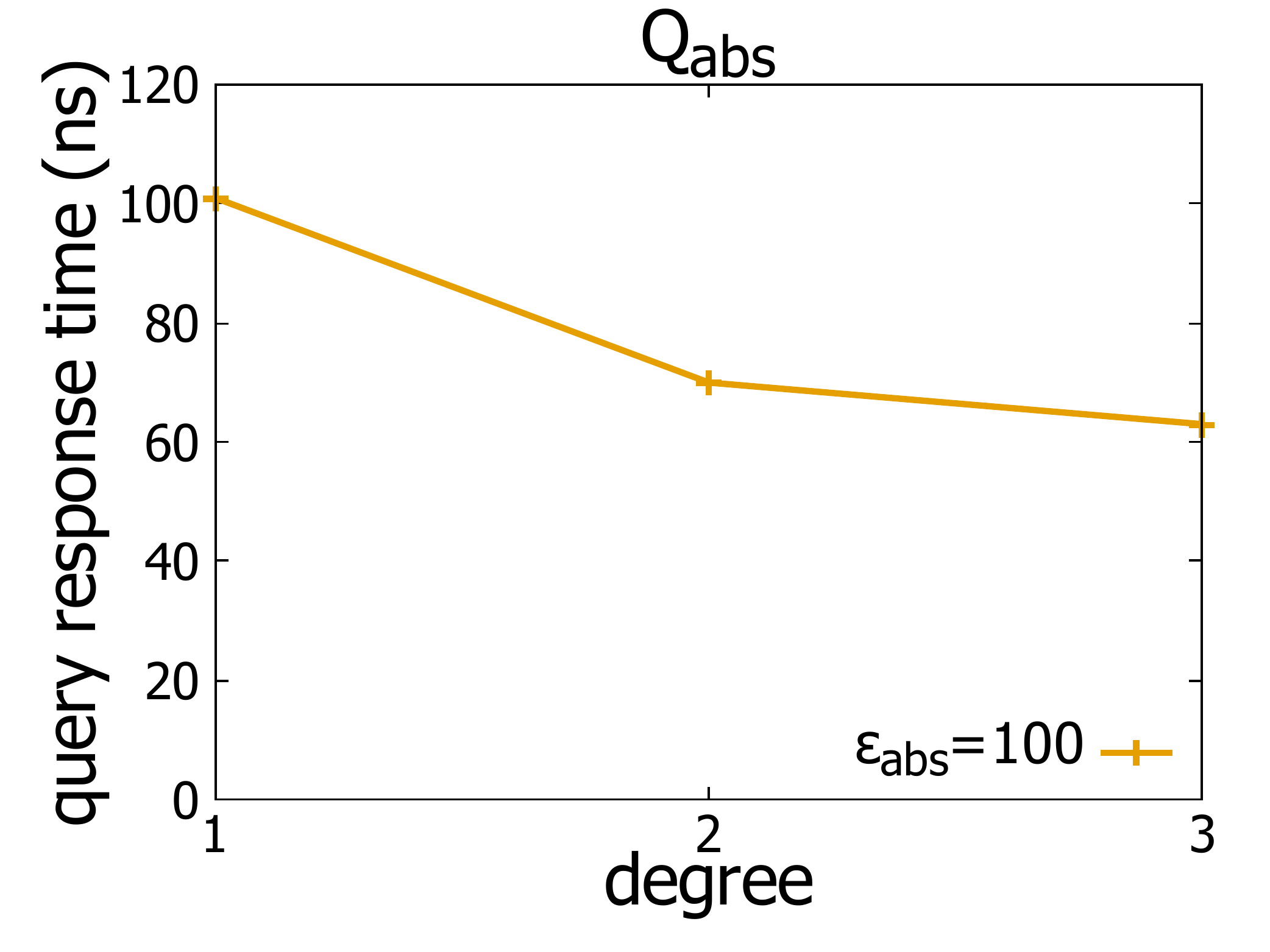}\\
(a) \verb"COUNT" query (single key) & (b) \verb"COUNT" query (two keys) & (c) \verb"MAX" query (single key)\\
\end{tabular}
\vspace{-4.0mm}
\caption{Running time for \texttt{COUNT} (single key), \texttt{COUNT} (two keys), and \texttt{MAX} queries on TWEET, OSM, and HKI datasets, respectively, varying the degree $deg$ of \polyfit{}}
\vspace{-1.5mm}
\label{fig:vary_deg}
\end{figure*}

\vspace{-1mm}
\subsubsection{{\bf Effect of $deg$ on the construction time}}
\label{sec:exp_construction_time}
\vspace{-1mm}
We further examine the construction time for \polyfit{}, varying the highest degree $deg$ from 1 to 4 in the polynomial function (cf. Figure \ref{fig:construction_time_vary_deg}). Since a polynomial function with a higher degree can produce error guarantee for a longer interval $I$, i.e., $E(I)\leq \delta$, the GS method needs to call the LP solver with longer intervals (cf. line 4 in Algorithm \ref{alg:GS}), which can increase the construction time when using polynomial functions with higher degree $deg$. 

\begin{figure}[!hbt]
	\centering
	\vspace{-3mm}
	\begin{tabular}{c c}
		\hspace{-3mm}
		\includegraphics[width=0.5\columnwidth]{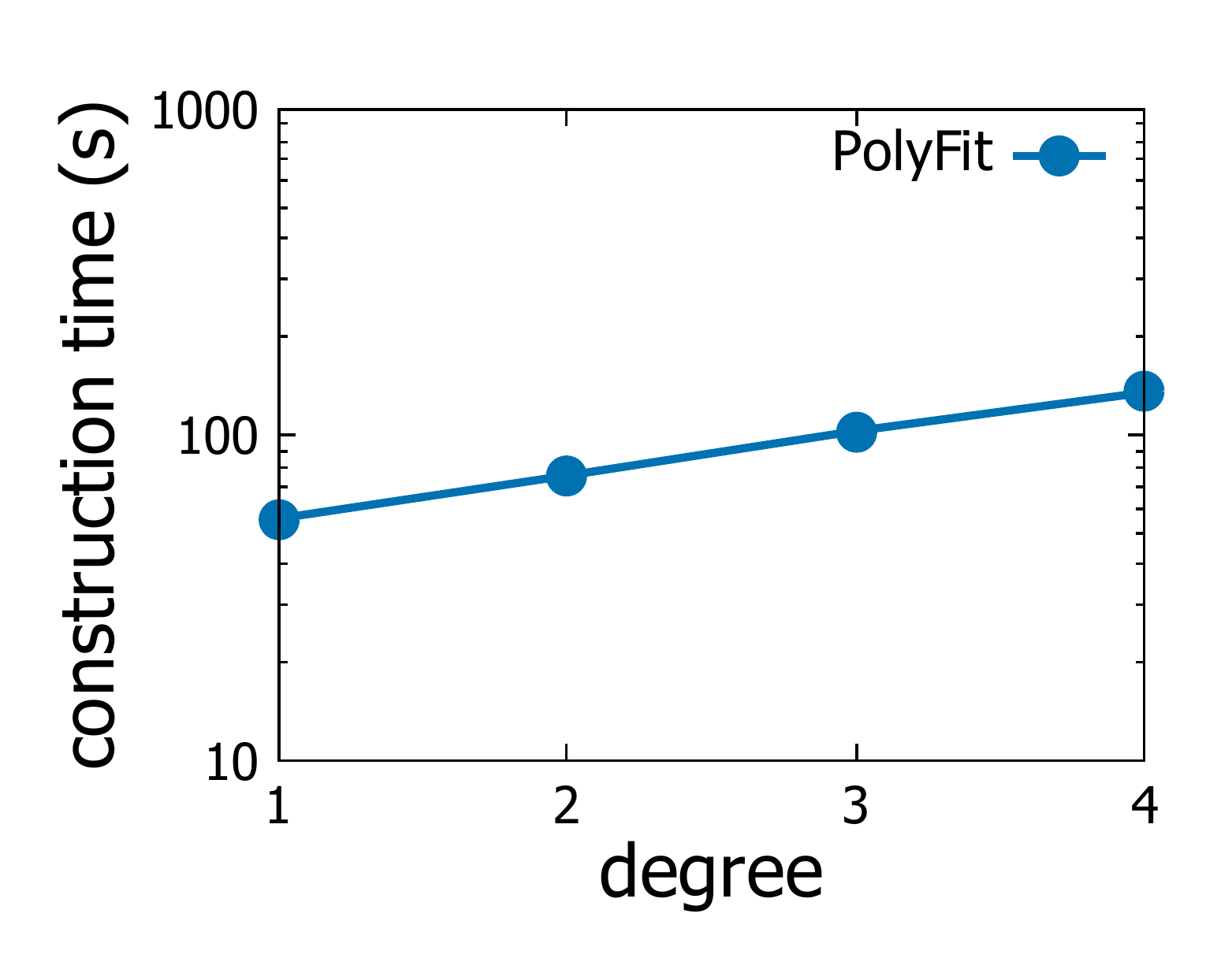} &
		\hspace{-3mm}
		\includegraphics[width=0.5\columnwidth]{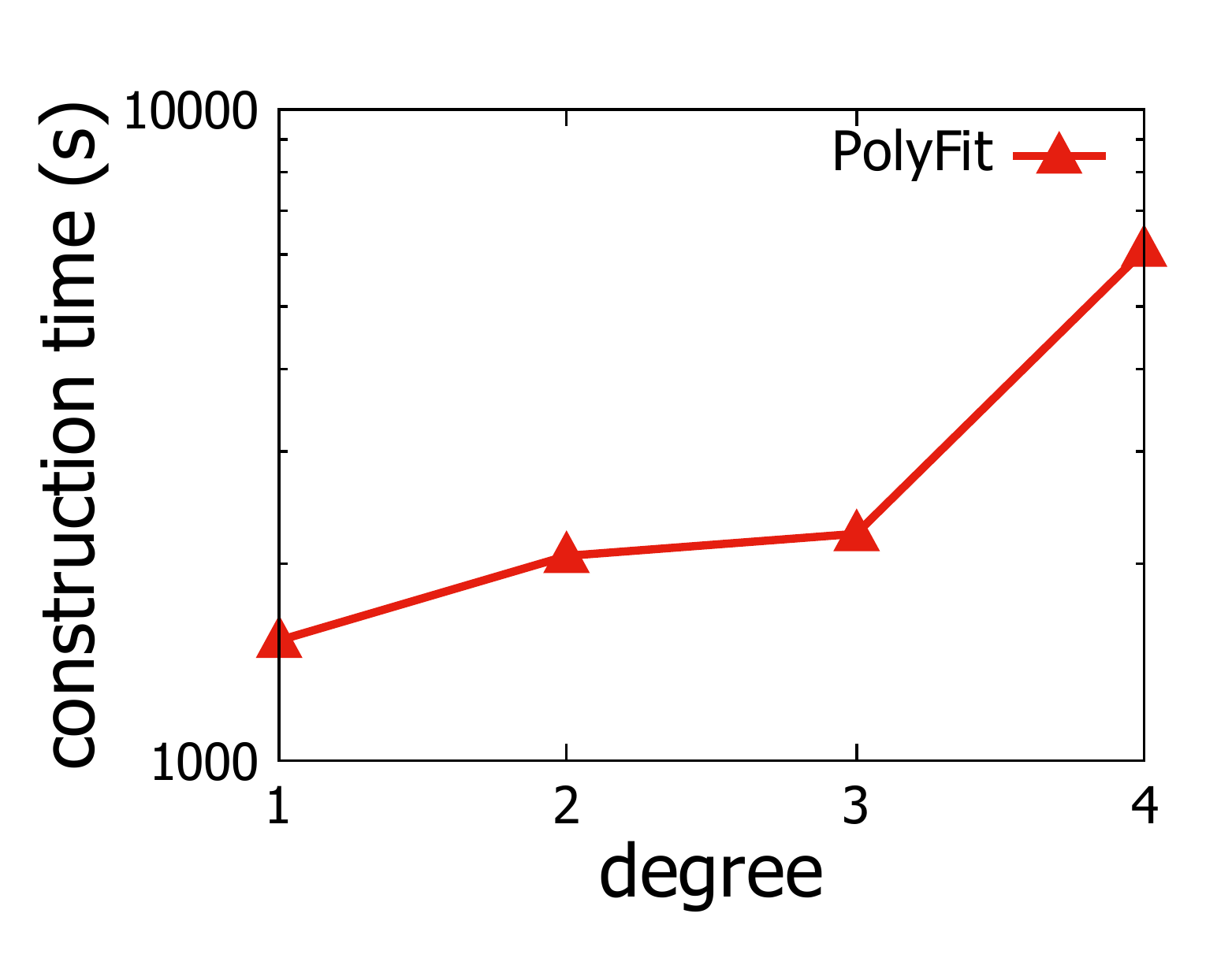} \\
		(a) \verb"COUNT" query (single key) & (b) \verb"COUNT" query (two keys)
	\end{tabular}
	\vspace{-3.5mm}
	\caption{Index construction time of \polyfit{} for \texttt{COUNT} query with single key (using TWEET dataset) and two keys (using OSM dataset), varying the degree $deg$}
	\vspace{-0.5mm}
	\label{fig:construction_time_vary_deg}
\end{figure}

\vspace{-2mm}
\subsection{Comparing with Error-Bounded Methods}
\label{exp:response_time_guarantee_error}
\vspace{-1mm}
In this section, we test the response time of the different methods that can fulfill the absolute and relative error guarantees. Here, we adopt the default settings for these methods (cf. Section \ref{exp:setting}) and use the datasets HKI, TWEET, and OSM for testing the performance of \verb"COUNT" (single key), \verb"MAX" (single key), and \verb"COUNT" (two keys) queries, respectively. For Problem \ref{prob:abs_error} ($Q_{abs}$), we fix the absolute error $\varepsilon_{abs}=100$ and $\varepsilon_{abs}=200$ for the experiments with one key and two keys, respectively. For Problem \ref{prob:rel_error} ($Q_{rel}$), we fix the relative error $\varepsilon_{rel}=0.01$. Table \ref{tab:all_methods_error_guarantee} shows the response time of different methods. Observe that \polyfit{} achieves the best performance for all the types of queries. For the \verb"COUNT" query with two keys, \polyfit{} can achieve a speedup of at least two orders of magnitude over the existing methods.

\vspace{-3mm}
\begin{table}[hbt]
\caption{Response time (nanoseconds) for all methods with error guarantees}
\label{tab:all_methods_error_guarantee}
\vspace{-3mm}
\begin{tabular}{|@{ }c@{ }|@{ }c@{ }|@{ }c@{ }|@{ }c@{ }|@{ }c@{ }|@{ }c@{ }|@{ }c@{ }|} \hline
	Error guarantee & \multicolumn{3}{|c|}{$Q_{abs}$} & \multicolumn{3}{|c|}{$Q_{rel}$} \\ \hline
	Query type& \verb"COUNT" & \verb"MAX" & \verb"COUNT" & \verb"COUNT" & \verb"MAX" & \verb"COUNT" \\ \hline
	\# of keys & 1 & 1 & 2 & 1 & 1 & 2 \\ \hline
	\hline
	aR-tree & 590 & 3592 & 357457 & 590 & 3592 & 357457 \\ \hline
	MRTree & 565 & 182 & 385391 & 335 & 138 & 98919 \\ \hline
	RMI & 568 & n/a & n/a & 579 & n/a & n/a \\ \hline
	FITing-tree & 135 & n/a & n/a & 147 & n/a & n/a \\ \hline
	PGM & 104 & n/a & n/a & 118 & n/a & n/a \\ \hline
	Polyfit & {\bf 68} & {\bf 63} & {\bf 5274} & {\bf  79} & {\bf 65} & {\bf 5299} \\ \hline
\end{tabular}
\end{table}
\vspace{-3mm}

{\bf Sensitivity of} $\varepsilon_{abs}$ {\bf for \verb"COUNT" query.} We investigate how the absolute error $\varepsilon_{abs}$ affects the response times of different methods. For the \verb"COUNT" query with single key, we choose five absolute error values for testing, which are 100, 200, 400, 1000, and 2000. Observe from Figure \ref{fig:count_varying_abs}a that since \polyfit{}, FITing-tree, and PGM can provide more compact index structures for the datasets, these methods can significantly improve the efficiency, compared with the traditional index structures, i.e., the aR-tree and the MRTree. In addition, due to the better approximation with nonlinear polynomial functions ($deg=2$), \polyfit{} can achieve 1.33x to 6x speedups, over the existing learned-index structures, including RMI, FITing-tree, and PGM. For the \verb"COUNT" query with two keys, we choose 200, 400, 800, 2000, and 4000 as the absolute error values for testing. Since the state-of-the-art learned index structures (RMI, FITing-tree, and PGM) can only support queries with a single key, we omit these methods in this experiment. Figure \ref{fig:count_varying_abs}b shows that \polyfit{} achieves at least one order of magnitude speedups compared with the existing methods (aR-tree and MRTree), which is due to its compact index structure and query processing method.

\begin{figure}[!hbt]
\vspace{-3.5mm}
\begin{tabular}{c c}
\hspace{-3mm}
\includegraphics[width=0.5\columnwidth]{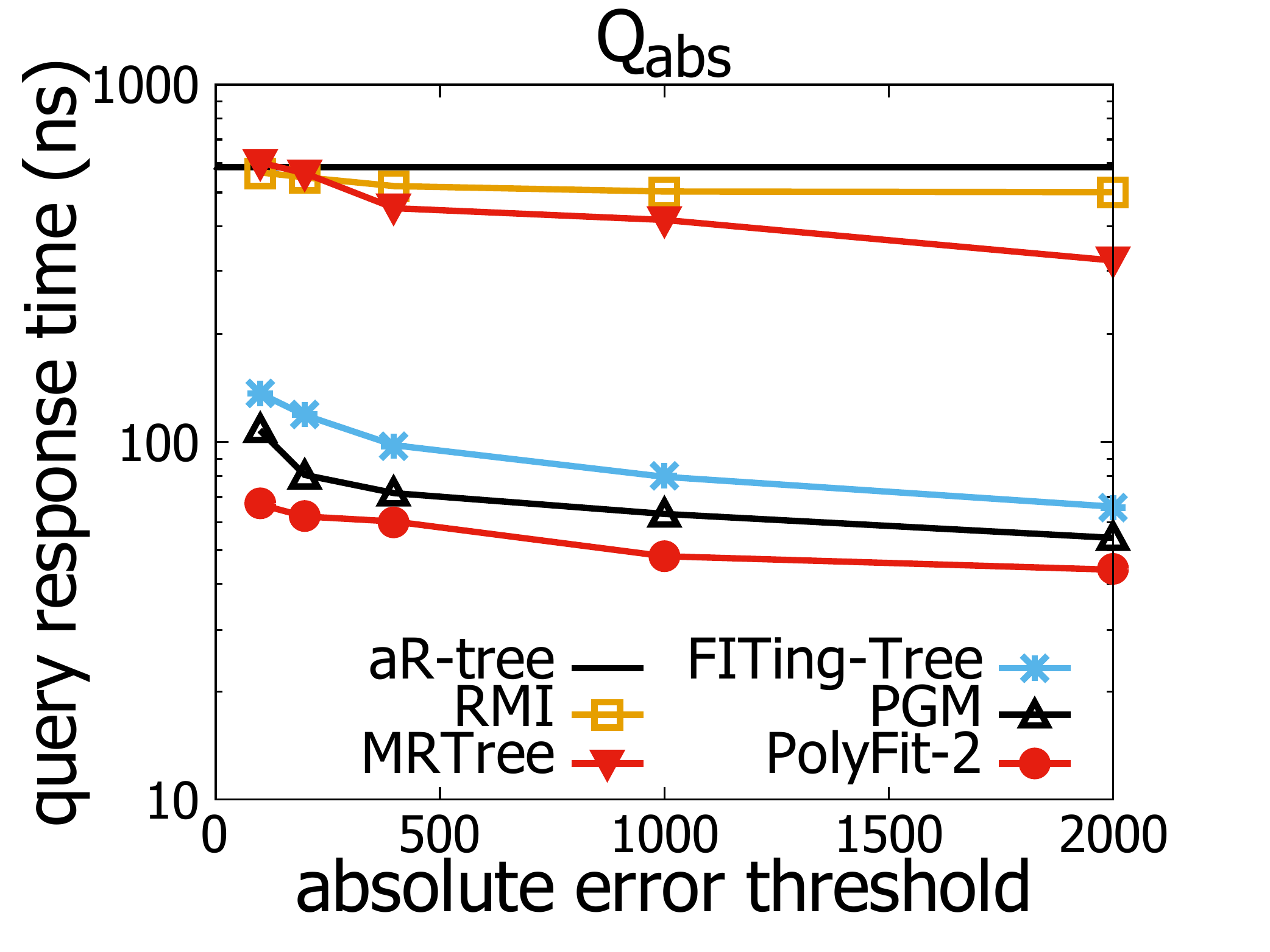} &
\hspace{-3mm}
\includegraphics[width=0.5\columnwidth]{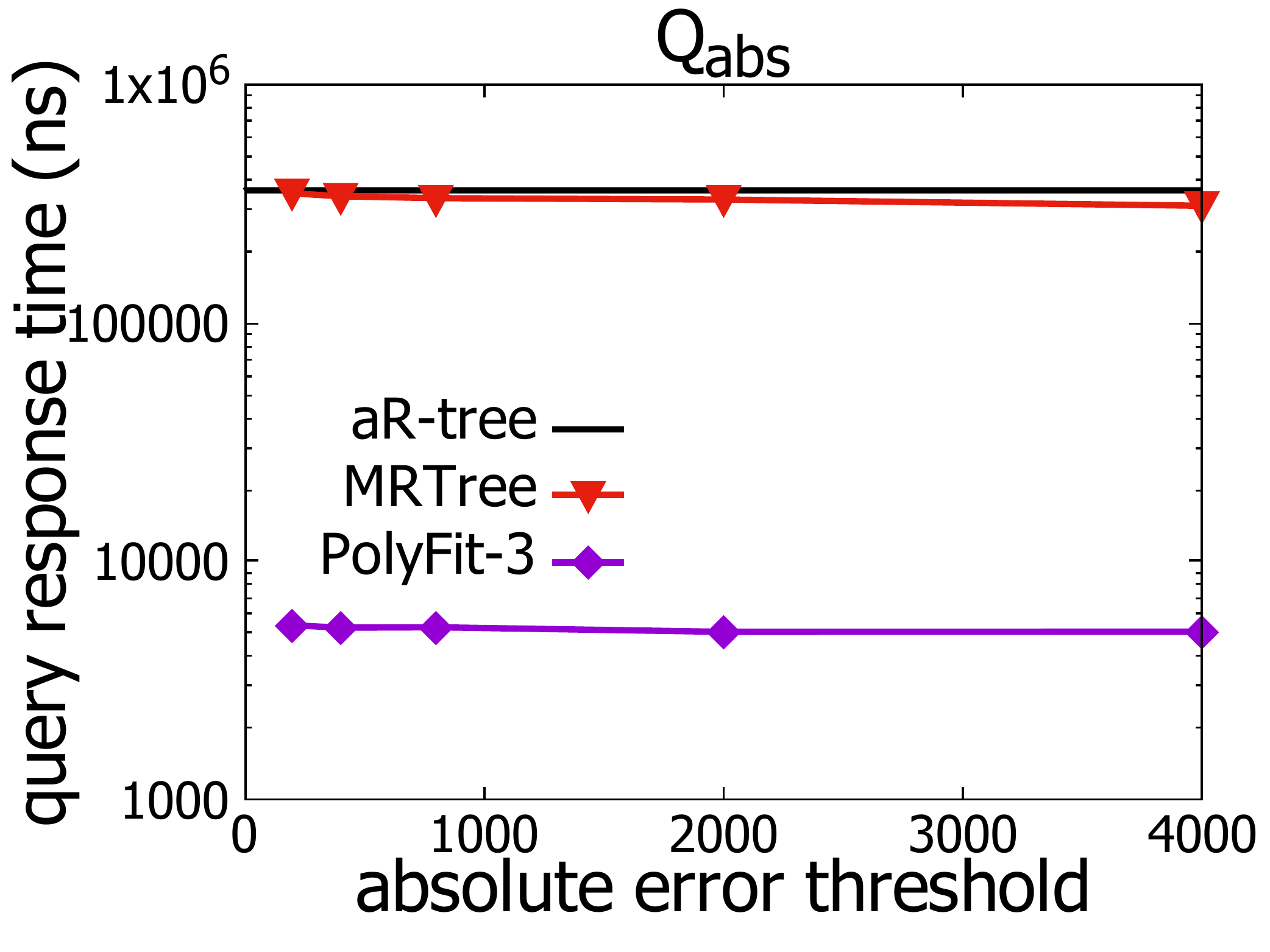}\\
(a) \verb"COUNT" query (single key) & (b) \verb"COUNT" query (two keys)
\end{tabular}
\vspace{-3.5mm}
\caption{Response time for \texttt{COUNT} query in TWEET dataset (for single key) and OSM dataset (for two keys), varying the absolute error $\varepsilon_{abs}$}
\vspace{-1.5mm}
\label{fig:count_varying_abs}
\end{figure}

{\bf Sensitivity of} $\varepsilon_{rel}$ {\bf for \verb"COUNT" query.} We proceed to test how the relative error $\varepsilon_{rel}$ affects the response time of the different methods. In this experiment, we choose five relative error values, which are 0.005, 0.01, 0.05, 0.1, and 0.2. Based on the more compact index structure, \polyfit{} is able to achieve better performance, compared with the existing methods (cf. Figure \ref{fig:count_varying_rel}a). For the \verb"COUNT" query with two keys, \polyfit{} significantly outperforms the existing methods, i.e., the aR-tree and the MRTree, by at least one order of magnitude (cf. Figure \ref{fig:count_varying_rel}b).

\begin{figure}[!hbt]
\vspace{-2mm}
\begin{tabular}{c c}
\hspace{-3mm}
\includegraphics[width=0.5\columnwidth]{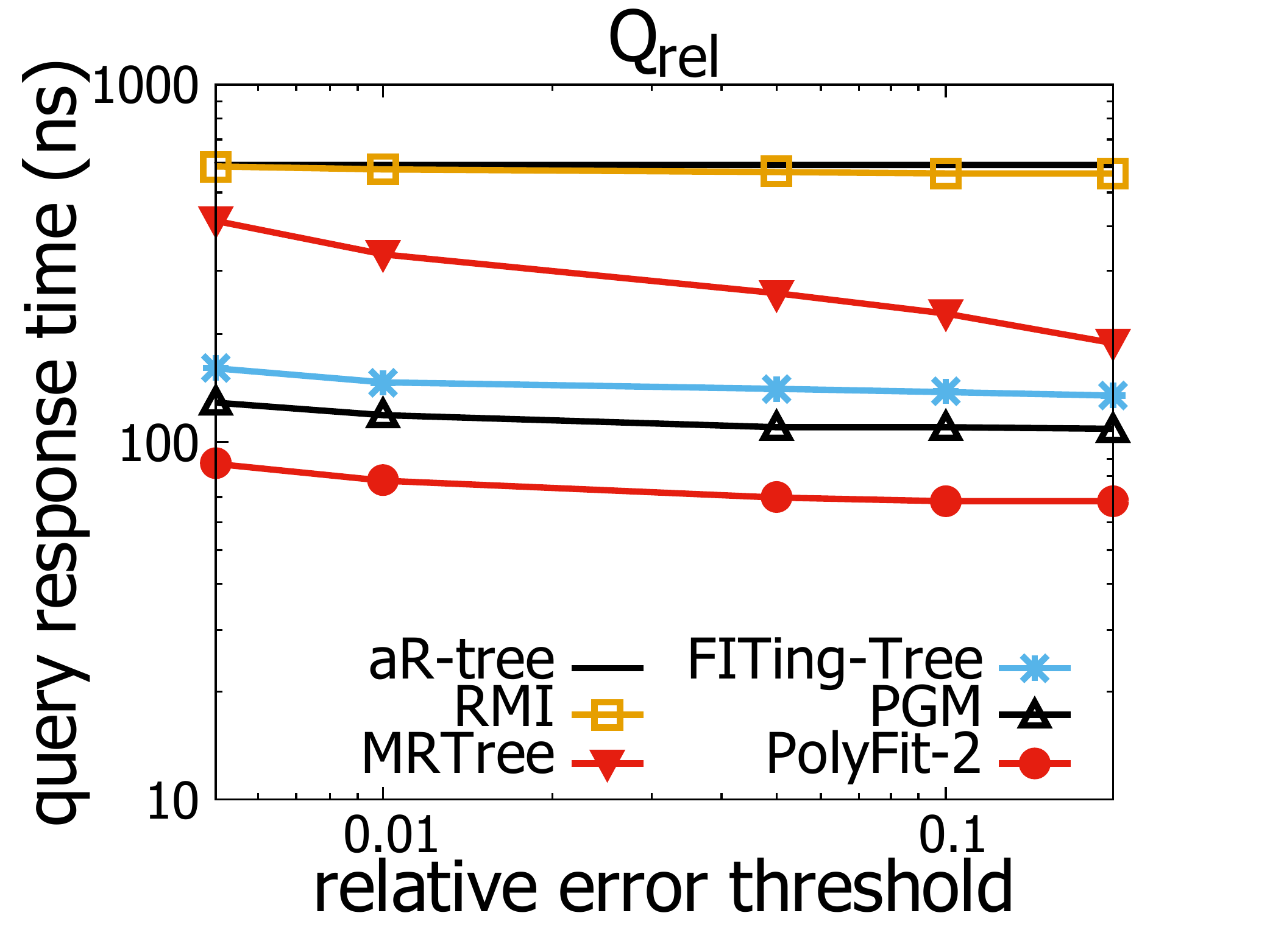} &
\hspace{-3mm}
\includegraphics[width=0.5\columnwidth]{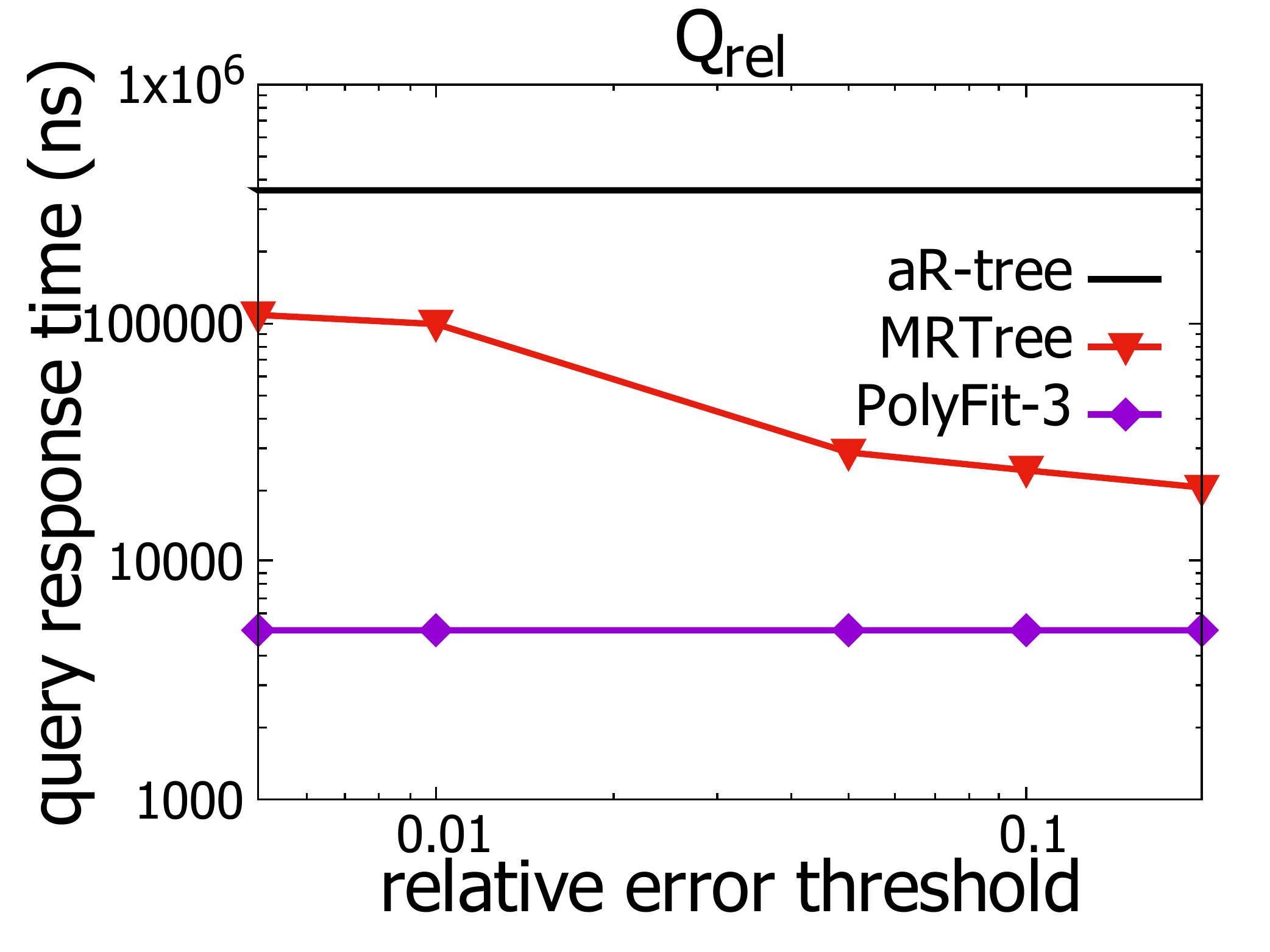} \\
(a) \verb"COUNT" query (single key) & (b) \verb"COUNT" query (two keys)  \\
\end{tabular}
\vspace{-3.5mm}
\caption{Response time for \texttt{COUNT} query in TWEET dataset (for single key) and OSM dataset (for two keys), varying the relative error $\varepsilon_{rel}$}
\vspace{-2.5mm}
\label{fig:count_varying_rel}
\end{figure}

\begin{figure}[!hbt]
\centering
\vspace{-1.5mm}
\begin{tabular}{c c}
\hspace{-3mm}
\includegraphics[width=0.5\columnwidth]{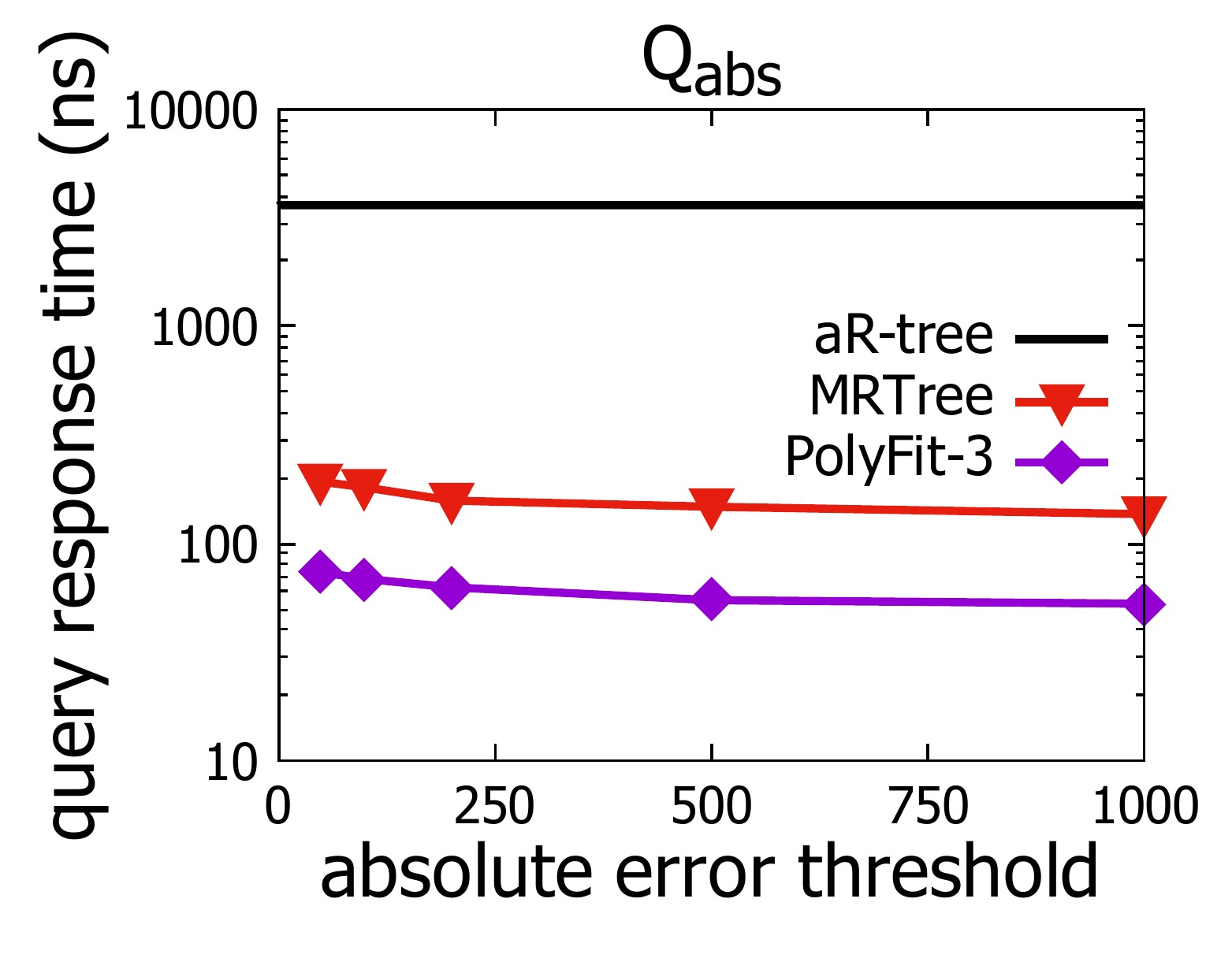} &
\hspace{-3mm}
\includegraphics[width=0.5\columnwidth]{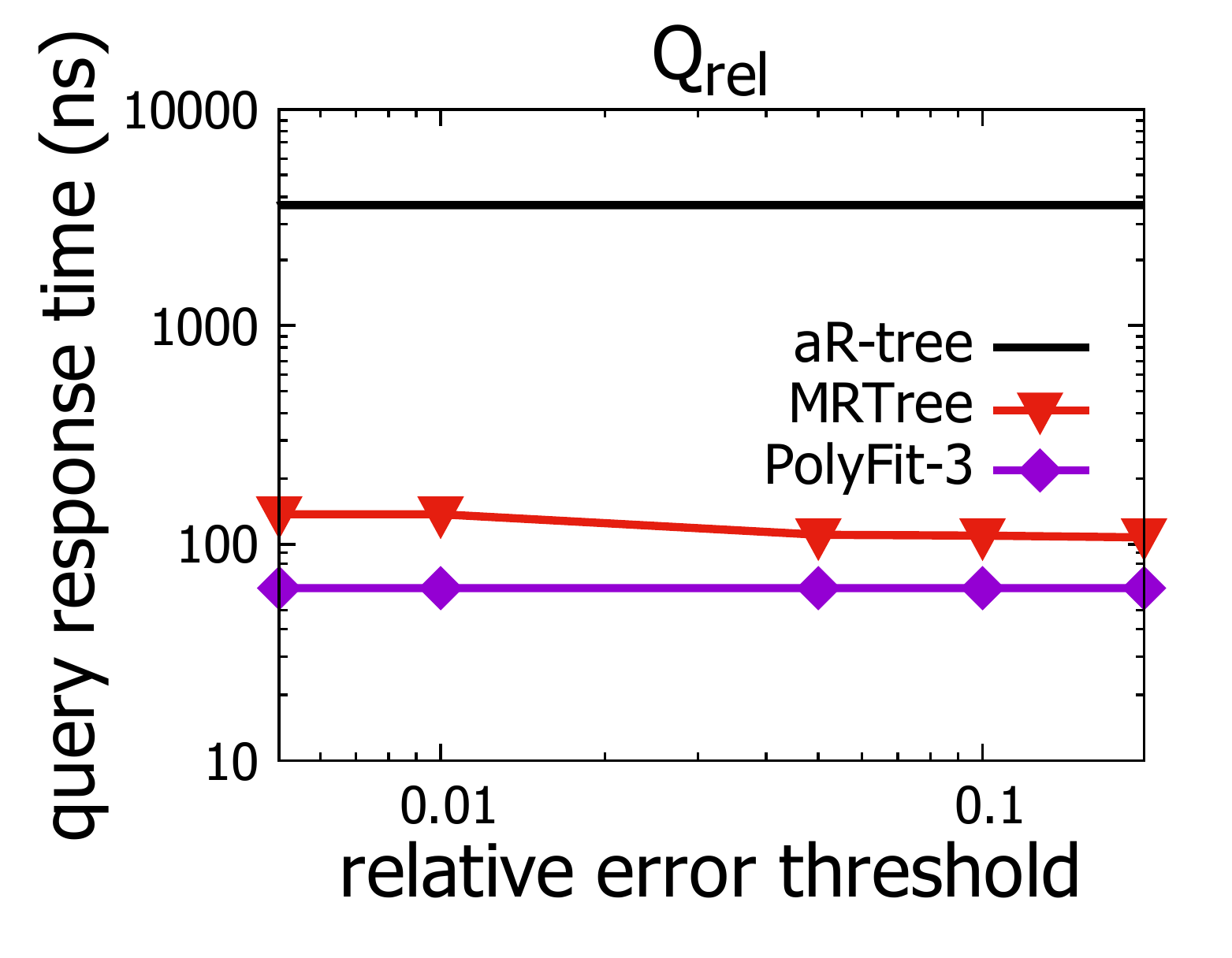} \\
(a) \verb"MAX" query, varying $\varepsilon_{abs}$ & (b) \verb"MAX" query, varying $\varepsilon_{rel}$
\end{tabular}
\vspace{-3.5mm}
\caption{Response time for \texttt{MAX} query in HKI dataset}
\vspace{-3.5mm}
\label{fig:max_varying_error}
\end{figure}

\begin{figure}[!hbt]
\centering
\vspace{-3.5mm}
\begin{tabular}{c c}
\hspace{-3mm}
\includegraphics[width=0.5\columnwidth]{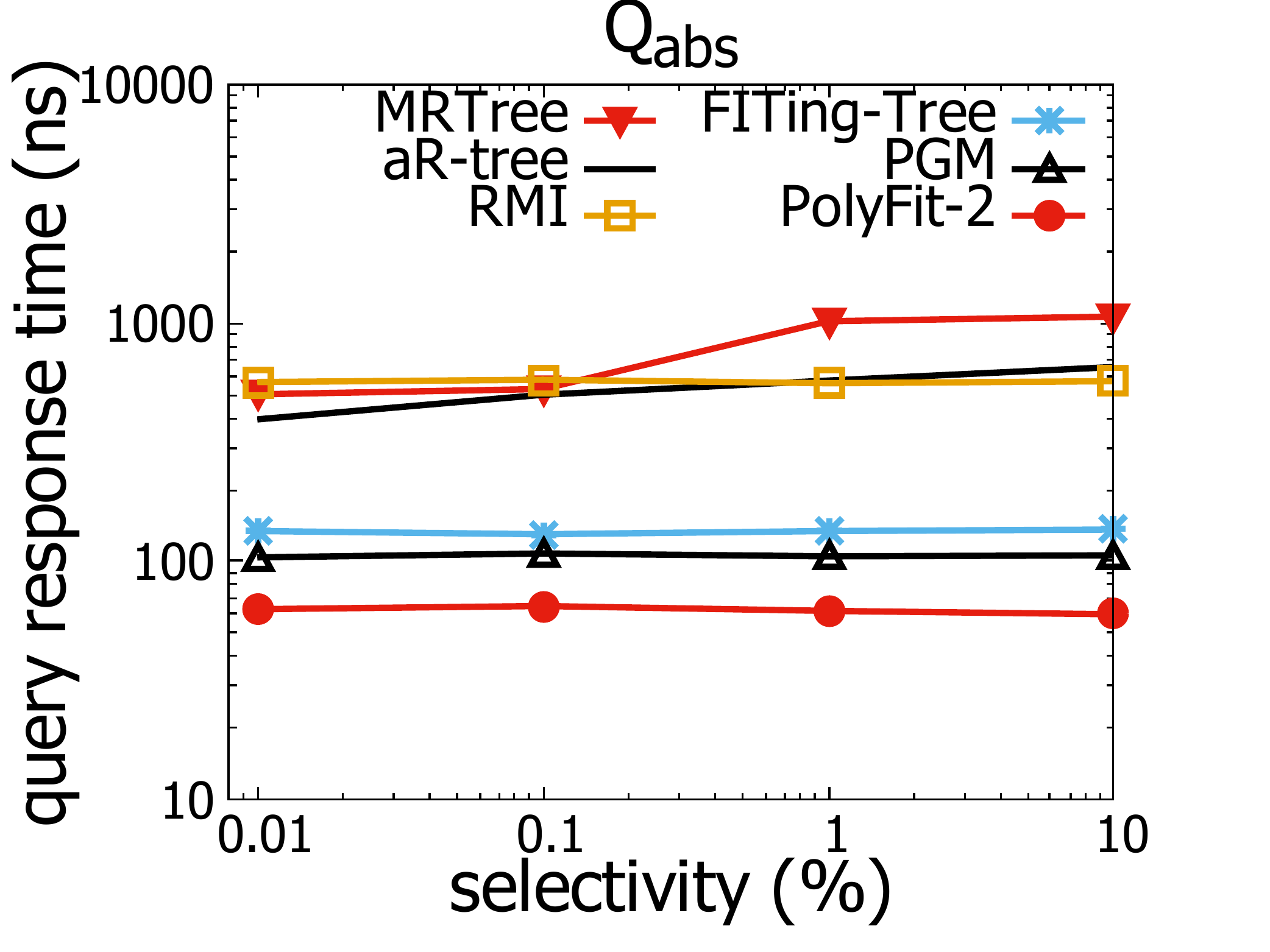} &
\hspace{-3mm}
\includegraphics[width=0.5\columnwidth]{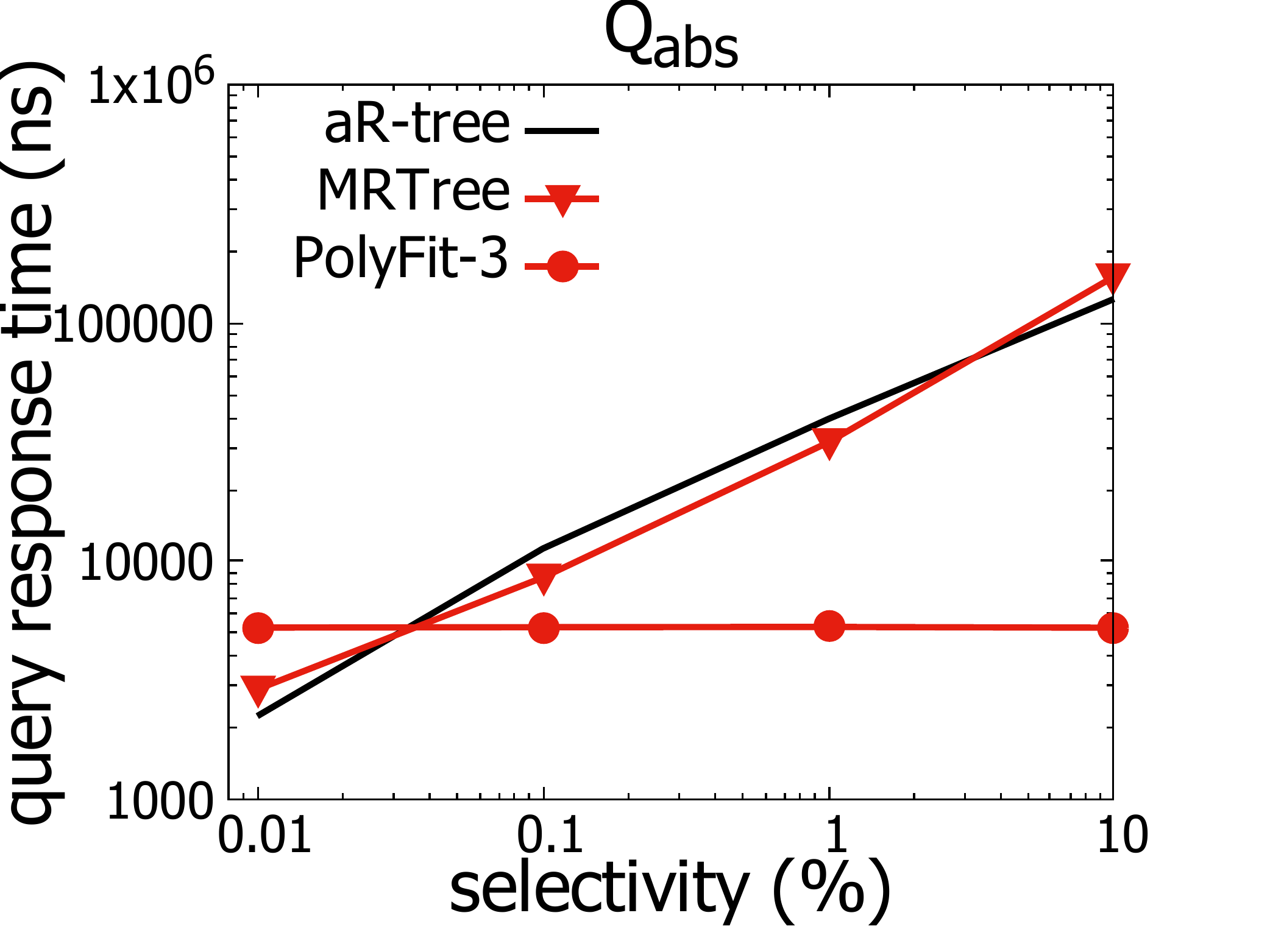} \\
(a) \verb"COUNT" query (single key) & (b) \verb"COUNT" query (two keys)
\end{tabular}
\vspace{-3.5mm}
\caption{Response time for \texttt{COUNT} query in TWEET dataset (for single key) and OSM dataset (for two keys), varying the selectivity of the query}
\vspace{-1.0mm}
\label{fig:selectivity}
\end{figure}

\begin{figure}[!hbt]
	\centering
	\vspace{-3.5mm}
	\begin{tabular}{c c}
		\hspace{-3mm}
		\includegraphics[width=0.5\columnwidth]{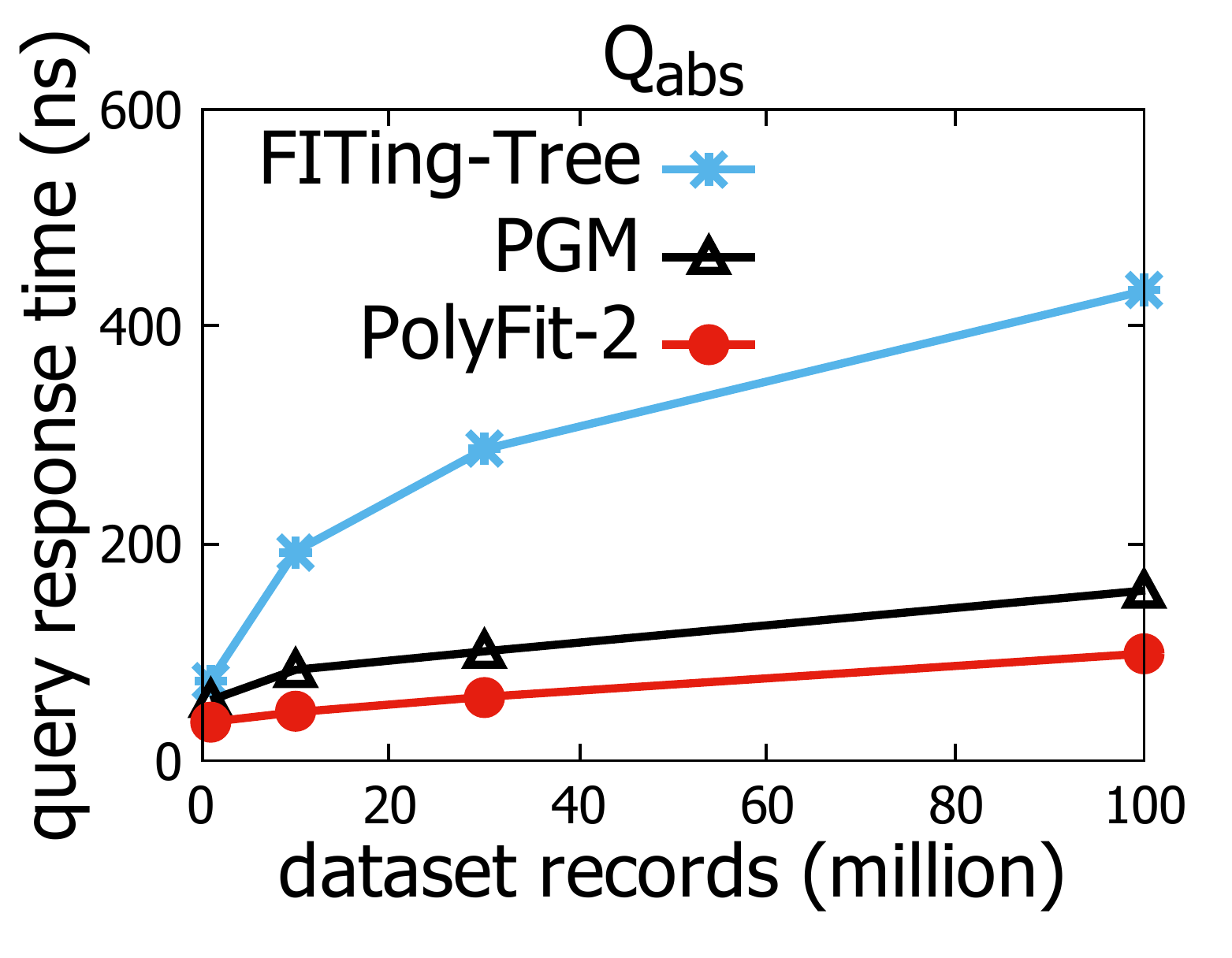} &
		\hspace{-3mm}
		\includegraphics[width=0.5\columnwidth]{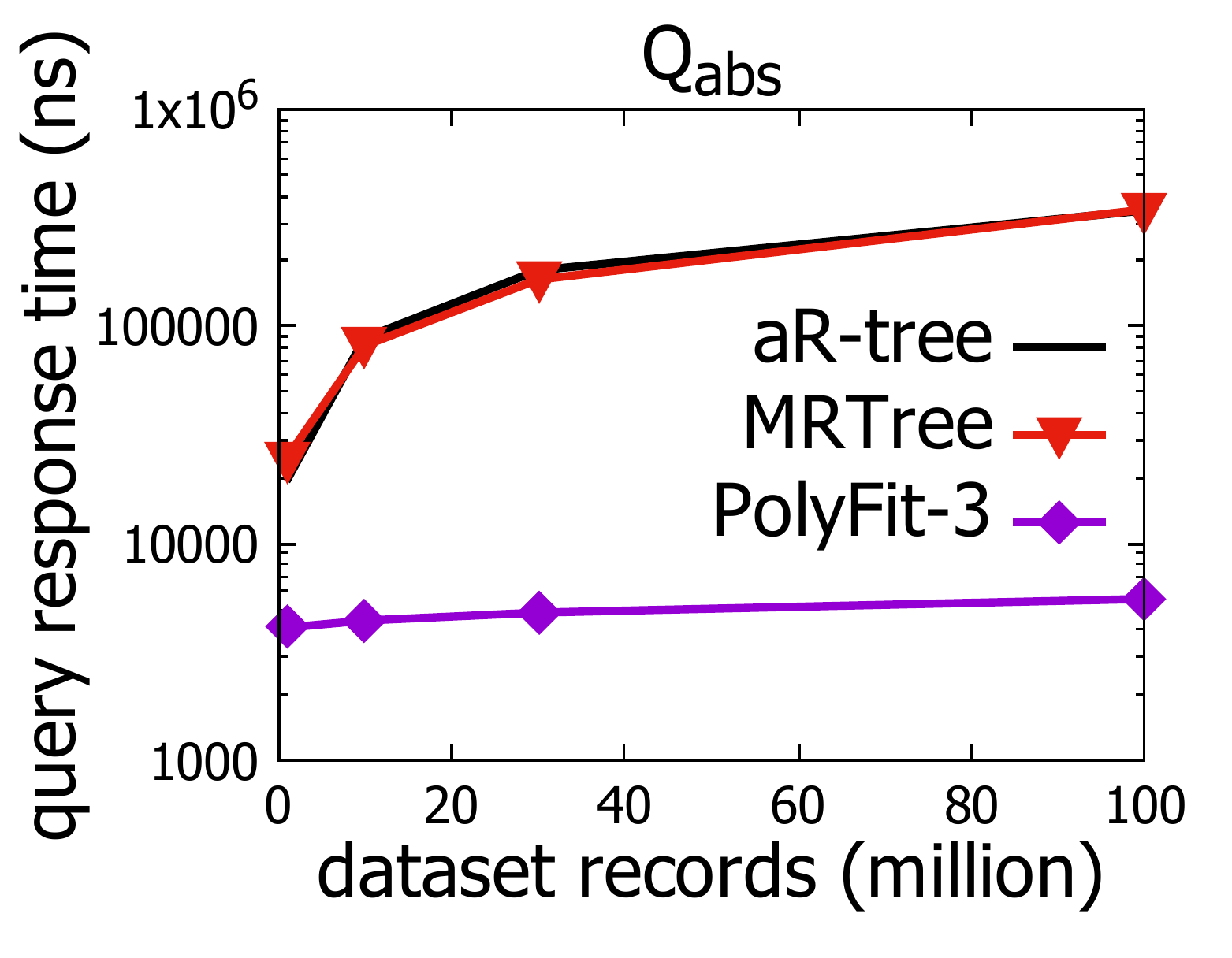} \\
		(a) \verb"COUNT" query (single key) & (b) \verb"COUNT" query (two keys)
	\end{tabular}
	\vspace{-3.5mm}
	\caption{Response time for \texttt{COUNT} query in OSM dataset, varying the dataset size}
	\vspace{-1.0mm}
	\label{fig:count_varying_data_size}
\end{figure}

\begin{figure}[!hbt]
\centering
\vspace{-3.5mm}
\includegraphics[width=1.0\columnwidth]{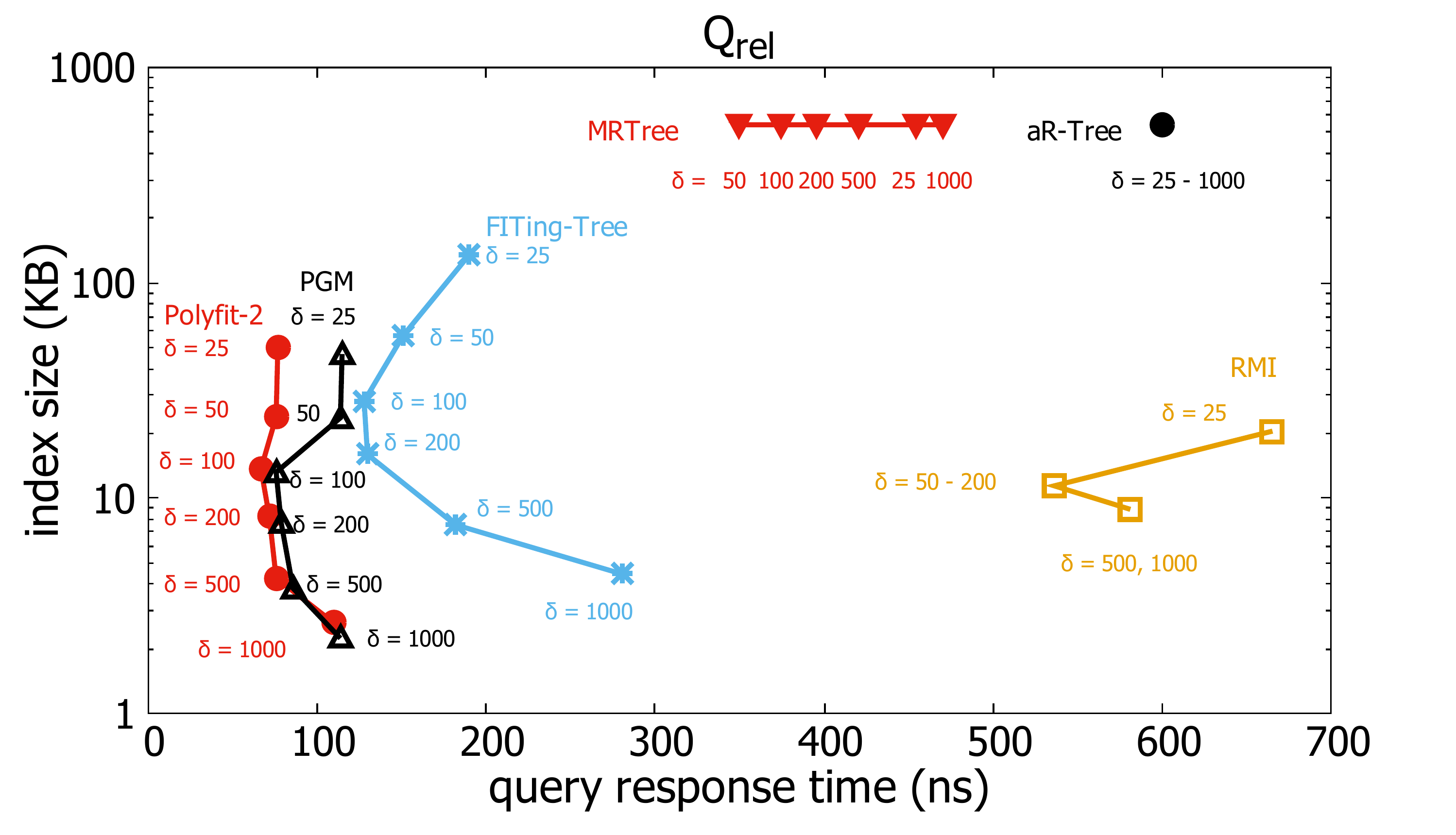}
\vspace{-8mm}
\caption{Trade-off between the query response time and index size of \texttt{COUNT} query (single key) in TWEET dataset, for $Q_{rel}$ with $\varepsilon_{rel}= 0.01$, varying $\delta$ from $25$ to $1000$}
\vspace{-5mm}
\label{fig:time_vs_index_size}
\end{figure}

{\bf Sensitivity of} $\varepsilon_{abs}$ {\bf and} $\varepsilon_{rel}$ {\bf for \verb"MAX" query.} In this experiment, we proceed to investigate how the absolute error $\varepsilon_{abs}$ and relative error $\varepsilon_{rel}$ affect the efficiency performance of different methods. Observe from Figure \ref{fig:max_varying_error}, \polyfit{} can achieve at least 2x speedup, compared with other methods, even though the selected error is small. 

{\bf Sensitivity of the selectivity for \verb"COUNT" query.} We further test the response time of the different methods, varying the selectivity of the \verb"COUNT" query. Figure \ref{fig:selectivity} shows that when we increase the selectivity of the \verb"COUNT" query (i.e., each query covers a larger region), the query response time normally increases in different methods. Since all methods for the \verb"COUNT" query with a single key have logarithmic time complexity, they are not sensitive to the selectivity (cf. Figure \ref{fig:selectivity}a). Unlike the single key case, both the existing methods aR-tree and MRTree are sensitive to the selectivity, compared with \polyfit{} (cf. Figure \ref{fig:selectivity}b).

In both cases, \polyfit{} achieves better performance across different selectivities. Since the methods MRTree, aR-tree, and RMI always provide inferior efficiency in the single key case (cf. Figures \ref{fig:count_varying_abs}a, \ref{fig:count_varying_rel}a and \ref{fig:selectivity}a), compared with FITing-Tree, PGM and \polyfit{}, we omit their results in subsequent experiments.

{\bf Scalability to the dataset size.} We proceed to test how the dataset size affects the efficiency of \polyfit{} and other methods. In this experiment, we choose the largest dataset OSM (with 100M records) for testing. Here, we focus on solving Problem \ref{prob:abs_error} ($Q_{abs}$) for \verb"COUNT" query, in which we adopt the default absolute errors, i.e., $\varepsilon_{abs}=100$ and $\varepsilon_{abs}=200$ for the cases in single key and two keys, respectively, and choose the latitude attribute as the key. To conduct this experiment, we choose five dataset sizes, which are 1M, 10M, 30M, and 100M. Figure \ref{fig:count_varying_data_size} shows that \polyfit{} scales well with the dataset size and outperforms other methods.


{\bf Trade-off between the query response time and index size.} We proceed to investigate the trade-off between the query response time and index size of the different indexing methods. To conduct this experiment, we focus on Problem \ref{prob:rel_error} ($Q_{rel}$) and choose 25, 50, 100, 200, 500, and 1000 as values of $\delta$ for testing. In Figure \ref{fig:time_vs_index_size}, since the changes to $\delta$ cannot affect the index construction methods of the aR-tree and MRTree, parameter $\delta$ cannot affect the index sizes of these two methods. We also notice that these index structures consistently provide inferior performance in terms of index size and query response time, compared with the FITing-tree, PGM, and the PolyFit methods. For the other methods, we can observe that the smaller the $\delta$, the larger the index size and query response time. The reason is that smaller $\delta$ values lead to more leaf nodes in the index structures in the different methods (e.g., more intervals are generated by the GS method (cf. Algorithm \ref{alg:GS}) in \polyfit{}). On the other hand, if $\delta$ is too large, it is easier for an online query to violate the error condition for $Q_{rel}$ (i.e., Lemma \ref{lem:count_rel_error}), and thus the query response time can also be larger. As such, all curves (except for the MRTree and aR-tree methods) in Figure \ref{fig:time_vs_index_size} resemble the ``C''-shape. In general, \polyfit{}-2 offers a better trade-off compared with other methods.

\vspace{-3.0mm}
\subsection{Comparing with Heuristic Methods}
\vspace{-1.0mm}
\label{exp:heuristics}
We compare the response time of \polyfit{} with other heuristic methods, which cannot fulfill deterministic error guarantees, i.e., $Q_{abs}$ (cf. Problem \ref{prob:abs_error}) and $Q_{rel}$ (cf. Problem \ref{prob:rel_error}). In this experiment, we adopt the default setting for the method PLATO \cite{Plato20}, vary the bin size for the method Hist and vary the sampling size for the sampling-based methods, including S-tree, S2, VerdictDB, and DBest. Since S2 cannot achieve less than 100000ns query response time with 10\% measured relative error, we omit the result of S2 in Figure \ref{fig:heuristics_error_time}a. In addition, we only report the results of the heuristic methods DBest and VerdictDB in Figure \ref{fig:heuristics_error_time}b, as the other heuristic methods cannot support \verb"COUNT" queries with two keys (cf. Table \ref{tab:methods}). In these two figures, \polyfit{} yields the smallest query response time with similar relative error.

\begin{figure}[!hbt]
\centering
\vspace{-3.5mm}
\begin{tabular}{c c}
\hspace{-3mm}
\includegraphics[width=0.5\columnwidth]{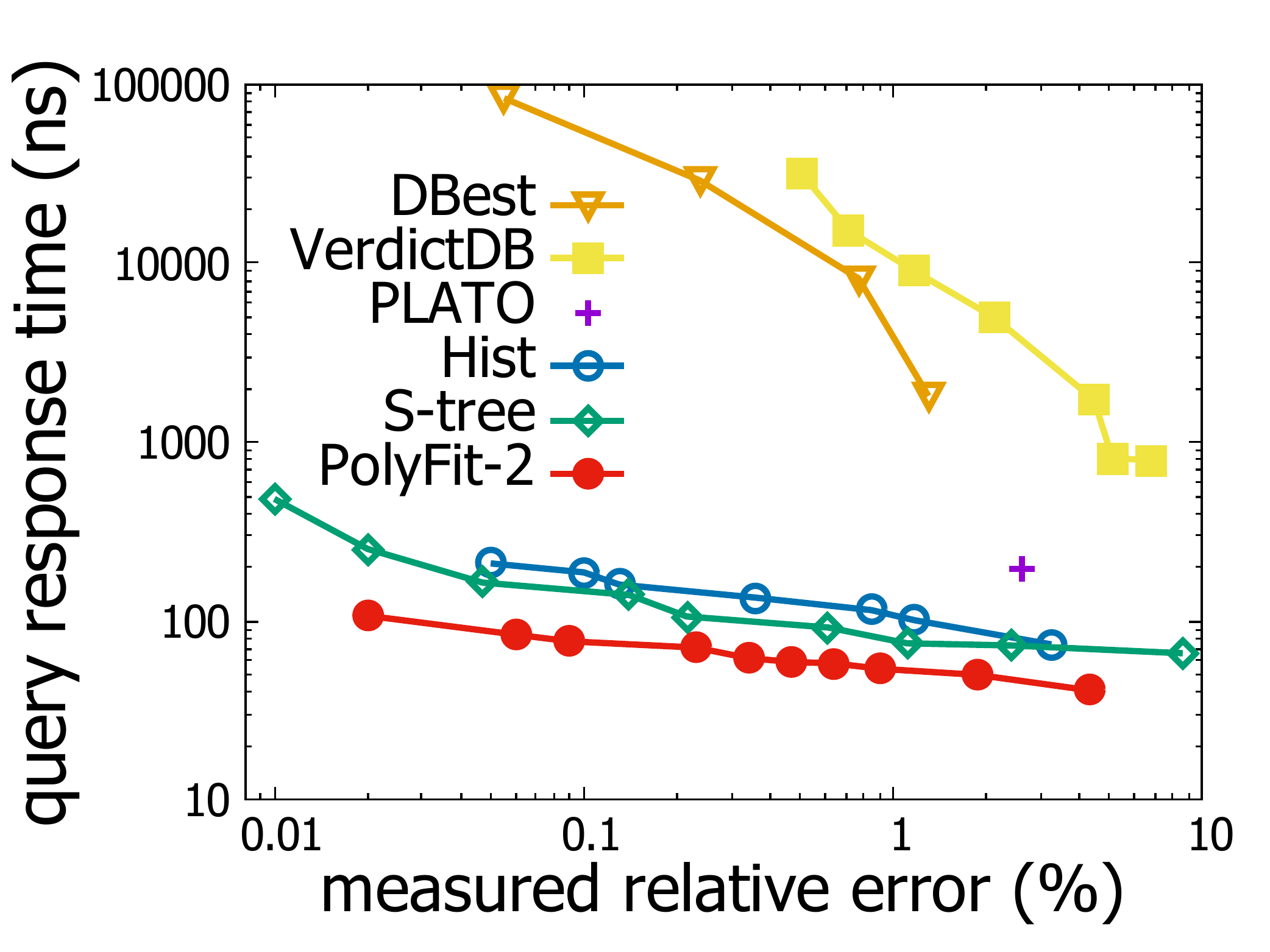} &
\hspace{-3mm}
\includegraphics[width=0.5\columnwidth]{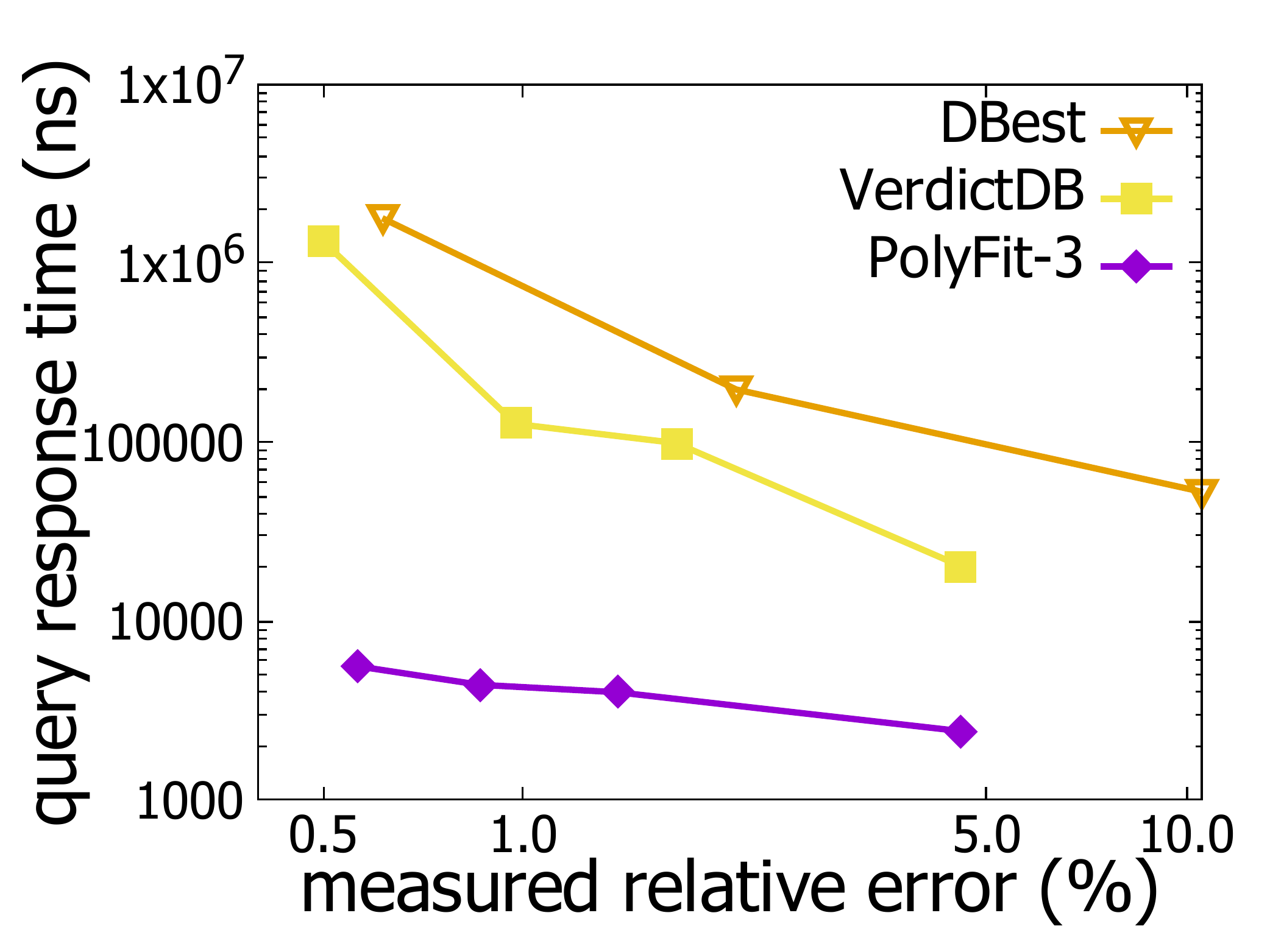} \\
(a) \verb"COUNT" query (single key) & (b) \verb"COUNT" query (two keys)
\end{tabular}
\vspace{-3.5mm}
\caption{Response time between \polyfit{} and the heuristic methods for \texttt{COUNT} query with single key and two keys in TWEET and OSM datasets, respectively}
\vspace{-3.5mm}
\label{fig:heuristics_error_time}
\end{figure}

\vspace{-3.0mm}
\subsection{Comparing the Construction Time of All Methods}
\vspace{-1.0mm}
\label{sec:construction_time_all_methods}
We proceed to investigate further how the construction times of all methods change across different dataset sizes. Here, we adopt the default degrees, i.e., $deg=2$ and $deg=3$, for the polynomial functions in the \texttt{COUNT} query with a single key and two keys, respectively. In Figure \ref{fig:construction_time_vary_dataset_size}, \polyfit{} consistently achieves faster construction time than Hist and DBest. Although \polyfit{} may not achieve the fastest construction time, compared with some methods (e.g., the aR-tree and the MRTree), \polyfit{} takes less than 150s and 2500s (with default $deg$) in the construction stage with 1 million (TWEET) and 30 million records (OSM), respectively, which are acceptable in practice where the datasets are static during data analytics tasks. 

\begin{figure}[!hbt]
	\centering
	\vspace{-3.5mm}
	\begin{tabular}{c c}
		\hspace{-8mm}
		\includegraphics[width=0.58\columnwidth]{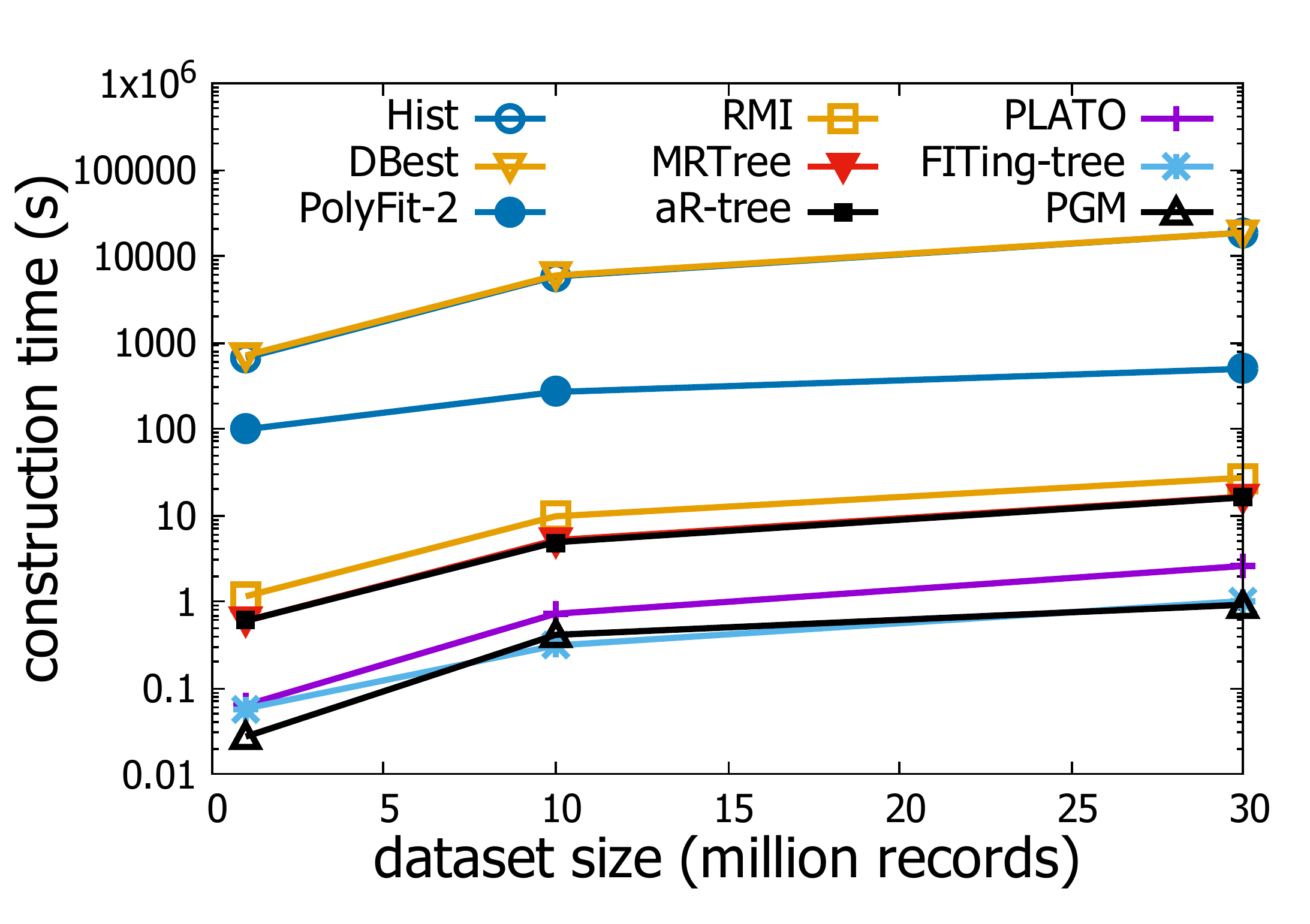} &
		\hspace{-5mm}
		\includegraphics[width=0.58\columnwidth]{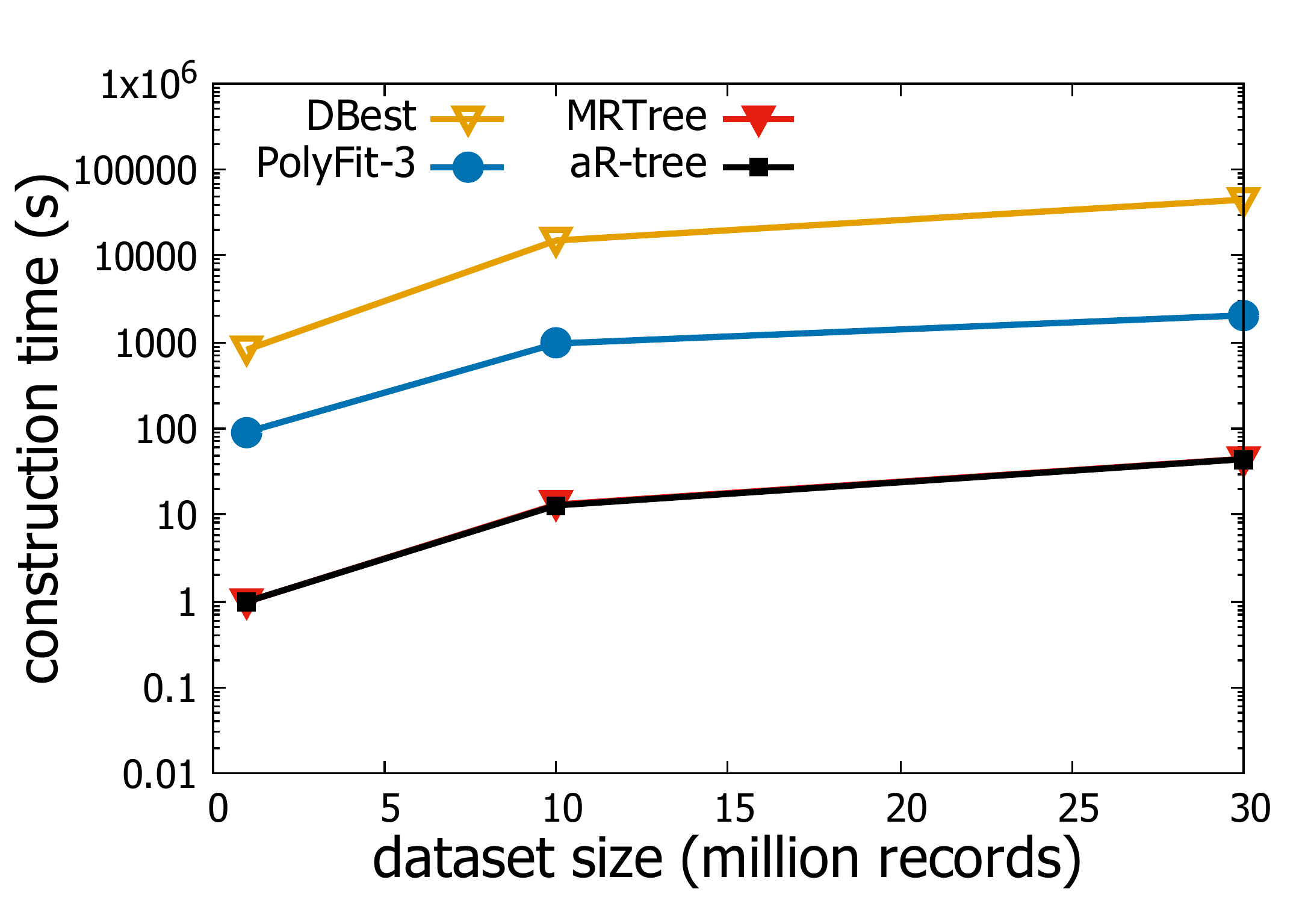} \\
		(a) \verb"COUNT" query (single key) & (b) \verb"COUNT" query (two keys)
	\end{tabular}
	\vspace{-2.5mm}
	\caption{Index construction time of methods for \texttt{COUNT} query with single key and two keys (using OSM dataset for both settings), varying the dataset size}
	\vspace{-2.5mm}
	\label{fig:construction_time_vary_dataset_size}
\end{figure}



%% file: conclusion.tex
\vspace{-5.0mm}
\section{Conclusion}
\label{sec:conclusion}
\vspace{-1.5mm}
In this paper, we study the range aggregate queries with two types of approximate guarantees, which are (1) absolute error guarantees (cf. Problem \ref{prob:abs_error} ($Q_{abs}$)) and (2) relative error guarantees (cf. Problem \ref{prob:rel_error} ($Q_{rel}$)). Unlike the existing methods, our work can efficiently support the most commonly used range aggregate queries (\verb"SUM", \verb"COUNT", \verb"MIN", \verb"MAX"), fulfill the error guarantees, and support the setting of two keys.

In order to improve the efficiency of computing these queries, we utilize several polynomial functions to fit the data points and then build the compact index structure \polyfit{} on top of these polynomial functions. An experimental study shows that \polyfit{} can achieve significant speedups compared with existing learned-index methods and other traditional exact/ approximate methods for different query types. In particular, we can achieve at most 5$\mu$s query response time in a dataset with 30 million records, which cannot be achieved by the state-of-the-art methods. 

In the future, we plan to further develop advanced techniques to improve the efficiency of constructing \polyfit{}, in order to handle updates of records in large-scale datasets. In addition, we aim to extend our methods to support other fundamental analytics operations, including standard deviation, median, etc. Moreover, we plan to investigate how to utilize the idea of \polyfit{} to further improve the efficiency of other types of statistics and machine learning models, e.g., kernel density estimation \cite{TRM20,TML19}, and support vector machines \cite{TLRMS20,TML19}.

%% file: appendix.tex
\appendix
\section{Appendix}
\label{sec:appendix}

\subsection{Approximate Range Aggregate Queries via Learn Index Methods}
The learned index methods, including RMI~\cite{kraska2018case}, FITing-tree~\cite{fiting2019}, and PGM~\cite{PG20}, are originally designed for range and point queries. In order to support range aggregate queries, e.g., \verb"SUM" and \verb"MAX", with both absolute error and relative error guarantees (cf. Problem \ref{prob:abs_error} and Problem \ref{prob:rel_error}), we follow the same mechanism of these index structures to fit the curve of either $CF_{sum}(k)$ or $DF_{max}(k)$ cf. Equation \ref{eq:Fk}. Instead of finding the exact result for either range query or point query, we utilize our querying methods (i.e., Lemma \ref{lem:count_abs_error} to \ref{lem:max_rel_error}) to solve both Problem \ref{prob:abs_error} and Problem \ref{prob:rel_error}.

\subsection{Tuning the RMI}
RMI \cite{kraska2018case} is a flexible learned index structure, which contains many parameters for tuning this index, including: (1) types of machine learning models, (2) the number of stages in RMI, (3) the number of models for each stage in RMI. Here, we adopt the TWEET dataset (cf. Table \ref{tab:datasets}) to tune the parameters, so as to obtain the best performance for RMI.

\subsubsection{{\bf Model Selection}}
In \cite{kraska2018case}, they adopt the neural network (NN) with at most two hidden layers and linear regression (LR) for testing the performance. Table \ref{tab:model_comparison} summarizes the response time and measured relative error, using single model to fit $CF_{sum}(k)$ in TWEET dataset for approximate \verb"SUM" query. Here, we use 1:X:Y:1 to represent the NN architecture with two hidden layers, i.e., X and Y neurons in the first and second hidden layers respectively, where the first and last one denote the input and output respectively. Similarly, we also use 1:X:1 to represent one hidden layer of the NN architecture.

Even though NN model can generally provide accurate fitting to the curve ($CF_{sum}(k)$ in this experiment), the response time can be much larger, compared with linear regression model. As an example, once we choose the shallow NN architecture 1:8:1, the response time can achieve more than 100ns, which can be worse than the performance of FITing-tree (cf. Table \ref{tab:all_methods_error_guarantee}). Due to the inefficiency issue of highly non-linear NN model, we choose LR model for RMI.


\begin{table}[htbp]
	\centering
	\caption{Comparison of different machine learning models}
	\label{tab:model_comparison}
	\begin{tabular}{|c|c|c|c|} \hline
		Model & NN & Prediction time  & Measured \\
		 & architecture & (ns) & relative error (\%) \\ \hline
		 LR & n.a. & 20 & 38 \\ \hline
		NN & 1:4:1 & 119 & 24.1 \\ \hline
		NN & 1:8:1 & 189 & 25.3 \\ \hline
		NN & 1:16:1 & 275 & 25.3 \\ \hline
		NN & 1:4:4:1 & 152 & 24.8 \\ \hline
		NN & 1:8:8:1 & 347 & 21.4 \\ \hline
		NN & 1:16:16:1 & 972 & 23.3 \\ \hline
	\end{tabular}
\end{table}

\subsubsection{{\bf Tuning RMI Structure}}
\label{sec:staged_models}
RMI utilize multiple models (e.g., LR) to obtain the approximate searching position (cf. Figure \ref{fig:staged_models}). However, due to the large degree of flexibility for RMI, including the number of stages and the number of models for each stage in RMI, we only test some of the combinations for choosing the structure of models, i.e., RMI structure.

\begin{figure}
\centering
\includegraphics[scale=0.19]{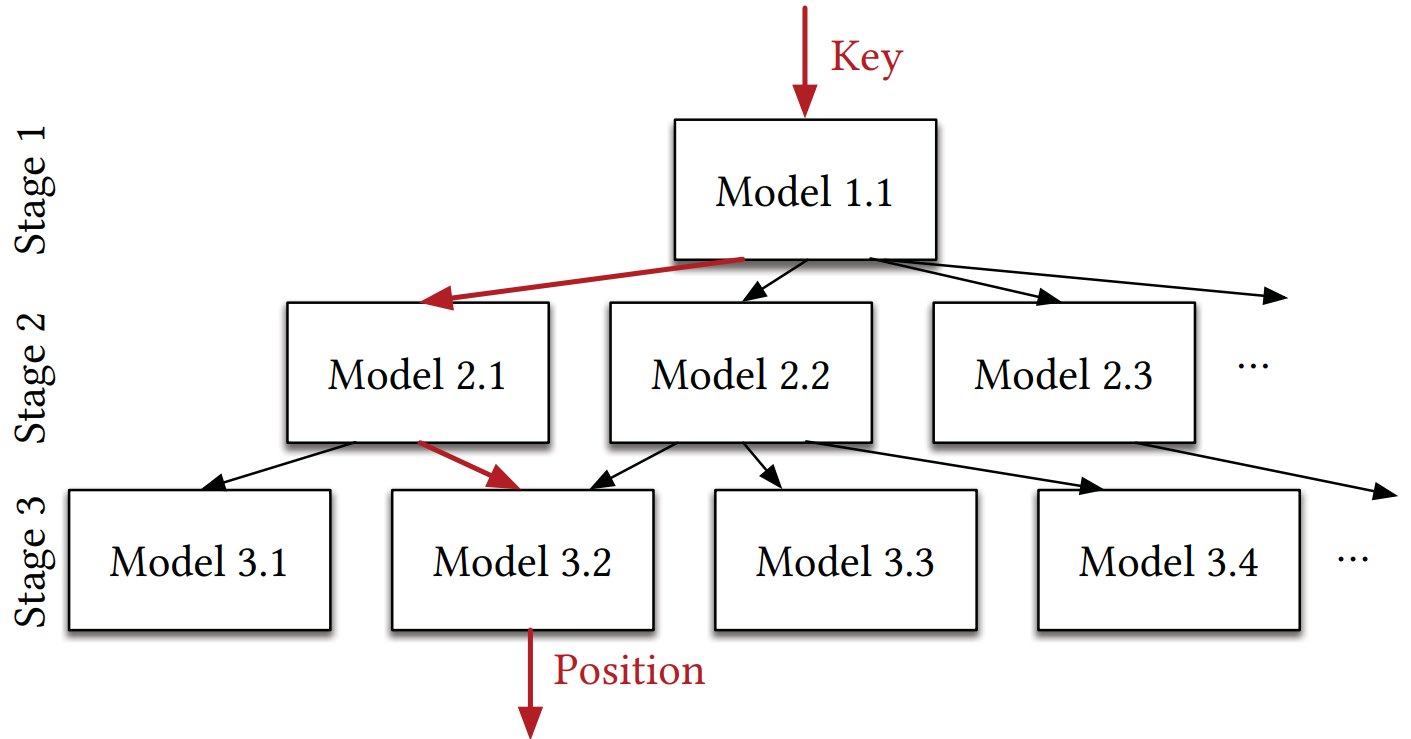}
\caption{RMI structure (Cropped from \cite{kraska2018case})}
\label{fig:staged_models}
\end{figure}

Theoretically, RMI can provide better performance with the large number of models, in which RMI can consume more memory resources. As such, we restrict the number of models in the leaf level (stage 3 in Figure \ref{fig:staged_models} as an example) of RMI to be approximately the same as PolyFit for the sake of fairness. Here, we use 1$\rightarrow$ 10 $\rightarrow$ X to denote the three-stage RMI structure with one model in stage 1, ten models in stage 2 and X models in stage 3. Similarly, we also use 1$\rightarrow$ 10 $\rightarrow$ 100 $\rightarrow$ Y to denote the four-stage RMI structure. Here, we vary X from 100 to 1000 and vary Y from 100 to 1000, by increasing X/Y 100 each time. In our experiment, we find that the RMI structure 1$\rightarrow$ 10 $\rightarrow$ 100 $\rightarrow$ 1000 can normally provide the smallest query response time, compared with other RMI structures. As such, we choose the RMI structure with 1$\rightarrow$ 10 $\rightarrow$ 100 $\rightarrow$ 1000 and LR for each model in our experiments (cf. Section \ref{sec:exp}).

\subsection{The Case of Repeated Keys}
In this scenario, we assume that the records with repeated keys in $\mathcal{D}$
are arranged in the ascending order on measure, i.e.,  \\
$(k_i, m_i), (k_{i+1}, m_{i+1}),\cdots,(k_{i+x}, m_{i+x})$, where $k_i=k_{i+1}=...=k_{i+x}$.
We propose to pre-process the dataset $\mathcal{D}$ as follows.

For \verb"SUM", \verb"MIN", or \verb"MAX" queries,
we propose to replace the repeated-key records by a single pair $(k_i, x)$,
where $\mathcal{G}$ is an aggregate function
and $x=\mathcal{G}(\{m_i, m_{i+1}, \cdots, m_{i+x}\})$.

For \verb"COUNT" queries, we replace the repeated-key records
by a single pair $(k_i, x)$, where  $x=CF_{count}(k_i) - CF_{count}(k_{i-1})$.
Then, during query evaluation, we execute a \verb"SUM" query instead of a \verb"COUNT" query.


\subsection{The Case of Negative Measure Values}
Our problem definitions in Section~\ref{sec:background} are applicable to negative measure values.
Nevertheless, the error conditions for query evaluation need to be examined and revised 
in order to preserve correctness.

For the absolute error guarantee, the error conditions 
(in Lemmas~\ref{lem:count_abs_error}, \ref{lem:max_abs_error} and \ref{lem:sum_abs_error_two_D})
are directly applicable to negative measure values.

In contrast, for the relative error guarantee,
we need to revise the error conditions in our lemmas with respect to negative measure values.

For example, we can replace Lemma~\ref{lem:count_rel_error} by the following:
\begin{lemma}
If $\tilde{A}_{sum} \geq 2\delta(1+\frac{1}{\varepsilon_{rel}})$ or $\tilde{A}_{sum} \leq -2\delta(1+\frac{1}{\varepsilon_{rel}})$, then $\tilde{A}_{sum}$ satisfies the relative error guarantee with respect to $\varepsilon_{rel}$.
\end{lemma}

Similarly, we can replace Lemma~\ref{lem:max_rel_error} by the following:
\begin{lemma}
If $\tilde{A}_{max}\geq \delta(1+\frac{1}{\varepsilon_{rel}})$ or If $\tilde{A}_{max}\leq -\delta(1+\frac{1}{\varepsilon_{rel}})$, then $\tilde{A}_{max}$ satisfies the relative error guarantee with respect to $\varepsilon_{rel}$.
\end{lemma}

Finally, we can replace Lemma~\ref{lem:sum_rel_error_two_D} by the following:
\begin{lemma}
If $\tilde{A}_{count}\geq 4\delta(1+\frac{1}{\varepsilon_{rel}})$ or $\tilde{A}_{count}\leq -4\delta(1+\frac{1}{\varepsilon_{rel}})$ then $\tilde{A}_{count}$ satisfies the relative error guarantee $\varepsilon_{rel}$.
\end{lemma}




\subsection{The Case of Multiple Keys}
We now discuss how to extend the techniques in Section~\ref{sec:extend_2d} to support queries with multiple ($d>2$) keys.

It is straightforward to extend the key-cumulative surface (in Figure \ref{fig:example2D}b)
for the $d$-dimensional space.

According to Ho et al.~\cite{ho1997range}, the \verb"COUNT" of any rectangular region
can be expressed as the sum (and the difference) of $2^d$ precomputed terms.
This idea enables us to compute the result in $O(2^d)$ time.
However, the time complexity increases rapidly when $d$ increases.

Furthermore, if we apply the quad-tree based partitioning approach (in Figure~\ref{fig:quad_tree})
to fulfill the error guarantee $\delta$,
the required number of partitions (and models) grows exponentially with the increase of $d$.
The reason is that, at a high dimensionality, the key-cumulative surface becomes more complicated,
rendering it hard to approximate it well.



\subsection{The Case of Parallel Construction}
The index construction process can be accelerated by parallel computation.

First, we consider the single key scenario in Section~\ref{sec:index_construction}. 
We may apply a fast heuristics method (e.g., equi-width partitioning) 
to divide the key domain into $m$ intervals: $I_1, I_2, \cdots, I_m$.
Then, we create $m$ threads, where the $j$-th thread is used to run
Algorithm~\ref{alg:GS} on the $j$-th interval $I_j$ only.
Since these $m$ intervals are disjoint, it is feasible to run these $m$ threads in parallel.
The speedup of this method is at most $m$ times when compared to the single-thread execution.
Nevertheless, if the workload of these $m$ threads are not balanced, the speedup may become much lower than $m$.
In future, we will examine a fast heuristics method
that can balance the workload of threads.


Observe that the above idea can also be extended to the two keys scenario in Section~\ref{sec:extend_2d}.
The only difference is that we use a heuristic method to partitioning the 2-dimensional domain into $m$ disjoint regions.

%
%